\documentclass[%
 reprint,
 amsmath,amssymb,
 aps,
]{revtex4-1}

\usepackage{graphicx}
\usepackage{dcolumn}
\usepackage{bm}
\usepackage{wrapfig}
\usepackage{color, mathrsfs}
\usepackage[usenames,dvipsnames,svgnames,table]{xcolor}
\usepackage{simplewick, amssymb, amsmath}
\usepackage[ruled,vlined]{algorithm2e}

\newtheorem{proof}{Proof}
\newtheorem{claim}{Claim}

\newtheorem{lemma}{Lemma}

\newcommand*{\bs}{\boldsymbol}

\newcommand*{\bl}{{\bs l}}
\newcommand*{\bk}{{\bs k}}
\newcommand*{\bx}{{\bs x}}

\newcommand*{\gw}{ g^\text{\tiny $X\!Y$} }
\newcommand*{\aw}{A^\text{\tiny $X\!Y$} }
\newcommand*{\xy}{\text{\tiny $X\!Y$} }
\newcommand*{\xandy}[2]{\text{\tiny $#1\!#2$} }

\def\Ver{1}
\def\LongVer{0}
\begin{document}


\title{Decomposing CMB lensing power with simulation}

\author{Ethan Anderes}
\thanks{Supported by NSF grant 1007480.}
 \email{anderes@stat.ucdavis.edu}
\affiliation{%
Statistics Department\\ University of California, Davis CA 95616}%

%
%


\begin{abstract}
The reconstruction of the CMB lensing potential is based on a Taylor expansion of lensing effects which is known to have poor convergence properties. For lensing of temperature fluctuations, an understanding of the higher order terms in this expansion which is accurate enough for current experimental sensitivity levels has been developed in Hanson et. al. (2010), as well as a slightly modified Okamoto and Hu quadratic estimator which incorporates lensed rather than unlensed spectra into the estimator weights to mitigate the effect of higher order terms. We extend these results in several ways:
(1) We generalize this analysis to the full set of quadratic temperature/polarization lensing estimators,
(2) We study the effect of higher order terms for more futuristic experimental noise levels,
(3) We show that the ability of the modified quadratic estimator to mitigate the effect of higher order terms relies on a delicate cancellation which occurs only when the true lensed spectra are known. We investigate the sensitivity of this cancellation to uncertainties in or knowledge of these spectra.
We find that higher order terms in the Taylor expansion can impact projected error bars at experimental sensitivities similar to those found in future ACTpol/SPTpol experiments.
\end{abstract}

\maketitle


\section{Introduction}

 Over the past year, data from two ground based telescopes, ACT and SPT, have resulted in the first direct measurement of the weak lensing power spectrum solely from CMB  measurements \cite{Das:2011fk,vanE}. In the coming years, the data  from  {\it Planck} and  upcoming  experiments ACTpol and SPTpol will begin probing this lensing at much greater resolution. The state-of-the-art estimator of weak lensing, the quadratic estimator developed by Hu and Okomoto  \cite{Hu2001b, HuOka2002}, works in part through a delicate cancelation of terms in a Taylor expansion of the lensing effect on the CMB. In this paper we present a simulation based approach for exploring the nature of this cancelation for both the CMB intensity and the polarization fields. In particular, we study a  slightly modified quadratic estimator which: incorporates lensed rather than unlensed spectra into the estimator weights to mitigate the effect of higher order terms; and uses the observed lensed CMB fields to correct for the, so called, $N^{(0)}_\bl$ bias.

 The simulation methodology presented here  allows a stochastic exploration of the higher order bias terms of the quadratic estimate and can be used to reduce the computational load associated with iterative de-biasing algorithms for the quadratic estimate.  In this paper, we use our simulation methodology  to present a detailed study of the, so called, $N^{(1)}_\bl$ and $N^{(2)}_\bl$ bias for the full set of quadratic temperature/polarization lensing estimators. The $N^{(1)}_\bl$ bias was first explored  for the standard flat sky quadratic estimate in Kesden et al. \cite{Kesden:2003uq}. For full sky CMB temperature maps, Hanson et al. \cite{Hanson:2011fk}   developed an approximation to the higher order bias terms, including $N^{(2)}_\bl$, which is accurate enough for current experimental sensitivity levels as well as for the slightly modified quadratic estimator which incorporates lensed rather than unlensed spectra into the estimator weights. We generalize this analysis to the full set of modified quadratic temperature/polarization lensing estimators and demonstrate that, indeed, the lensed spectra weights mitigate the combined higher order bias. However, this  mitigation is obtained only by an increase in the magnitude  of both $N^{(1)}_\bl$ and $N^{(2)}_\bl$ to the extent that they nearly cancel. We explore the extent with which this cancelation is sensitive to fiducial uncertainty in the way the lensed spectra weights are computed. We find that, under experimental conditions similar to those in future ACTpol and SPTpol  experiments, the $EB$ quadratic estimator is not sensitive to low $l$ fiducial uncertainty whereas the $EE$ and $TE$ are sensitive to the point of degrading inferential power.

 The remainder of the paper is organized as follows. In sections \ref{tqe}, \ref{tsd} and \ref{ecpp} we give an overview of the modified quadratic estimate, derive the spectral density of the quadratic estimate in terms of higher order bias terms, and discuss the estimation of the spectral density of the lensing potential. In Section \ref{fmca} we present two simulation based methods for  estimating the higher order bias terms. The first simulation method works exclusively for estimating $N^{(1)}_\bl$, and is mainly used to validate the second algorithm  which can produce all higher order terms $N^{(j)}_\bl$ for $j\geq 1$.
In Section \ref{sim} we use these methods to study the  higher order terms for experimental noise levels similar to those found in future ACTpol and SPTpol experiments.  
The paper concludes with the Appendix which gives fast Fourier transform  (FFT) algorithms for computing the modified quadratic estimate. These algorithms extend the FFT techniques developed in  \cite{Hu2001b}  to the computation of all quadratic normalization constants and provides fast (non-stochastic) algorithms which extend the simulation techniques found in  \cite{Das:2011fk} for computing the `Gaussian bias' from the lensed CMB four-point function for all temperature/polarization quadratic estimators.

\section{The quadratic estimator}
\label{tqe}

  The effect of weak lensing is to simply remap the CMB temperature  $T(\bs x)$ and 
  Stokes polarization fields $Q(\bs x)$ and $U(\bs x)$ for a flat sky coordinate system $\bs x\in \Bbb R^2$.   Up to leading order, the remapping  displacements are given by $\nabla \phi(\bs x)$, where $\phi(\bs x)$ denotes a lensing potential and is the planar projection of a three dimensional gravitational potential (see \cite{DodelsonBook}). Therefore, for any CMB field  $X\in\{T,Q,U \}$ the corresponding lensed field can be written $X(\bs x+\nabla \phi(\bs x))$. For the remainder of the paper we let 
  \[\tilde X(\bs x)\equiv X(\bs x+\nabla \phi(\bs x)) + N^\text{\tiny $X$}(\bs x)\] 
   denote the corresponding lensed CMB field with additive independent experimental noise given by $N^\text{\tiny $X$}$ (which includes a beam deconvolution). Using this notation the corresponding lensed $E$ and $B$ modes are given by 
$\tilde E_\bl  \equiv -\cos(2\varphi_\bl) \tilde Q_\bl - \sin(2\varphi_\bl) \tilde U_\bl $ and
$\tilde B_\bl  \equiv \sin(2\varphi_\bl) \tilde Q_\bl - \cos(2\varphi_\bl) \tilde U_\bl$ where $X_\bl\equiv \int \frac{d^2\bs x}{2\pi}  e^{-i\bs x\cdot \bl} X(\bs x) $   (unitary angular frequency) and $\varphi_\bl$ denotes the phase angle of frequency $\bl$.

For any field $X, Y\in \{ T, Q, U, E, B\}$ the spectral density $C^{\xandy XY}_{\bl}$ is defined to satisfy $\langle X^{\vphantom{*}}_\bl Y_{\bl^\prime}^* \rangle =\delta_{\bl-\bl^\prime} C^{\xandy XY}_{\bl} $ where $\delta_\bl \equiv \int \frac{d^2\bs x}{(2\pi)^2}  e^{i\bs x\cdot \bl} $.
The angled brackets $\langle \cdot \rangle$  denote ensemble averaging (or expected value) over both the CMB fields and the large scale structure given by $\phi$. In addition, we let $\langle \cdot \rangle_{\xandy XY}$  denote  expected value with respect to the unlensed CMB  fields $ T, Q, U, E, B$  and $\langle\cdot \rangle_\phi$ denote expected value with respect to large scale structure given by $\phi$. Throughout this paper we stipulate $\phi$ is independent of $ T, Q, U, E, B$ which implies: $\langle \cdot \rangle =\langle  \langle \cdot \rangle_{\xandy XY} \rangle_\phi$.
 We let $\widetilde C^{\xandy XY}_{\bl}$  denote the lensed CMB spectral density  {\em without} experimental noise and  let $C^{\xandy XY}_{\bl,\text{\tiny exp}}\equiv C^{\xandy XY}_{\bl} +  C^{\text{$N^X\! N^Y$}}_\bl$  and  $\widetilde C^{\xandy XY}_{\bl,\text{\tiny exp}} \equiv \widetilde C^{\xandy XY}_{\bl} +  C^{\text{ $N^X\!N^Y$}}_\bl$  denote the corresponding unlensed and lensed spectral densities {\em with} the additional experimental noise.

The quadratic estimate, based on two lensed CMB fields $\tilde X$ and $\tilde Y$, is derived from the following two statements:
\begin{align}
\label{fundEQ1}
\bigl\langle \tilde X^{\phantom{*}}_{\bk + \bl} \tilde Y^*_{\bk}&\bigr\rangle_{\xandy XY}  \approx    \phi^{\phantom{*}}_\bl f^{\xandy XY}_{\bl,\bk},\quad\text{when $\bl\neq 0$;}\\
\label{fundEQ2}
\bigl\langle \tilde X^{\phantom{*}}_{\bk + \bl} \tilde Y^*_{\bk}&\bigr\rangle_{\phantom{\xandy XY}}  =  \delta_{\bl} \widetilde C^\xy_{\bk,\text{\tiny exp}}
\end{align}
 which hold for any $X,Y\in \{ T,E, B \}$ and where the coefficients $f^{\xandy XY}_{\bl,\bk}$ are given in the Appendix.
Equation (\ref{fundEQ1}) approximates the cross frequency correlation  (at separation lag $\bl$) induced by the  nonstationarity in $\tilde X$ (when regarding $\phi$ as a fixed nonrandom field). 
This is derived through a Taylor expansion of  the lensing operation for any  $X\in\{T,Q, U\}$ 
\begin{equation}
\label{ex1}
\tilde X(\bs x)= \delta^0X(\bs x)  + \delta^1 X(\bs x) +\delta^2X(\bs x) +\cdots
\end{equation}
where $\delta^0X(\bs x)\equiv X(\bs x) + N^\text{\tiny $X$}(\bs x)$,  $\delta^1X(\bs x) \equiv \nabla^aX(\bx) \nabla_a \phi (\bx)$, etc. 
When $X\in\{E,B\}$  one defines $
\delta^j E_\bl \equiv -\cos(2\varphi_\bl) \delta^jQ_\bl - \sin(2\varphi_\bl) \delta^j U_\bl $ and $
\delta^j B_\bl \equiv \sin(2\varphi_\bl) \delta^j Q_\bl - \cos(2\varphi_\bl) \delta^j U_\bl$. Then  by expanding $\tilde X^{\phantom{*}}_{\bk + \bl} \tilde Y^*_{\bk}$ with (\ref{ex1}), regrouping terms by the order of $\phi$, one obtains
$f^\xy_{\bl,\bk} \phi_\bl = \langle \delta^1X_{\bk+\bl} \delta^0 Y_\bk^* \rangle_{\xy} +  \langle \delta^0X_{\bk+\bl} \delta^1 Y_\bk^* \rangle_{\xy} $ which gives approximation (\ref{fundEQ1}).
Equation (\ref{fundEQ2}), on the other hand, is obtained by treating both the CMB and the large scale structure $\phi$ as random so that $\tilde X$, from this viewpoint, is isotropic (but non-Gaussian). 

Hu and Okamoto  \cite{Hu2001b,HuOka2002} used approximations (\ref{fundEQ1}) and (\ref{fundEQ2}) to construct  the optimal quadratic estimate of $\phi$ based on  $\tilde X$ and $\tilde Y$ as follows
\begin{equation}
\label{qe1}
\hat \phi_{\bs l}^\xy \equiv  \aw_\bl \int \frac{d^2\bs k}{2\pi} \gw_{\bs l,\bs k\,}  \tilde X_{\bs k+\bs l} \tilde Y^*_{\bs k} 
\end{equation}
where $\gw_{\bs l,\bs k\,}  \equiv 2\pi{f^\xy_{\bl, \bk}}[\widetilde C^\text{\tiny X\!X}_{\bk+\bl,\text{\tiny exp}} \widetilde C^\text{\tiny $Y\!Y$}_{\bk,\text{\tiny exp}} ]^{-1} $. 
The normalizing constant $\aw_\bl$ is determined through an unbiased constraint. In particular, using the fact that  $f_{\bl,\bk}^\text{\tiny $X\!Y$}$ is real we have that $ \langle \hat\phi^\text{\tiny $X\!Y$}_\bl \rangle_\xy = \phi_\bl A_\bl^\text{\tiny $X\!Y$} \int \frac{d^2 \bk}{(2\pi)^2}  |g^\text{\tiny $X\!Y$}_{\bl,\bk} |^2 \widetilde C^\text{\tiny X\!X}_{\bk+\bl,\text{\tiny exp}} \widetilde C^\text{\tiny $Y\!Y$}_{\bk,\text{\tiny exp}} $ by equation (\ref{fundEQ1}). 
Then requiring that $\bigl\langle \hat\phi^\text{\tiny $X\!Y$}_\bl\bigr\rangle = \phi_\bl$ determines $A^\text{\tiny $X\!Y$}_\bl$ as follows
\begin{equation}
\label{fA} 
A_\bl^\text{\tiny $X\!Y$}
= \Bigl[ \int \frac{d^2 \bk}{(2\pi)^2}  |\gw_{\bl,\bk}|^2 \widetilde C^\text{\tiny X\!X}_{\bk+\bl,\text{\tiny exp}} \widetilde C^\text{\tiny $Y\!Y$}_{\bk,\text{\tiny exp}} \Bigr]^{-1} .   
\end{equation}

\subsection{Lensed versus unlensed  weights}
\label{defB}
There is a small modification to the standard quadratic estimate $\hat \phi^\xy_\bl$ which can mitigate the low $l$ bias (arising from the $N^{(2)}_\bl$ term discussed in the next section) when using the observed power $|\hat \phi^{\xandy XY}_\bl|^2$ to estimate $C^{\phi\phi}_\bl$. This modified estimate, denoted $\tilde \phi^\xy_\bl$, is obtained   by replacing all occurrences of  unlensed spectra in  $g_{\bl, \bk}^\xy$ and  $A_\bl^\xy$  with the corresponding lensed spectra. In particular $\tilde \phi^\xy_\bl$ is defined as
\begin{equation}
\label{qe2}
\tilde  \phi_{\bs l}^\xy \equiv  \tilde A^\xy_\bl \int \frac{d^2\bs k}{2\pi} \tilde g^\xy_{\bs l,\bs k\,}  \tilde X_{\bs k+\bs l} \tilde Y^*_{\bs k} 
\end{equation}
where $\tilde A_\bl^\text{\tiny $X\!Y$}= \bigl[ \int \frac{d^2 \bk}{(2\pi)^2}  |\tilde g^\xy_{\bl,\bk}|^2 \widetilde C^\text{\tiny X\!X}_{\bk+\bl,\text{\tiny exp}} \widetilde C^\text{\tiny $Y\!Y$}_{\bk,\text{\tiny exp}}\bigr]^{-1} $ and $\tilde g^\xy_{\bs l,\bs k\,}$ 
is obtained from $g^\xy_{\bs l,\bs k\,}$ by  replacing every occurrence of $C^\xy_\bl$, $C^\xandy{X}{X}_\bl$ and $C^\xandy{X}{X}_\bl$ with the corresponding lensed spectra ${\widetilde C}^{\xy}_\bl$, $\widetilde C^\xandy{X}{X}_\bl$ and ${\widetilde C}^\xandy{Y}{Y}_\bl$. 
For example, when $X=Y=T$ one has
\begin{align*}
g^\text{\tiny $TT$}_{\bl, \bk} \equiv \frac{ \bl \cdot [(\bk + \bl) C^\text{\tiny $TT$}_{\bk+\bl}  - \bk  C^\text{\tiny $TT$}_{\bk}]}{\widetilde C^\text{\tiny T\!T}_{\bk+\bl,\text{\tiny exp}} \widetilde C^\text{\tiny $T\!T$}_{\bk,\text{\tiny exp}}}; \\
\tilde g^\text{\tiny $TT$}_{\bl, \bk} \equiv \frac{\bl \cdot [(\bk + \bl)  \widetilde C^\text{\tiny $TT$}_{\bk+\bl}  - \bk   \widetilde C^\text{\tiny $TT$}_{\bk}]}{\widetilde C^\text{\tiny T\!T}_{\bk+\bl,\text{\tiny exp}} \widetilde C^\text{\tiny $T\!T$}_{\bk,\text{\tiny exp}}}.
\end{align*}
Notice that the estimate  $\tilde \phi^\xy_\bk$ is not normalized to be unbiased. Indeed from equation (\ref{fundEQ1}) one has 
\begin{align*}
\langle \tilde \phi^\xy_\bl\rangle_\xy
& \approx \phi_\bl \underbrace{\Bigl[\tilde A_\bl^\xy  \int \frac{d^2 \bk}{(2\pi)^2}  \tilde g^\xy_{\bs l,\bs k\,}  g^\xy_{\bs l,\bs k\,} \widetilde C^\text{\tiny X\!X}_{\bk+\bl,\text{\tiny exp}} \widetilde C^\text{\tiny $Y\!Y$}_{\bk,\text{\tiny exp}}\Bigr]}_{\equiv B_\bl^\xy}.
\end{align*}

\section{The spectral density of the quadratic estimate}
\label{tsd}

In this section we derive the following all order decomposition of the spectral density of the quadratic estimate 
\begin{align}
\bigl\langle\hat\phi^\text{\tiny $X\!Y$}_\bl  \hat\phi^{\text{\tiny $X\!Y$}^*}_{\bl^\prime}  \bigr\rangle 
&= \delta_{\bl - \bl^\prime}[ C_\bl^{\phi\phi} + N^{(0)}_{\bl} + N^{(1)}_{\bl} + \cdots ]\label{expan}
\end{align}
which  will then be used, in a subsequent section, to derive the estimation bias for $C_\bl^{\phi\phi}$.
The first term $N^{(0)}_\bl$ is related to the disconnected terms of the lensed CMB four-point function, whereas the higher order terms $N_{\bl}^{(j)}$ for $j\geq 1$ are related to the connected terms of the four-point function segmented by the order of  $C_\bl^{\phi\phi}$. Most of this section focuses on  the quadratic estimate  $\hat\phi^\xy$ followed by a brief discussion of the corresponding decomposition for the modified quadratic estimate $\tilde \phi^\xy$.

 Our derivation of the  spectral density of $\hat\phi^\xy_\bl$  in the flat sky is similar  to the analysis of the full sky trispectrum  done in Hanson et al. \cite{Hanson:2011fk}.  One starts by relating $\langle \hat\phi^\text{\tiny $X\!Y$}_\bl  \hat\phi^{\text{\tiny $X\!Y$}^*}_{\bl^\prime} \rangle$ to the lensed CMB four-point function by  distributing the expected value as follows
 \begin{align}
\langle &\hat\phi^\text{\tiny $X\!Y$}_\bl  \hat\phi^{\text{\tiny $X\!Y$}^*}_{\bl^\prime} \rangle=A_\bl^\text{\tiny $X\!Y$}A_{\bl^\prime}^\text{\tiny $X\!Y$}\nonumber\\
&\times \int\frac{ d^2 \bk}{2\pi} \int \frac{d^2 \bk^\prime}{2\pi} g^\text{\tiny $X\!Y$}_{\bl,\bk\,}   g^\text{\tiny $X\!Y$}_{\bl^\prime,\bk^\prime}  \bigl\langle \widetilde X^{\phantom{*}}_{\bs k +\bs l}    \widetilde Y^*_{\bk}   \widetilde X^*_{\bs k^\prime +\bl^\prime} \widetilde Y_{\bk^\prime}\bigr\rangle.
\label{totalSum}
\end{align}
To decompose (\ref{totalSum}) one then expands the four-point product term in the above integrand by expanding   the lensed CMB  Taylor expansion  (\ref{ex1}) to obtain
\begin{align}
\bigl\langle \widetilde X^{\phantom{*}}_{\bs k +\bs l}    \widetilde Y^*_{\bk}   \widetilde X^*_{\bs k^\prime +\bs l^\prime} \widetilde Y_{\bk^\prime}\bigr\rangle
&= \underbrace{ \sum_{i,j,p,q}\bigl\langle  \delta^i X^{\phantom{*}}_{\bs k +\bs l}   \delta^j Y^*_{\bk}   \delta^p X^*_{\bs k^\prime +\bs l^\prime}  \delta^q Y_{\bk^\prime}\bigr\rangle}_{\shortstack{ \text{\small \it decomposes further into connected } \\ \text{\small  \it and disconnected terms}}}
\label{SSum}
\end{align}

\begin{figure*}[ht]
\includegraphics[height = 2.8in]{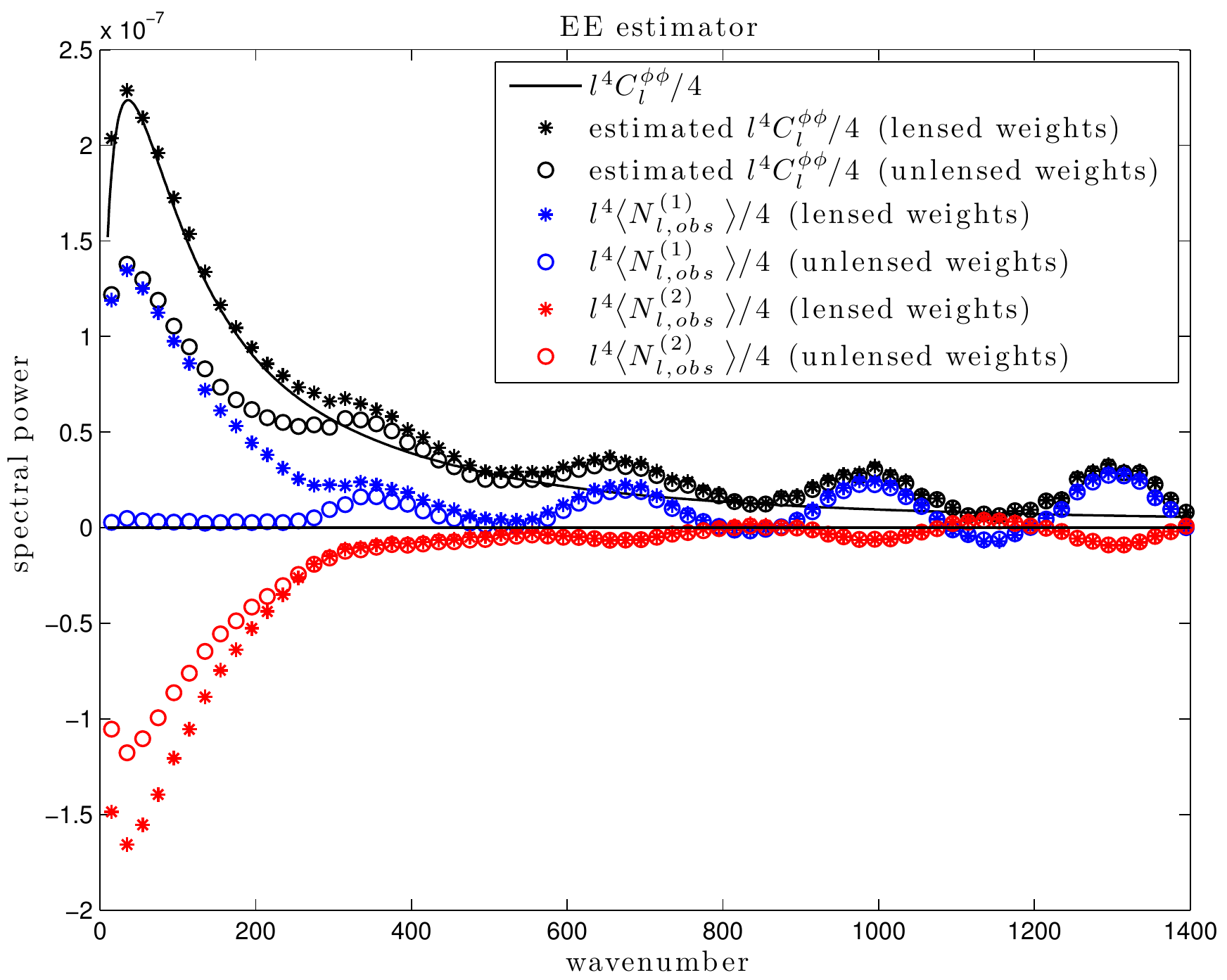} 
\includegraphics[height = 2.8in]{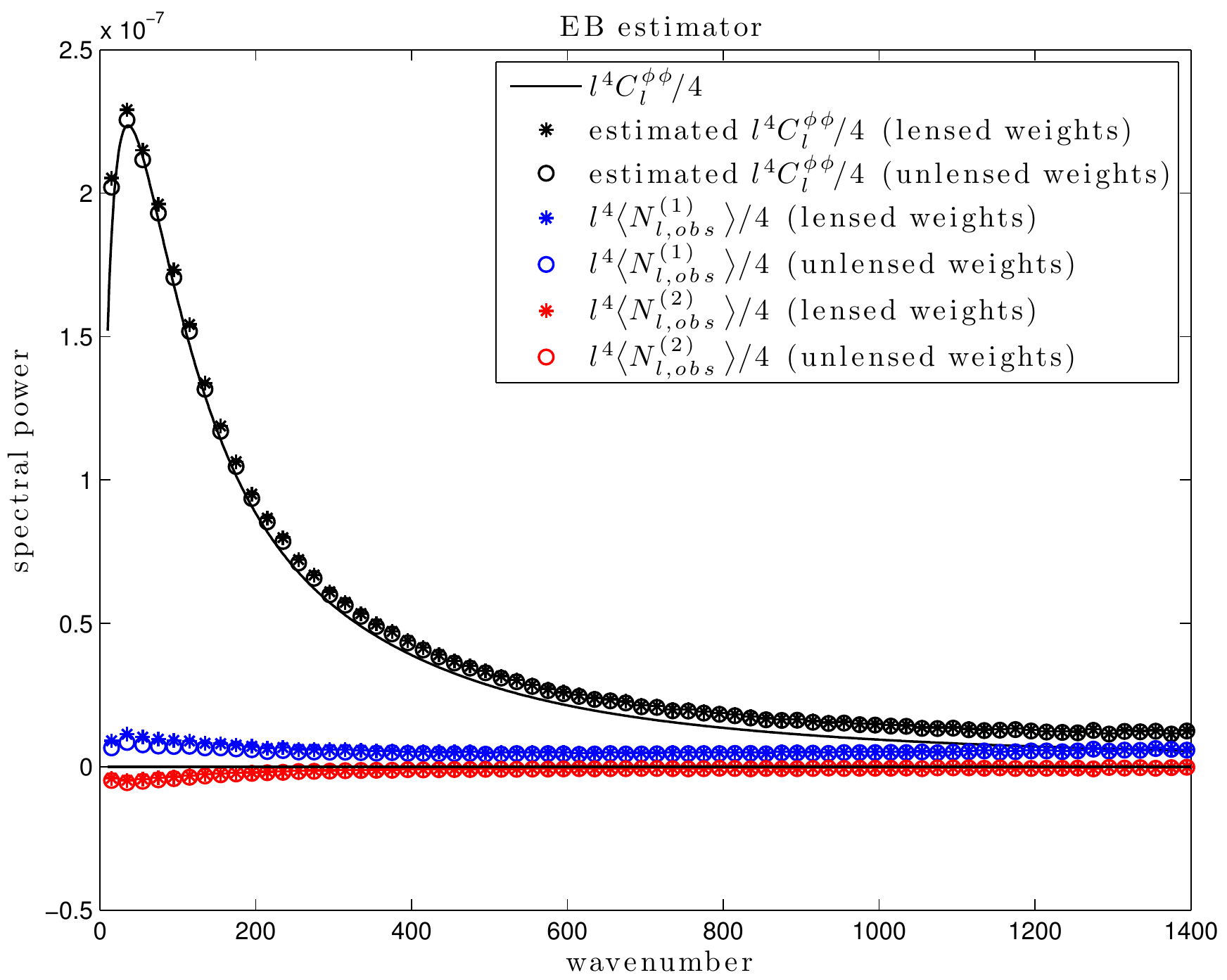} 
\caption{
These figures demonstrate the difference between EE and EB quadratic estimators when using lensed versus unlensed spectra into the estimator weights.
The main feature of the left hand plot  is the reduction of bias---comparing `{\large{$\ast$}}'  with `{\large {$\circ$}}'---in the EE estimator when the lensed spectra are used for the estimator weights.  Also by comparing `{\large\textcolor{red}{$\ast$}}'  with `{\large \textcolor{red}{$\circ$}}' and `{\large\textcolor{blue}{$\ast$}}'  with `{\large \textcolor{blue}{$\circ$}}'  it is clear that this reduction   is obtained  by an  increase in magnitude of both the so called $N^{(1)}$  and  $N^{(2)}$ bias  to extent that they nearly cancel.
The first order bias term  $l^4\bigl\langle N^{(1)}_{\bl, \text{\tiny obs}}\bigr\rangle/4$ is denoted `{\large\textcolor{blue}{$\ast$}}' for $\tilde\phi^\xy_\bl$ and  `{\large \textcolor{blue}{$\circ$}}' for $\hat\phi^\xy_\bl$. The second order bias term  $l^4\bigl\langle N^{(2)}_{\bl, \text{\tiny obs}}\bigr\rangle/4$ is  denoted  `{\large\textcolor{red}{$\ast$}}' for $\tilde\phi^\xy_\bl$ and  `{\large \textcolor{red}{$\circ$}}' for $\hat\phi^\xy_\bl$. Finally, the expected value of the spectral density estimates  $l^4\langle\delta_0^{-1}  |\hat \phi_\bl^\xy  |^2 -N^{(0)}_{\bl,\text{\tiny obs}}\rangle/4$ and $l^4\langle \delta_0^{-1} |\tilde  \phi_\bl^\xy  |^2 -N^{(0)}_{\bl,\text{\tiny obs}}\rangle/4$  are denoted by `{\large\textcolor{black}{$\circ$}}' and `{\large\textcolor{black}{$\ast$}}' respectively. 
See Section \ref{sim} for details.
 \label{fig1} }
\end{figure*}

\subsection{Disconnected terms}

After distributing the expected value in the right hand side of (\ref{SSum}) through the $j$-fold convolution which makes up $\delta^j X_\bl$ 
and subsequently applying Wicks theorem,  one can further decompose  $\bigl\langle  \delta^i X^{\phantom{*}}_{\bs k +\bs l}   \delta^j Y^*_{\bk}   \delta^p X^*_{\bs k^\prime +\bs l^\prime}  \delta^q Y_{\bk^\prime}\bigr\rangle$ into what are called connected and disconnected terms. 
The disconnected terms in the four-point product  are the terms which factor into cross-spectra of the fields  $\delta^i X$,  $\delta^j X$,  $\delta^p X$ and  $\delta^q X$. 
For example, if $(i,j,p,q)=(1,1,0,0)$ and $\phi$ is assumed independent of $X$ and $Y$ then
\begin{align*}
\bigl\langle  \delta^1 &X^{\phantom{*}}_{\bs k +\bs l}   \delta^1 Y^*_{\bk}   \delta^0 X^*_{\bs k^\prime +\bs l^\prime}  \delta^0 Y_{\bk^\prime}\bigr\rangle\\
&=
\bigl\langle 
\contraction{}{ \delta^1\!X^{\phantom{*}}_{\bs k +\bs l} } {} {\delta^1\! Y^*_{\bk}\,} 
\contraction{{ \delta^1\! X^{\phantom{*}}_{\bs k +\bs l} }  {\delta^1\! Y^*_{\bk}\,}}   {\delta^0\! X^*_{\bs k^\prime +\bs l^\prime}\,} {} {\delta^0\! Y_{\bk^\prime}\,} 
\bcontraction{}{ \delta^1\!X^{\phantom{*}}_{\bs k +\bs l} } {} {\delta^1\! Y^*_{\bk}\,} 
{ \delta^1\! X^{\phantom{*}}_{\bs k +\bs l} }  {\delta^1\! Y^*_{\bk}\,}   {\delta^0\! X^*_{\bs k^\prime +\bs l^\prime}\,}  {\delta^0\! Y_{\bk^\prime}\,}
\bigr\rangle\leftarrow \text{\small \it disconnected term}  \\
&\quad+ 
\bigl\langle 
\contraction[2ex]{} { \delta^1\! X^{\phantom{*}}_{\bs k +\bs l} } { {\delta^1\! Y^*_{\bk}\,}   {\delta^0\! X^*_{\bs k^\prime +\bs l^\prime}\,}}  {\delta^0\! Y_{\bk^\prime}\,}
\contraction{ \delta^1\! X^{\phantom{*}}_{\bs k +\bs l} }  {\delta^1\! Y^*_{\bk}\,} {}  {\delta^0\! X^*_{\bs k^\prime +\bs l^\prime}\,} 
\bcontraction{}{ \delta^1\!X^{\phantom{*}}_{\bs k +\bs l} } {} {\delta^1\! Y^*_{\bk}\,} 
{ \delta^1\! X^{\phantom{*}}_{\bs k +\bs l} }  {\delta^1\! Y^*_{\bk}\,}   {\delta^0\! X^*_{\bs k^\prime +\bs l^\prime}\,}  {\delta^0\! Y_{\bk^\prime}\,}
\bigr\rangle \leftarrow \text{\small \it connected term}\\
&\quad+ 
\bigl\langle 
\contraction{}{ \delta^1\! X^{\phantom{*}}_{\bs k +\bs l} }  {\delta^1\! Y^*_{\bk}\,}   {\delta^0\! X^*_{\bs k^\prime +\bs l^\prime}\,} 
\contraction[2ex]{ \delta^1\! X^{\phantom{*}}_{\bs k +\bs l} }{\delta^1\! Y^*_{\bk}\,}   {\delta^0\! X^*_{\bs k^\prime +\bs l^\prime}\,}  {\delta^0\! Y_{\bk^\prime}\,}
\bcontraction{}{ \delta^1\!X^{\phantom{*}}_{\bs k +\bs l} } {} {\delta^1\! Y^*_{\bk}\,} 
{ \delta^1\! X^{\phantom{*}}_{\bs k +\bs l} }  {\delta^1\! Y^*_{\bk}\,}   {\delta^0\! X^*_{\bs k^\prime +\bs l^\prime}\,}  {\delta^0\! Y_{\bk^\prime}\,}
\bigr\rangle\leftarrow \text{\small \it connected term} 
\end{align*}
where the top contraction symbols correspond to pairing the CMB fields and the bottom contraction symbols correspond to pairing the lensing potential in $\delta^1\!X$ and $\delta^1\!Y$. 
The contraction pairings on the above disconnected term results in a product of two spectra as follows
\begin{align*} 
&\bigl\langle 
\contraction{}{ \delta^1\!X^{\phantom{*}}_{\bs k +\bs l} } {} {\delta^1\! Y^*_{\bk}\,} 
\contraction{{ \delta^1\! X^{\phantom{*}}_{\bs k +\bs l} }  {\delta^1\! Y^*_{\bk}\,}}   {\delta^0\! X^*_{\bs k^\prime +\bs l^\prime}\,} {} {\delta^0\! Y_{\bk^\prime}\,} 
\bcontraction{}{ \delta^1\!X^{\phantom{*}}_{\bs k +\bs l} } {} {\delta^1\! Y^*_{\bk}\,} 
{ \delta^1\! X^{\phantom{*}}_{\bs k +\bs l} }  {\delta^1\! Y^*_{\bk}\,}   {\delta^0\! X^*_{\bs k^\prime +\bs l^\prime}\,}  {\delta^0\! Y_{\bk^\prime}\,}
\bigr\rangle \\
&\qquad\qquad\quad= \bigl\langle { \delta^1\! X^{\phantom{*}}_{\bs k +\bs l} }  {\delta^1\! Y^*_{\bk}\,}  \bigr\rangle  \bigl\langle {\delta^0\! X^*_{\bs k^\prime +\bs l^\prime}\,}  {\delta^0\! Y_{\bk^\prime}\,}\bigr\rangle.
\end{align*}
In a similar manner,  for general $(i,j,p,q)$, the disconnected terms can be grouped into the three types: one for each possible configuration of the of top contraction symbols. Then regrouping all disconnected terms in (\ref{SSum}), by top contraction type, results in the following three terms:
\begin{align} 
 &\text{\small \it disconnected terms in} \sum_{i,j,p,q}\bigl\langle \delta^i X^{\phantom{*}}_{\bs k +\bs l}    \delta^j Y^*_{\bk}   \delta^p X^*_{\bs k^\prime +\bs l^\prime}  \delta^q Y_{\bk^\prime}\bigr\rangle\nonumber \\
&\qquad= 
\bigl\langle 
\contraction[2ex]{} { \tilde  X^{\phantom{*}}_{\bs k +\bs l} } { {\tilde  Y^*_{\bk}\,}   {\tilde  X^*_{\bs k^\prime +\bs l^\prime}\,}}  {\tilde  Y_{\bk^\prime}\,}
\contraction{ \tilde  X^{\phantom{*}}_{\bs k +\bs l} }  {\tilde  Y^*_{\bk}\,} {}  {\tilde  X^*_{\bs k^\prime +\bs l^\prime}\,} 
{ \tilde  X^{\phantom{*}}_{\bs k +\bs l} }  {\tilde  Y^*_{\bk}\,}   {\tilde  X^*_{\bs k^\prime +\bs l^\prime}\,}  {\tilde  Y_{\bk^\prime}\,}
\bigr\rangle
+ 
\bigl\langle 
\contraction{}{ \tilde  X^{\phantom{*}}_{\bs k +\bs l} }  {\tilde  Y^*_{\bk}\,}   {\tilde  X^*_{\bs k^\prime +\bs l^\prime}\,} 
\contraction[2ex]{ \tilde  X^{\phantom{*}}_{\bs k +\bs l} }{\tilde  Y^*_{\bk}\,}   {\tilde  X^*_{\bs k^\prime +\bs l^\prime}\,}  {\tilde  Y_{\bk^\prime}\,}
{ \tilde  X^{\phantom{*}}_{\bs k +\bs l} }  {\tilde  Y^*_{\bk}\,}   {\tilde  X^*_{\bs k^\prime +\bs l^\prime}\,}  {\tilde  Y_{\bk^\prime}\,}
\bigr\rangle 
 \nonumber\\\nonumber
 &\qquad\qquad\qquad\qquad\qquad\qquad\quad\!\!\!+
 \underbrace{
\bigl\langle 
\contraction{}{ \tilde X^{\phantom{*}}_{\bs k +\bs l} } {} {\tilde Y^*_{\bk}\,} 
\contraction{{ \tilde  X^{\phantom{*}}_{\bs k +\bs l} }  {\tilde  Y^*_{\bk}\,}}   {\tilde  X^*_{\bs k^\prime +\bs l^\prime}\,} {} {\tilde  Y_{\bk^\prime}\,} 
{ \tilde X^{\phantom{*}}_{\bs k +\bs l} }  {\tilde  Y^*_{\bk}\,}   {\tilde  X^*_{\bs k^\prime +\bs l^\prime}\,}  {\tilde  Y_{\bk^\prime}\,}
\bigr\rangle }_{=0 \text{  when $\bl \neq 0$ or $\bl^\prime \neq 0$} }\\
 &\qquad= \widetilde  C^\text{\tiny X\!Y}_{\bk+\bl,\text{\tiny exp}}  \widetilde  C^\text{\tiny X\!Y}_{\bk,\text{\tiny exp}} \delta_{\bk+\bk^\prime+\bl^\prime}  \delta_{\bk+\bk^\prime+\bl}   \label{disSum} \\
 &\qquad\qquad\qquad\qquad +   \widetilde C^\text{\tiny X\!X}_{\bk+\bl,\text{\tiny exp}}  \widetilde C^\text{\tiny Y\!Y}_{\bk,\text{\tiny exp}} \delta_{\bk-\bk^\prime} \delta_{\bk+\bl -\bk^\prime-\bl^\prime}  \nonumber
\end{align}
where the last line is obtained by applying approximation (\ref{fundEQ2}).
Substituting the four-point product term in (\ref{totalSum}) with (\ref{disSum})  results in, what is typically called, the $N^{(0)}_\bl$ bias:
 \begin{align*}
 \text{\small \it disconnected terms in $\bigl\langle\hat\phi^\text{\tiny $X\!Y$}_\bl  \hat\phi^{\text{\tiny $X\!Y$}^*}_{\bl^\prime} \bigr\rangle$}= \delta_{\bl - \bl^\prime}N^{(0)}_{\bl}
 \end{align*}
 where 
 \begin{align*}
N^{(0)}_{\bl} \equiv [\aw_\bl]^2
 \int \frac{d^2\bk}{(2\pi)^2} \Bigl(  \gw_{\bs l,\bk\,}   \gw_{\bl, -\bk-\bl\,}   \widetilde C^\text{\tiny $X\!Y$}_{\bk+\bl,\text{\tiny exp}}   \widetilde C^\text{\tiny $X\!Y$}_{\bk,\text{\tiny exp}}&\\
  + |\gw_{\bs l,\bs k}|^2 \widetilde C^\text{\tiny $X\!X$}_{\bk+\bl,\text{\tiny exp}}  \widetilde C^\text{\tiny $Y\!Y$}_{\bk,\text{\tiny exp}}\Bigr)&.
\end{align*}

\subsection{Connected terms}

The connected terms decompose further into what we call the `first connected terms' and the `second connected terms'.
There are only four `first connected terms' and are defined as follows:
\begin{align*}
\text{\small\it{}first connected } &\text{\small\it{}terms in}\sum_{i,j,p,q} \bigl\langle  \delta^i X^{\phantom{*}}_{\bs k +\bs l}   \delta^j Y^*_{\bk}   \delta^p X^*_{\bs k^\prime +\bs l^\prime}  \delta^q Y_{\bk^\prime}\bigr\rangle\\
 &\equiv
 \bigl\langle
\contraction{}{ \delta^1X_{\bk+\bl}} {}{\delta^0 Y_\bk^* } 
\contraction{{ \delta^1X_{\bk+\bl}} {\delta^0 Y_\bk^* }} { \delta^1X^*_{\bk^\prime +\bl^\prime}} {}{\delta^0 Y_{\bk^\prime}}
\bcontraction{}{ \delta^1X_{\bk+\bl}} {\delta^0 Y_\bk^* } { \delta^1X^*_{\bk^\prime +\bl^\prime}} 
{ \delta^1X_{\bk+\bl}} {\delta^0 Y_\bk^* } { \delta^1X^*_{\bk^\prime +\bl^\prime}} {\delta^0 Y_{\bk^\prime}}
\bigr\rangle\\
&\qquad\qquad
+
 \bigl\langle
\contraction{}{\delta^1X_{\bk+\bl}}{} {\delta^0 Y_\bk^*} 
\contraction{{\delta^1X_{\bk+\bl}} {\delta^0 Y_\bk^*}} {\delta^0X^*_{\bk^\prime +\bl^\prime}} {}{\delta^1 Y_{\bk^\prime}}
\bcontraction{}{\delta^1X_{\bk+\bl}} {{\delta^0 Y_\bk^*} {\delta^0X^*_{\bk^\prime +\bl^\prime}}} {\delta^1 Y_{\bk^\prime}}
{\delta^1X_{\bk+\bl}} {\delta^0 Y_\bk^*} {\delta^0X^*_{\bk^\prime +\bl^\prime}} {\delta^1 Y_{\bk^\prime}}
\bigr\rangle
\\
&\qquad\qquad+
\bigl\langle
\contraction {}{\delta^0X_{\bk+\bl}}{} {\delta^1 Y_\bk^*}
\contraction {{\delta^0X_{\bk+\bl}} {\delta^1 Y_\bk^*}} { \delta^1X^*_{\bk^\prime +\bl^\prime}} {} {\delta^0 Y_{\bk^\prime}}
\bcontraction {\delta^0X_{\bk+\bl}} {\delta^1 Y_\bk^*}{} { \delta^1X^*_{\bk^\prime +\bl^\prime}} 
 {\delta^0X_{\bk+\bl}} {\delta^1 Y_\bk^*} { \delta^1X^*_{\bk^\prime +\bl^\prime}} {\delta^0 Y_{\bk^\prime}}
 \bigr\rangle\\
&\qquad\qquad+
\bigl \langle
\contraction{} {\delta^0X_{\bk+\bl}} {} {\delta^1 Y_\bk^*} 
\contraction {\delta^0X_{\bk+\bl} {\delta^1 Y_\bk^*} } {\delta^0X^*_{\bk^\prime +\bl^\prime}} {} {\delta^1 Y_{\bk^\prime} }
\bcontraction {\delta^0X_{\bk+\bl}}  {\delta^1 Y_\bk^*}  {\delta^0X^*_{\bk^\prime +\bl^\prime}} { \delta^1 Y_{\bk^\prime} }
{\delta^0X_{\bk+\bl}}  {\delta^1 Y_\bk^*}  {\delta^0X^*_{\bk^\prime +\bl^\prime}} { \delta^1 Y_{\bk^\prime} }
\bigr\rangle\\
&=\delta_{\bl - \bl^\prime} f^\xy_{\bl,\bk}f^\xy_{\bl^\prime,\bk^\prime} C_\bl^{\phi\phi}.
\end{align*}
Then, by substituting  the four-point product term in  (\ref{totalSum}) with the  first connected terms, one gets
\begin{align}
\text{\small \it first connected terms in}\, \bigl\langle \hat\phi^\text{\tiny $X\!Y$}_\bl  \hat\phi^{\text{\tiny $X\!Y$}^*}_{\bl^\prime}\bigr\rangle  =\delta_{\bl - \bl^\prime} C_\bl^{\phi\phi}.  \label{thisGivesC}
\end{align}
The remaining connected terms in  the four-point product of (\ref{totalSum}), called `second connected terms', are then re-grouping corresponding to the order of $\phi$. After noticing that any term of order $\phi^{2j+1}$ has expected value zero one obtains the following expansion 
\begin{align*}
\text{\small \it second connected terms in}\, &\bigl\langle \hat\phi^\text{\tiny $X\!Y$}_\bl  \hat\phi^{\text{\tiny $X\!Y$}^*}_{\bl^\prime}\bigr\rangle \\
& = \delta_{\bl - \bl^\prime}\bigl[ N^{(1)}_\bl + N^{(2)}_\bl + \cdots\bigr]
\end{align*}
where $N^{(j)}_\bl$ is of order  $\phi^{2j}$. Putting all disconnected and connected terms together 
gives the desired expansion (\ref{expan}).

For reasons which will become clear in the next section, we define $N^{(1)}_\bl$ in a slightly different, but equivalent, way that extends more naturally to the modified quadratic estimate. In particular, instead of defining $N^{(1)}_\bl$ as the 
 as  the total contribution of the second connected terms  in (\ref{totalSum}) of order  $\phi^2$, we define  $N^{(1)}_\bl$ simply as the total contribution of {\em all} connected terms of order $\phi^2$ minus $C_\bl^{\phi\phi}$.

\begin{figure*}[ht]
\includegraphics[height = 2.8in]{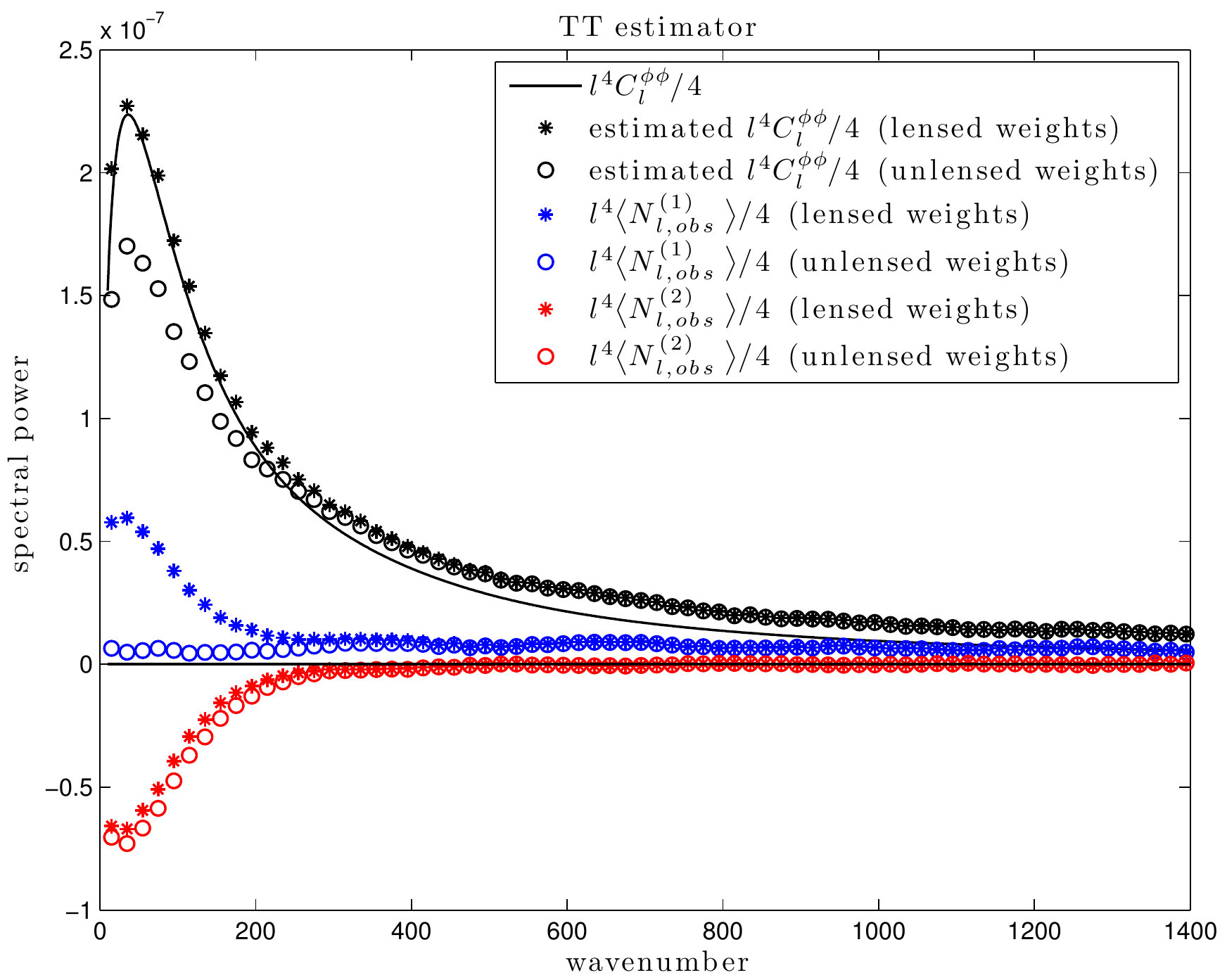} 
\includegraphics[height = 2.8in]{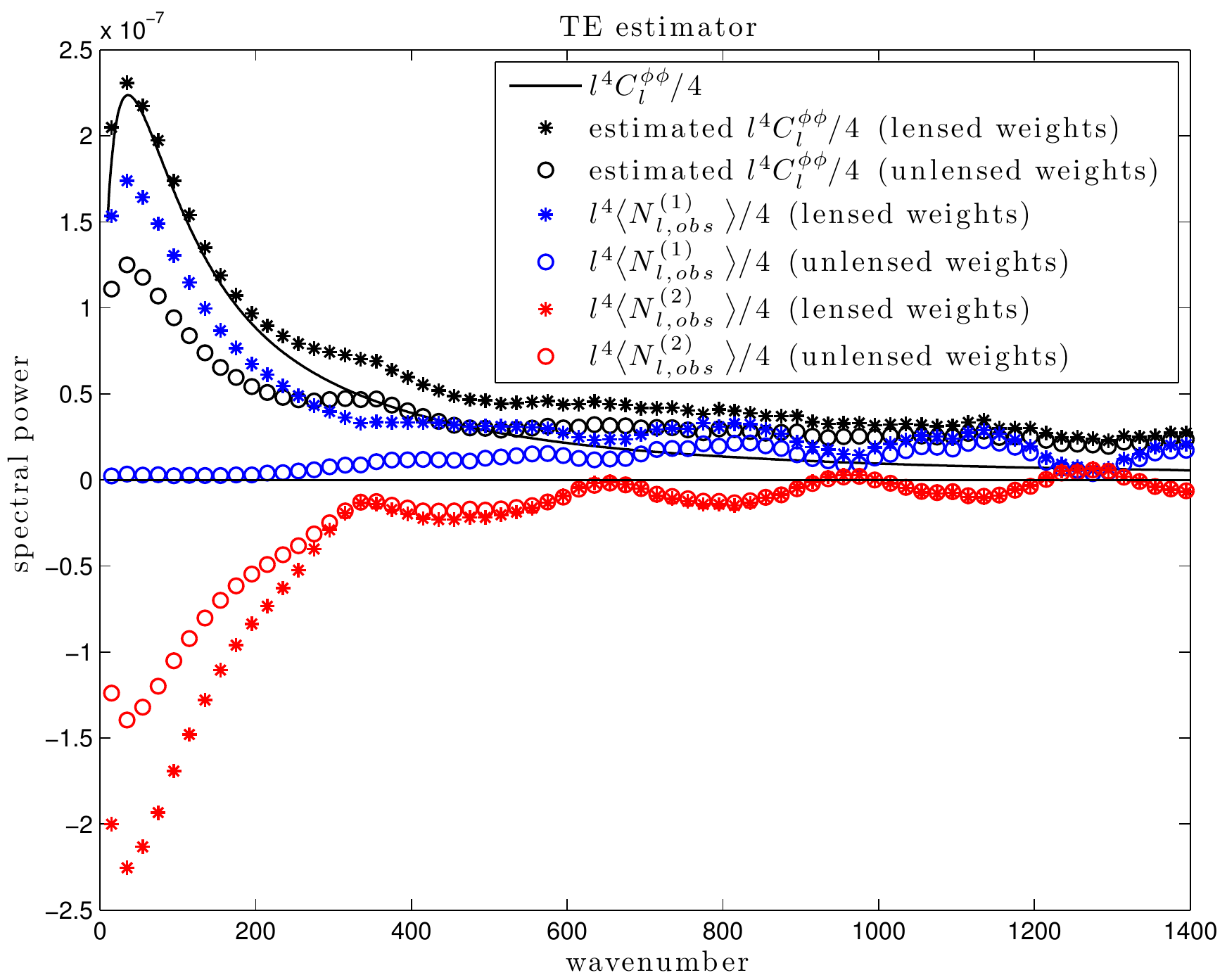} 
\caption{The first order bias term  $l^4\bigl\langle N^{(1)}_{\bl, \text{\tiny obs}}\bigr\rangle/4$ is denoted `{\large\textcolor{blue}{$\ast$}}' for $\tilde\phi^\xy_\bl$ and  `{\large \textcolor{blue}{$\circ$}}' for $\hat\phi^\xy_\bl$. The second order bias term  $l^4\bigl\langle N^{(2)}_{\bl, \text{\tiny obs}}\bigr\rangle/4$ is  denoted  `{\large\textcolor{red}{$\ast$}}' for $\tilde\phi^\xy_\bl$ and  `{\large \textcolor{red}{$\circ$}}' for $\hat\phi^\xy_\bl$. Finally, the expected value of the spectral density estimates  $l^4\langle\delta_0^{-1}  |\hat \phi_\bl^\xy  |^2 -N^{(0)}_{\bl,\text{\tiny obs}}\rangle/4$ and $l^4\langle \delta_0^{-1} |\tilde  \phi_\bl^\xy  |^2 -N^{(0)}_{\bl,\text{\tiny obs}}\rangle/4$  are denoted by `{\large\textcolor{black}{$\circ$}}' and `{\large\textcolor{black}{$\ast$}}' respectively. 
See Section \ref{sim} for details.
 \label{fig2} }
\end{figure*}

\subsection{The lensed weights}
For the modified quadratic estimator, $\tilde \phi^\xy_\bl$, one can derive the expansion (\ref{expan}) with a few minor adjustments.
In particular, the disconnected terms can be written
\begin{align*}
 \text{\small \it disconnected terms in $\bigl\langle\tilde\phi^\text{\tiny $X\!Y$}_\bl  \tilde\phi^{\text{\tiny $X\!Y$}^*}_{\bl^\prime} \bigr\rangle$}= \delta_{\bl - \bl^\prime}N^{(0)}_{\bl}
 \end{align*}
 where 
 \begin{align*}
N^{(0)}_{\bl} \equiv [\tilde A^\xy_\bl]^2
 \int \frac{d^2\bk}{(2\pi)^2} \Bigl(  \tilde g^\xy_{\bs l,\bk\,}  \tilde g^\xy_{\bl, -\bk-\bl\,}   \widetilde C^\text{\tiny $X\!Y$}_{\bk+\bl,\text{\tiny exp}}   \widetilde C^\text{\tiny $X\!Y$}_{\bk,\text{\tiny exp}}&\\
  + |\tilde g^\xy_{\bs l,\bs k}|^2 \widetilde C^\text{\tiny $X\!X$}_{\bk+\bl,\text{\tiny exp}}  \widetilde C^\text{\tiny $Y\!Y$}_{\bk,\text{\tiny exp}}\Bigr)&.
\end{align*}
The first connected terms are slightly different due to the  intentional bias in the modified quadratic estimate yielding
\begin{align*}
\text{\small \it first connected terms in}\, \bigl\langle \tilde\phi^\text{\tiny $X\!Y$}_\bl  \tilde\phi^{\text{\tiny $X\!Y$}^*}_{\bl^\prime}\bigr\rangle  =\delta_{\bl - \bl^\prime} [B^\xy_\bl]^2 C_\bl^{\phi\phi}.
\end{align*}
After these adjustments are made the remaining terms  $N^{(j)}_\bl$  are defined  exactly the same way as for  $\hat\phi^\xy_\bl$: when $j\geq 2$, $N_\bl^{(j)}$ is defined as the total contribution of the connected terms in (\ref{totalSum}) which are order $\phi^j$; and  $N^{(1)}_\bl$ as  the total contribution of the connected terms  in (\ref{totalSum}) which are order $\phi^2$ minus $C_\bl^{\phi\phi}$. This is consistent with the previous definition of $N^{(j)}_\bl$ for $\hat\phi_\bl^\xy$ and preserves the natural expansion
\begin{align}
\bigl\langle\tilde\phi^\text{\tiny $X\!Y$}_\bl  \tilde\phi^{\text{\tiny $X\!Y$}^*}_{\bl^\prime}  \bigr\rangle 
&= \delta_{\bl - \bl^\prime}[ C_\bl^{\phi\phi} + N^{(0)}_{\bl} + N^{(1)}_{\bl} + \cdots ].\label{expan2}
\end{align}

\section{Estimation of $C_\bl^{\phi\phi}$}
\label{ecpp}

In this section  we show how the expansions (\ref{expan}) and (\ref{expan2}) are used to derive and analyze estimates of  $C_\bl^{\phi\phi}$.
All the estimates presented in this section can also be additionally radially averaged, with inverse variance weights,  to reduce estimation variability.
We derive the following results using the estimate   $\hat \phi_\bl^\xy$ and simply remark that the results can be similarly derived for the modified quadratic estimate $\tilde \phi_\bl^\xy$.

For current experimental conditions the first term $N^{(0)}_\bl$  dominates the sum  $\sum_{j=0}^\infty N^{(j)}_{\bl} $. Therefore a natural bias corrected estimate  of $C_\bl^{\phi\phi}$  is given by $\delta_0^{-1} |\hat \phi_\bl^\xy  |^2 -N^{(0)}_{\bl}$.  
Notice, however,  that to compute $N^{(0)}_{\bl}$ one needs a model for the  lensed spectra with experimental noise: $\widetilde C^{\xandy XX}_{\bl,\text{\tiny exp}}$, $\widetilde C^{\xandy YY}_{\bl,\text{\tiny exp}}$ and $\widetilde C^\xy_{\bl,\text{\tiny exp}}$. 
This normally requires knowledge of the very quantity we are estimating: 
$C_\bl^{\phi\phi}$. One way to circumvent this difficulty is to replace  the experimental lensed spectrums with estimates from the observations $\tilde X$ and $\tilde Y$. This results 
 in  the {\em observed} $N^{(0)}_\bl$ bias and is given by 
  \begin{align*}
N^{(0)}_{\bl,\text{\tiny obs}} \equiv \delta_0^{-2}[\aw_\bl]^2 \int \frac{d^2\bk}{(2\pi)^2} \Bigl(   \gw_{\bs l,\bk\,}   &\gw_{\bl, -\bk-\bl\,}   \tilde X_{\bk+\bl}\tilde Y^*_{\bk+\bl}\tilde X_{\bk}\tilde Y^*_{\bk} 
\\&+ |\gw_{\bs l,\bs k}|^2 |\tilde X_{\bk+\bl}|^2 |\tilde Y_{\bk}|^2\Bigr).
\end{align*}
Using $N^{(0)}_{\bl,\text{\tiny obs}}$ in place of $N^{(0)}_{\bl}$ yields  $\delta_0^{-1} |\hat \phi_\bl^\xy  |^2 -N^{(0)}_{\bl,\text{\tiny obs}}$ as an estimate of $C_\bl^{\phi\phi}$ which does require knowledge of $C_\bl^{\phi\phi}$ to compute it.

{\em Remark:}
Up to this point we have been assuming an infinite sky when computing the Fourier transform $\int \frac{d^2 \bs x}{2\pi} e^{-i\bl\cdot \bs x}$. However,  the above Fourier transforms will typically be done on a pixelized periodic finite sky. In this case, $\delta_0$ is approximated as $1/ \Delta \bl$ where $\Delta \bl$ is the area element of the grid in Fourier space induced by finite area sky. For the remainder of the paper we do not distinguish the finite versus infinite case and simply equate $\delta_0$ with $1/ \Delta \bl$ leaving it understood that equality holds in the limit as $\Delta \bl\rightarrow 0$.

\begin{figure}[h]
\includegraphics[height = 2.8in]{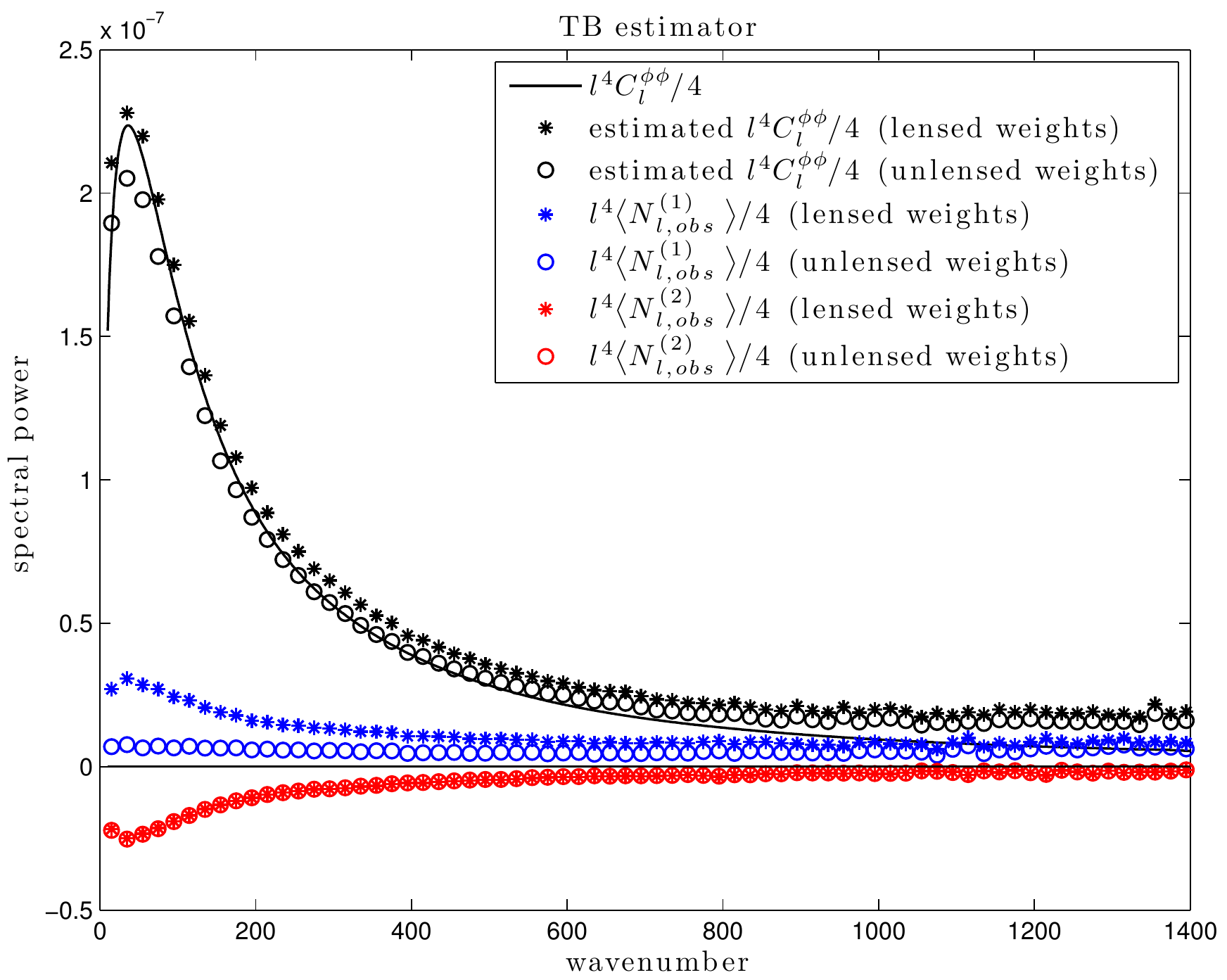} 
\caption{The first order bias term  $l^4\bigl\langle N^{(1)}_{\bl, \text{\tiny obs}}\bigr\rangle/4$ is denoted `{\large\textcolor{blue}{$\ast$}}' for $\tilde\phi^\xy_\bl$ and  `{\large \textcolor{blue}{$\circ$}}' for $\hat\phi^\xy_\bl$. The second order bias term  $l^4\bigl\langle N^{(2)}_{\bl, \text{\tiny obs}}\bigr\rangle/4$ is  denoted  `{\large\textcolor{red}{$\ast$}}' for $\tilde\phi^\xy_\bl$ and  `{\large \textcolor{red}{$\circ$}}' for $\hat\phi^\xy_\bl$. Finally, the expected value of the spectral density estimates  $l^4\langle\delta_0^{-1}  |\hat \phi_\bl^\xy  |^2 -N^{(0)}_{\bl,\text{\tiny obs}}\rangle/4$ and $l^4\langle \delta_0^{-1} |\tilde  \phi_\bl^\xy  |^2 -N^{(0)}_{\bl,\text{\tiny obs}}\rangle/4$  are denoted by `{\large\textcolor{black}{$\circ$}}' and `{\large\textcolor{black}{$\ast$}}' respectively. 
See Section \ref{sim} for details.
 \label{fig3} }
\end{figure}

\subsection{The bias of $\delta_0^{-1} |\hat \phi_\bl^\xy  |^2 -N^{(0)}_{\bl,\text{\tiny obs}}$}
\label{biasP}
We will derive three main expansions, which decompose the bias:  $\bigl\langle\delta_0^{-1} |\hat \phi_\bl^\xy  |^2 -N^{(0)}_{\bl,\text{\tiny obs}}\bigr\rangle - C_\bl^{\phi\phi}$ (and similar definitions for the modified quadratic estimate).  These expansions will be denoted as follows
\begin{align}
\hat \phi^\xy_{\bs l} &=  \mathscr L^{(0)}_\bl + \mathscr L^{(1)}_\bl+\mathscr L^{(2)}_\bl+\cdots \label{eex1} \\
|\hat \phi^\xy_{\bs l}|^2 &=  \mathscr O^{(0)}_\bl\,\, + \mathscr O^{(1)}_\bl\,\,+\mathscr O^{(2)}_\bl\,\,+\cdots \label{eex2}\\
 N^{(0)}_{\bl,\text{\tiny obs}}  &= \mathscr N^{(0)}_\bl + \mathscr N^{(1)}_\bl+\mathscr N^{(2)}_\bl+\cdots \label{eex3}
\end{align}
where the terms $\mathscr L^{(j)}_\bl$,  $\mathscr O^{(j)}_\bl$ and $\mathscr N^{(j)}_\bl$ are each of order $j$ in $\phi$.
Expansions (\ref{eex1}), (\ref{eex2}) and  (\ref{eex3}) are all obtained by expanding $\tilde X_\bk$ and $\tilde Y_\bk$ using (\ref{ex1}) then regrouping the terms by the order of $\phi$. This expansion and subsequent re-grouping yields the following analytic expressions for each term:
\begin{align*}
\mathscr L^{(j)}_\bl&\equiv \sum_{p=0}^j  A^{\xy}_{\bs l} \int \frac{d^2\bs k}{2\pi} g^{\xy}_{\bs l,\bs k\,} \left[    \delta^p X_{\bs k+\bs l}\delta^{j-p} Y_{\bs k}^*\right];  \\
 \mathscr O^{(j)}_\bl&\equiv \sum_{p=0}^j  \mathscr L^{(p)}_\bl  \mathscr L^{(j-p)^*}_\bl ; \\
 \mathscr N^{(j)}_\bl &\equiv \sum_{p=0}^j \! \delta_0^\text{\tiny $-2$} [\aw_\bl]^2\! \int \!\! \frac{d^2\bk}{(2\pi)^2}  \Bigl(   \gw_{\bs l,\bk\,}   \gw_{\bl, -\bk-\bl\,}  \mathscr P^{\xandy XY\!,p}_{\bk+\bl} \mathscr P^{\xandy XY\!,j-p}_{\bk} \\
 &\qquad\qquad\qquad\qquad\qquad+ |\gw_{\bs l,\bs k}|^2   \mathscr P^{\xandy XX,p}_{\bk+\bl} \mathscr P^{\xandy YY,j-p}_{\bk}  \Bigr),
\end{align*}
 where  $\mathscr P_\bk^{\xandy XY,j }\equiv \sum_{p=0}^j \delta^pX_\bk \delta^{j-p}Y^*_\bk $ so that $\tilde X_\bk \tilde Y_\bk^* = \sum_{p=0}^\infty \mathscr P^{\xandy XY,
 p}_\bk$.
Using the fact that $\langle\mathscr O^{(2j+1)}_\bl\rangle=\langle\mathscr N^{(2j+1)}_\bl\rangle=0$ one then gets that
\begin{align}
\label{ee22}
\bigl\langle |\hat \phi_\bl^\xy  |^2-  \delta_0 N^{(0)}_{\bl,\text{\tiny obs} }   \bigr\rangle & = \sum_{j=0}^\infty \bigl\langle\mathscr O^{(2j)}_\bl - \delta_0 \mathscr N^{(2j)}_\bl\bigr\rangle.
\end{align}
Now by defining $N_{\bl, \text{\tiny obs}}^{(j)}$ for $j\geq 1$ as
\begin{align}
\label{N1b}
N^{(j)}_{\bl,\text{\tiny obs}} \equiv 
\begin{cases}
 \delta_0^{-1} \mathscr O^{(2j)}_\bl - \mathscr N^{(2j)}_\bl  - C_\bl^{\phi\phi}& \text{when $j=1$}\\
 \delta_0^{-1} \mathscr O^{(2j)}_\bl - \mathscr N^{(2j)}_\bl & \text{when $j> 1$}\\
\end{cases}
\end{align}
we have that
\begin{align}
\bigl\langle |\hat \phi_\bl^\xy  |^2\bigr\rangle&= \delta_0\bigl[C_\bl^{\phi\phi} + \bigl\langle N^{(0)}_{\bl,\text{\tiny obs}}\bigr\rangle + \bigl\langle N^{(1)}_{\bl,\text{\tiny obs}} \bigr\rangle +\cdots   \bigr] \nonumber\\
&\quad+\bigl\langle\mathscr O^{(0)}_\bl - \delta_0  \mathscr N^{(0)}_{\bl}   \bigr\rangle
\end{align}
where $N^{(j)}_{\bl,\text{\tiny obs}}$ is of order $\phi^{2j}$. 
 Therefore the order $\phi^{2j}$ term in the bias of  $\delta_0^{-1} |\hat \phi_\bl^\xy  |^2 -N^{(0)}_{\bl,\text{\tiny obs}}$ is given by $ \bigl\langle N^{(j)}_{\bl,\text{\tiny obs}} \bigr\rangle$ for $j\geq 1$.

{\em Remark:} 
Notice that the zero order term can be computed easily 
\begin{align*}
\bigl\langle\mathscr O^{(0)}_\bl - \delta_0  \mathscr N^{(0)}_{\bl}   \bigr\rangle  &=-\frac{[\aw_\bl \gw_{\bs l,-\bl/2}]^2}{(2\pi)^2}  \\
&\quad\times \Bigl[   C^{\xandy XX}_{\bl/2,\text{\tiny exp}}   C^{\xandy YY}_{\bl/2,\text{\tiny exp}} +   (C^{\xandy XY}_{\bl/2,\text{\tiny exp}})^2 \Bigr]
\end{align*}
 when $\bl \neq 0$. Moreover, since this term does not depend on $\phi$ one can simply subtract it out of the estimator when using $\delta_0^{-1} |\hat \phi_\bl^\xy  |^2 -N^{(0)}_{\bl,\text{\tiny obs}}$  for estimation.

\subsection{The difference between $N^{(1)}_\bl$ and $\bigl\langle N^{(1)}_{\bl, \text{\tiny obs}}\bigr\rangle$}
By matching  the right hand side of equation (\ref{eex2}) with the right hand side of (\ref{expan}) one gets that 
   $\delta_0 N^{(1)}_\bl$ equals $\langle\mathscr O^{(2)}_\bl\rangle$ minus  $\delta_0C_\bl^{\phi\phi}$ and any disconnected terms.  
In contrast,  $\delta_0 \bigl\langle N^{(1)}_{\bl,\text{\tiny obs}}\bigr\rangle$ equals $\langle\mathscr O^{(2)}_\bl\rangle$ minus  $\delta_0C_\bl^{\phi\phi}$ and $\langle \delta_0 \mathscr N_\bl^{(2)}\rangle$.  
Therefore the difference between $\delta_0N^{(1)}_\bl$ and $\delta_0 \bigl\langle N^{(1)}_{\bl,\text{\tiny obs}}\bigr\rangle$ is the difference between $\langle \delta_0 \mathscr N_\bl^{(2)}\rangle$ and the disconnected terms  in  $\langle\mathscr O^{(2)}_\bl\rangle$.
By  expanding  $\delta_0\widetilde C^\xy_{\bl,\text{\tiny exp}} =\sum_{p=0}^\infty  \langle\mathscr P^{\xandy XY,p}_\bl \rangle$ in the definition of $\delta_0 N^{(0)}_\bl$ (which equals all the disconnected terms),  and retaining only the order $\phi^2$ terms (what remains equals all the disconnected terms in $\langle\mathscr O^{(2)}_\bl\rangle$)
we get 
\begin{align}
&\sum_{p=0}^2  \delta_0^{-1} [\aw_\bl]^2 \int \frac{d^2\bk}{(2\pi)^2}\Bigl(   \gw_{\bs l,\bk\,}   \gw_{\bl, -\bk-\bl\,}  \bigl\langle\mathscr P^{\xandy XY,p}_{\bk+\bl} \bigr\rangle  \bigl\langle\mathscr P^{\xandy XY,2-p}_{\bk} \bigr\rangle\nonumber \\
&\qquad\qquad\qquad\qquad
+ |\gw_{\bs l,\bs k}|^2   \bigl\langle \mathscr P^{\xandy XX,p}_{\bk+\bl}\bigr\rangle \bigl\langle \mathscr P^{\xandy YY,2-p}_{\bk} \bigr\rangle \Bigr). \label{ddee}
\end{align} 
Therefore the difference between $\delta_0N^{(1)}_\bl$ and $\delta_0\bigl\langle N^{(1)}_{\bl,\text{\tiny obs}}\bigr\rangle$ is given by the difference between (\ref{ddee}) and $\langle \delta_0 \mathscr N_\bl^{(2)}\rangle$ which equals 
\begin{align}
 &\sum_{p=0}^2  \delta_0^{-1} [\aw_\bl]^2 \int \frac{d^2\bk}{(2\pi)^2} \Bigl(   \gw_{\bs l,\bk\,}   \gw_{\bl, -\bk-\bl\,}  \left\langle\mathscr P^{\xandy XY,p}_{\bk+\bl} \mathscr P^{\xandy XY,2-p}_{\bk} \right\rangle\nonumber\\
&\qquad\qquad\qquad\qquad+ |\gw_{\bs l,\bs k}|^2   \left\langle \mathscr P^{\xandy XX,p}_{\bk+\bl} \mathscr P^{\xandy YY,2-p}_{\bk} \right\rangle \Bigr).
\label{enn}
\end{align}
For the experimental conditions analyzed in this paper this difference is small.

\section{Fast Monte Carlo algorithms}
\label{fmca}

In this section we give two simulation based methods for quickly estimating $ N_{\bl}^{(j)}$ and $\langle N_{\bl, \text{\tiny obs}}^{(j)}\rangle$ for $j\geq 0$. The first method  simply observes that each term  $\mathscr L^{(j)}_\bl$,  $\mathscr O^{(j)}_\bl$ and $\mathscr N^{(j)}_{\bl}$ have fast Fourier transform characterizations which can be used for simulating $N_{\bl, \text{\tiny obs}}^{(j)}$ and---by averaging multiple realizations---for estimating $\bigl\langle N_{\bl, \text{\tiny obs}}^{(j)}\bigr\rangle$. This algorithm also extends to $N^{(j)}_{\bl}$ by replacing the expansion of $N_{\bl, \text{\tiny obs}}^{(0)}$ in (\ref{eex3}) with the corresponding expansion for $N_{\bl}^{(0)}$. The second method is exclusive to $N^{(1)}_{\bl}$ and  uses correlated and uncorrelated  CMB fields to mimic the appropriate Wick contractions for  the connected terms in (\ref{totalSum}). 

\subsection{FFT algorithms for $N^{(j)}_{\bl,\text{\tiny obs}}$ and $N^{(j)}_{\bl}$ }
\label{fft1}

The fast simulation techniques presented in this section depend on the fact that the transforms which characterize $\hat\phi_\bl^\xy$ and $\tilde\phi_\bl^\xy$ can be derived as Fourier transforms of point-wise products of gradients in pixel space. This was first utilized  in \cite{Hu2001b}  and \cite{HuOka2002} for the flat sky quadratic estimators.   In the appendix, we present these transforms along with some additional FFT transforms which allow fast simulation of the fields $\mathscr L^{(j)}_\bl$,  $\mathscr O^{(j)}_\bl$ and $\mathscr N^{(j)}_{\bl}$. This is the basis of the algorithm which  then uses equation (\ref{N1b}) to simulate  $N^{(j)}_{\bl,\text{\tiny obs}}$.

For $\mathscr L^{(j)}_\bl$ and  $\mathscr O^{(j)}_\bl$ first notice that each term $\delta^jX_{\bs l}$ can be easily simulated  in the pixel domain since $\delta^j X(\bs x)$ is the point-wise product of derivatives of $X(\bs x)$ and $\phi(\bs x)$. 
 Moreover the quadratic estimate applied to these terms,  resulting in $A^{\xy}_{\bs l} \int \frac{d^2\bs k}{2\pi} g^{\xy}_{\bs l,\bs k\,}  \left[ \delta^p X_{\bs k+\bs l}\delta^{j-p} Y_{\bs k}^*\right]$, is also easily computed  by direct application of the formulas presented in the Appendix. Now summing over $p\in\{0,\ldots, j\}$ gives  fast simulation of $\mathscr L^{(j)}_\bl$ which, in turn, gives $\mathscr O^{(j)}_\bl$ by  taking quadratic combinations of  $\mathscr L^{(0)}_\bl\!, \ldots, \mathscr L^{(j)}_\bl$.
Finally,  to simulate  $\mathscr N^{(j)}_\bl$ start by  noticing that each term $\mathscr P_\bk^{\xandy XY,j }\equiv \sum_{p=0}^j \delta^pX_\bk \delta^{j-p}Y^*_\bk $ simulates easily from  $\delta^pX_{\bs k}$  and $\delta^{j-p}Y_{\bs k}^*$. 
 Then to recover $\mathscr N^{(j)}_\bl$ one uses the FFT transformations presented in the appendix for quick simulations of $ [\aw_\bl]^2\! \int \!\! \frac{d^2\bk}{(2\pi)^2}    \gw_{\bs l,\bk\,}   \gw_{\bl, -\bk-\bl\,}  \mathscr P^{\xandy XY\!,p}_{\bk+\bl} \mathscr P^{\xandy XY\!,j-p}_{\bk}$ and $  [\aw_\bl]^2\! \int \!\! \frac{d^2\bk}{(2\pi)^2}  |\gw_{\bs l,\bs k}|^2   \mathscr P^{\xandy XX,p}_{\bk+\bl} \mathscr P^{\xandy YY,2-p}_{\bk} $.
 
 Monte Carlo averaging $n$ independent simulations of $N_{\bl, \text{\tiny obs}}^{(j)}$ will yield estimates of $\bigl\langle N_{\bl, \text{\tiny obs}}^{(j)}\bigr\rangle$ with error bars that can be approximated by $s_\bl/\sqrt{n}$ where $s_\bl$ denotes the standard deviation of the samples at each frequency $\bl$. In addition, when the noise and beam structure are isotropic one can reduce the error bars by radially averaging  each simulation of $N_{\bl, \text{\tiny obs}}^{(j)}$. 
 These error bars can then be used for re-fitting algorithms where $C_\bl^{\phi\phi}$ is iteratively fit to $\delta_0^{-1} |\hat \phi_\bl^\xy  |^2 - N^{(0)}_{\bl,\text{\tiny obs}} - \mu_\bl$ where $\mu_\bl$ is an approximate bias correction based on the estimates of the higher order bias correction terms $\bigl\langle N_{\bl, \text{\tiny obs}}^{(j)}\bigr\rangle$. Notice that for fast re-fitting algorithms, it maybe be advantageous initially to tolerate relatively large the Monte Carlo error bars, then iteratively increase the number of Monte Carlo samples at each re-fit.

We finally mention that the above simulation methods which are used to estimate  $\bigl\langle N_{\bl, \text{\tiny obs}}^{(j)}\bigr\rangle$ can also be used to estimate  $N^{(j)}_{\bl}$. The only change is to replace  each term $\mathscr P^{\xandy XY,p}_{\bk+\bl}$,  $\mathscr P^{\xandy XY,2j-p}_{\bk}$,  $\mathscr P^{\xandy XX,p}_{\bk+\bl}$ and $\mathscr P^{\xandy XX,2j-p}_{\bk}$ in the definition of $\mathscr N^{(2j)}_\bl $ with their respective expected values:  $\bigl\langle\mathscr P^{\xandy XY,p}_{\bk+\bl}\bigr\rangle$,  $\bigl\langle\mathscr P^{\xandy XY,2j-p}_{\bk}\bigr\rangle$,  $\bigl\langle\mathscr P^{\xandy XX,p}_{\bk+\bl}\bigr\rangle$ and $\bigl\langle\mathscr P^{\xandy XX,2j-p}_{\bk}\bigr\rangle$.

%
%
%

\begin{figure*}[ht]
\includegraphics[height = 2.8in]{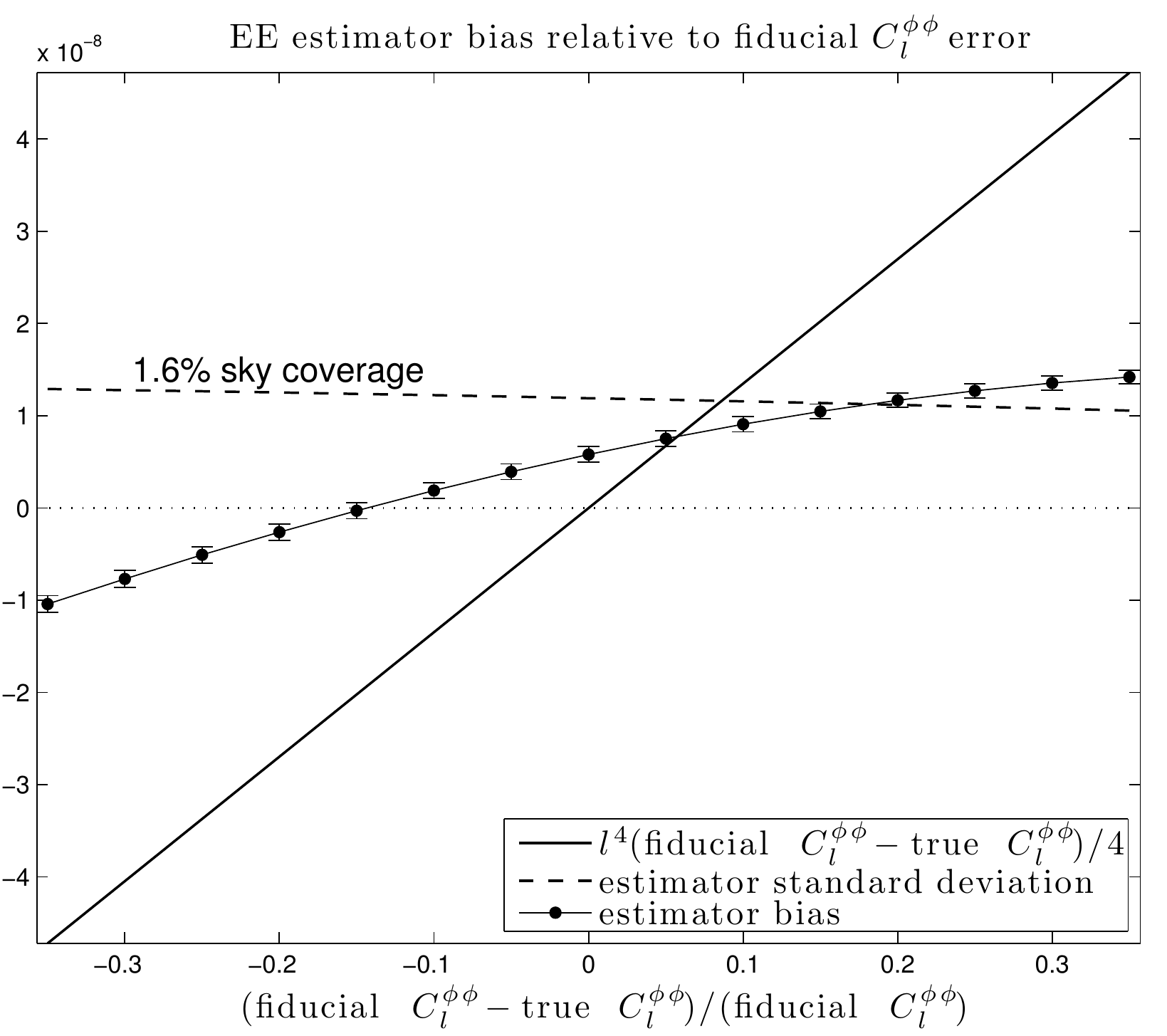}
\includegraphics[height = 2.8in]{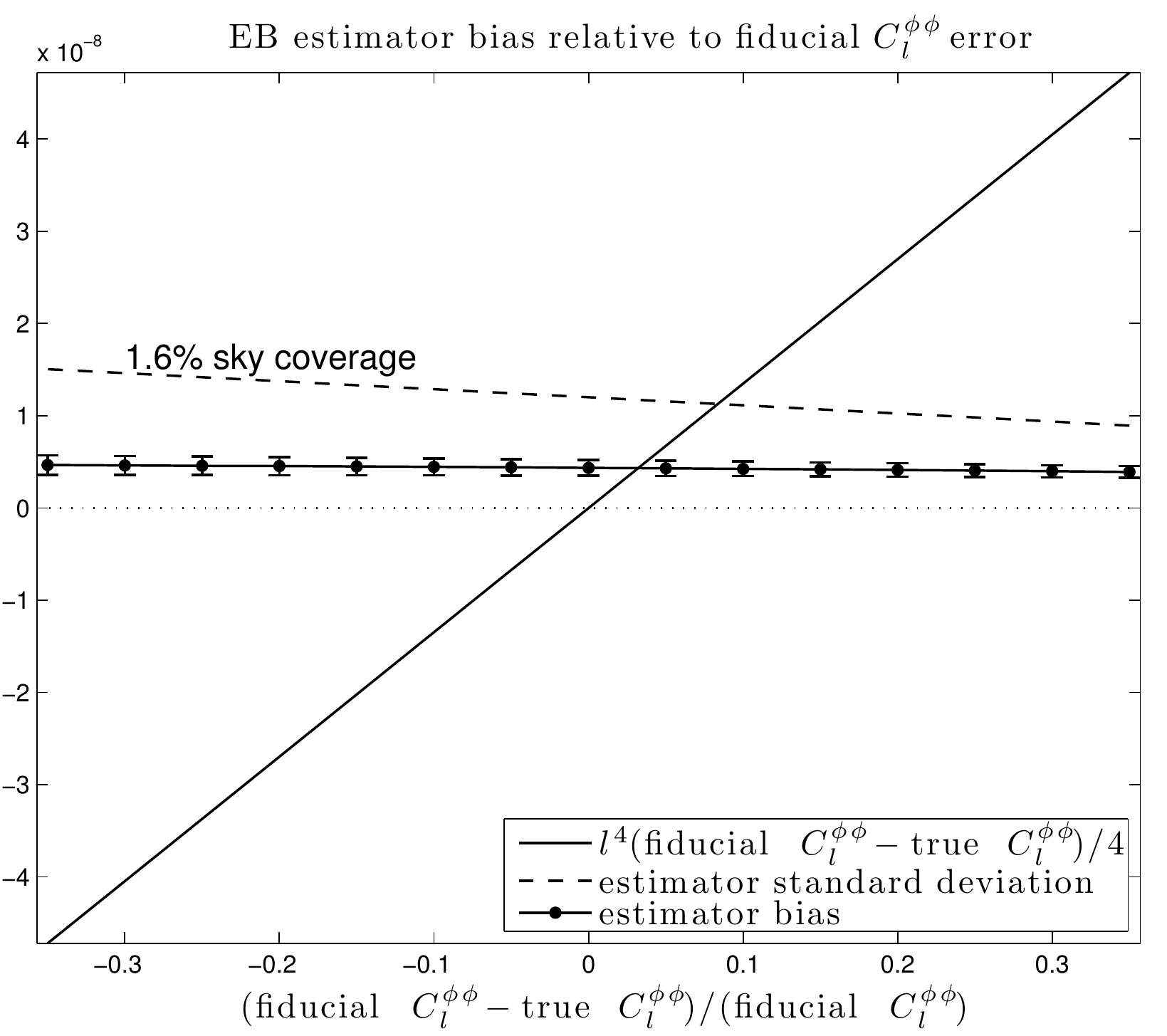}
\caption{
Spectral estimation bias averaged over the multiple bin $l\in[10,200]$ as a function of fiducial uncertainty. See Section \ref{sim} for details.
\label{fig4}
 }
\end{figure*}

\subsection{Coupling lensing fields for $N^{(1)}_{\bl}$}
\label{fft2}

A second algorithm, for Monte Carlo estimation of $N^{(1)}_{\bl}$, is to use correlated and uncorrelated CMB fields to mimic the Wick contraction structure appearing in the definition of $N^{(1)}_{\bl}$.
 The algorithm is easiest to illustrate  with the EB  quadratic estimator since the definition of $N^{(1)}_{\bl}$ only involves a small number of  connected terms in (\ref{totalSum}) (when assuming a zero B mode in the unlensed CMB polarization field).  The algorithm is derived by first noticing that the sum of the connected terms of order $\phi^2$ in $\langle \hat\phi^\text{\tiny $X\!Y$}_\bl  \hat\phi^{\text{\tiny $X\!Y$}^*}_{\bl^\prime} \rangle$ is given by
\begin{align*}
\nonumber
\delta_{\bl - \bl^\prime} \bigl[N_\bl^{(1)} +C_\bl^{\phi\phi}\bigr]
 = &\aw_\bl \aw_{\bl^\prime} \int \frac{d^2\bk}{2\pi} \int \frac{d^2\bk^\prime}{2\pi} \gw_{\bs l, \bk\,} g^{\xy }_{\bs l^\prime,\bk^\prime\,}\\
 &\times
\Bigl[
\bigl\langle 
\contraction{}{\delta^0\!E_{\bk+\bl}}{}{\delta^1\!B^*_{\bk}}
\contraction{\delta^0\!E_{\bk+\bl}\delta^1\!B^*_{\bk}}{\delta^0\!E^*_{\bk^\prime+\bl^\prime}}{}{\delta^1\!B_{\bk^\prime}}
\bcontraction{\delta^0\!E_{\bk+\bl}}{\delta^1\!B^*_{\bk}}{\delta^0\!E^*_{\bk^\prime+\bl^\prime}}{\delta^1\!B_{\bk^\prime}}
{\delta^0\!E_{\bk+\bl}}{\delta^1\!B^*_{\bk}}{\delta^0\!E^*_{\bk^\prime+\bl^\prime}}{\delta^1\!B_{\bk^\prime}}
\bigr\rangle\nonumber\\
&\qquad+
\bigl\langle 
\contraction[2ex]{}{\delta^0\!E_{\bk+\bl}}{\delta^1\!B^*_{\bk}\delta^0\!E^*_{\bk^\prime+\bl^\prime}}{\delta^1\!B_{\bk^\prime}}
\contraction{\delta^0\!E_{\bk+\bl}}{\delta^1\!B^*_{\bk}}{}{\delta^0\!E^*_{\bk^\prime+\bl^\prime}}
\bcontraction{\delta^0\!E_{\bk+\bl}}{\delta^1\!B^*_{\bk}}{\delta^0\!E^*_{\bk^\prime+\bl^\prime}}{\delta^1\!B_{\bk^\prime}}
{\delta^0\!E_{\bk+\bl}}{\delta^1\!B^*_{\bk}}{\delta^0\!E^*_{\bk^\prime+\bl^\prime}}{\delta^1\!B_{\bk^\prime}}
\bigr\rangle \Bigr] 
\end{align*}
These two Wick contraction terms have a Monte Carlo characterization as follows. Let $(Q_\bl^{\phantom{\prime}},U_\bl^{\phantom{\prime}})$ and    $(Q_\bl^\prime,U_\bl^\prime)$ denote two independent realizations of the CMB polarization. Now let $\hat\phi^{\xandy EB}_{\bl,1}$ and $\hat\phi^{\xandy EB}_{\bl,2}$ denote the EB quadratic estimator applied to the pairs $(E_\bl^{\phantom{\prime}}, \delta^1\!B_\bl^{\phantom{\prime}})$ and $(E^\prime_\bl, \delta^1\!B_\bl^\prime)$, respectively, where $\delta^1\!B^{\phantom{\prime}}$ and $\delta^1\!B^\prime $  use the same lensing potential $\phi$
 so that 
\begin{align*}
\delta^1\!B_\bl &\equiv \sin(2\varphi_\bl) \delta^1 Q_\bl - \cos(2\varphi_\bl) \delta^1 U_\bl \\
\delta^1\!B^\prime_\bl &\equiv \sin(2\varphi_\bl) \delta^1 Q^\prime_\bl - \cos(2\varphi_\bl) \delta^1 U^\prime_\bl
\end{align*}
where
$ \delta^1\!Q^{\phantom{\prime}}(\bs x) \equiv \nabla^aQ(\bs x) \nabla_a \phi(\bs x)$ and $\delta^1\!Q^\prime(\bs x) \equiv \nabla^aQ^\prime(\bs x) \nabla_a \phi (\bs x) $ (similar definitions for $ \delta^1\!U^{\phantom{\prime}}(\bs x)$ and  $\delta^1\!U^{\prime}(\bs x)$). Notice the independence structure of the simulated fields--that $E$ and $B$ are independent of  $E^\prime$ and $B^\prime$---implies
\begin{align*}
   \bigl\langle  \hat\phi^{\xandy EB}_{\bl,1}\hat\phi^{\text{\tiny $E\!B$}^*}_{\bl^\prime,2} \bigr\rangle 
 = \aw_\bl \aw_{\bl^\prime}& \int \frac{d^2\bk}{2\pi} \int \frac{d^2\bk^\prime}{2\pi} \gw_{\bs l, \bk\,} g^{\xy }_{\bs l^\prime,\bk^\prime\,}\\
 &\times
\bigl\langle 
\contraction{}{\delta^0\!E_{\bk+\bl}}{}{\delta^1\!B^*_{\bk}}
\contraction{\delta^0\!E_{\bk+\bl}\delta^1\!B^*_{\bk}}{\delta^0\!E^*_{\bk^\prime+\bl^\prime}}{}{\delta^1\!B_{\bk^\prime}}
\bcontraction{\delta^0\!E_{\bk+\bl}}{\delta^1\!B^*_{\bk}}{\delta^0\!E^*_{\bk^\prime+\bl^\prime}}{\delta^1\!B_{\bk^\prime}}
{\delta^0\!E_{\bk+\bl}}{\delta^1\!B^*_{\bk}}{\delta^0\!E^*_{\bk^\prime+\bl^\prime}}{\delta^1\!B_{\bk^\prime}}
\bigr\rangle.
\end{align*}
 Similarly let $\hat\phi^{\xandy EB}_{\bl,3}$ and $\hat\phi^{\xandy EB}_{\bl,4}$ denote the EB quadratic estimator applied to the pairs $(E_\bl^\prime, \delta^1\!B_\bl^{\phantom{\prime}})$ and $(E_\bl^{\phantom{\prime}}, \delta^1\!B_\bl^\prime)$ respectively (again using the same lensing potential $\phi$). 
 Then
  \begin{align*}
    \bigl\langle  \hat\phi^{\xandy EB}_{\bl,3}\hat\phi^{\text{\tiny $E\!B$}^*}_{\bl^\prime,4} \bigr\rangle 
   = \aw_\bl \aw_{\bl^\prime} &\int \frac{d^2\bk}{2\pi} \int \frac{d^2\bk^\prime}{2\pi} \gw_{\bs l, \bk\,} g^{\xy }_{\bs l^\prime,\bk^\prime\,}\\
 &\times
\bigl\langle 
\contraction[2ex]{}{\delta^0\!E_{\bk+\bl}}{\delta^1\!B^*_{\bk}\delta^0\!E^*_{\bk^\prime+\bl^\prime}}{\delta^1\!B_{\bk^\prime}}
\contraction{\delta^0\!E_{\bk+\bl}}{\delta^1\!B^*_{\bk}}{}{\delta^0\!E^*_{\bk^\prime+\bl^\prime}}
\bcontraction{\delta^0\!E_{\bk+\bl}}{\delta^1\!B^*_{\bk}}{\delta^0\!E^*_{\bk^\prime+\bl^\prime}}{\delta^1\!B_{\bk^\prime}}
{\delta^0\!E_{\bk+\bl}}{\delta^1\!B^*_{\bk}}{\delta^0\!E^*_{\bk^\prime+\bl^\prime}}{\delta^1\!B_{\bk^\prime}}
\bigr\rangle.
\end{align*}
This leads to the following Monte Carlo averaging characterization of $N^{(1)}_\bl$
\begin{equation}
\label{uuu}
 N_\bl^{(1)}=   \delta_0^{-1} \bigl\langle  \hat\phi^{\xandy EB}_{\bl,1}\hat\phi^{\text{\tiny $E\!B$}^*}_{\bl,2} +  \hat\phi^{\xandy EB}_{\bl,3}\hat\phi^{\text{\tiny $E\!B$}^*}_{\bl,4} \bigr\rangle  - C_\bl^{\phi\phi}.  
 \end{equation}
 The above formula also holds for the lensed quadratic estimator  $\tilde\phi^{\xandy EB}_{\bl} $ as well.
It is easy to see that this method can be extended to all other polarization quadratic estimators, but with decidedly more connected terms.

Notice that there is an additional simplification when using the quadratic estimator, $\hat\phi^{\xandy EB}_\bl$, without lensed weights. In particular,  $ \delta_0^{-1} \bigl\langle   \hat\phi^{\xandy EB}_{\bl,3}\hat\phi^{\text{\tiny $E\!B$}^*}_{\bl,4} \bigr\rangle  = C_\bl^{\phi\phi}$ so that equation (\ref{uuu}) simplifies to  $N_\bl^{(1)}=   \delta_0^{-1} \bigl\langle  \hat\phi^{\xandy EB}_{\bl,1}\hat\phi^{\text{\tiny $E\!B$}^*}_{\bl,2}  \bigr\rangle$. However, this simplification does not hold for $\tilde \phi^{\xandy EB}_\bl$ since the bias factor $B_\bl^{\xandy EB}$, defined in Section \ref{defB},  implies $ \delta_0^{-1} \bigl\langle   \tilde\phi^{\xandy EB}_{\bl,3}\tilde\phi^{\text{\tiny $E\!B$}^*}_{\bl,4} \bigr\rangle  = [B_\bl^{\xandy EB}]^2 C_\bl^{\phi\phi}\neq C_\bl^{\phi\phi}$.

One advantage of using this coupling technique is that the simulation of the $\delta^0X$ and $\delta^0Y$ can be done without including the additive experimental noise. 
To see why, notice that any CMB Wick contraction connecting $\delta^0X$ and $\delta^0Y$ (which depends on the additive experimental noise) must have  bottom contraction symbols connect $\delta^1X$ and $\delta^1Y$. This must yield a disconnected term which does not contribute  to $N^{(1)}_\bl$.

\begin{figure*}[ht]
\includegraphics[height = 2.8in]{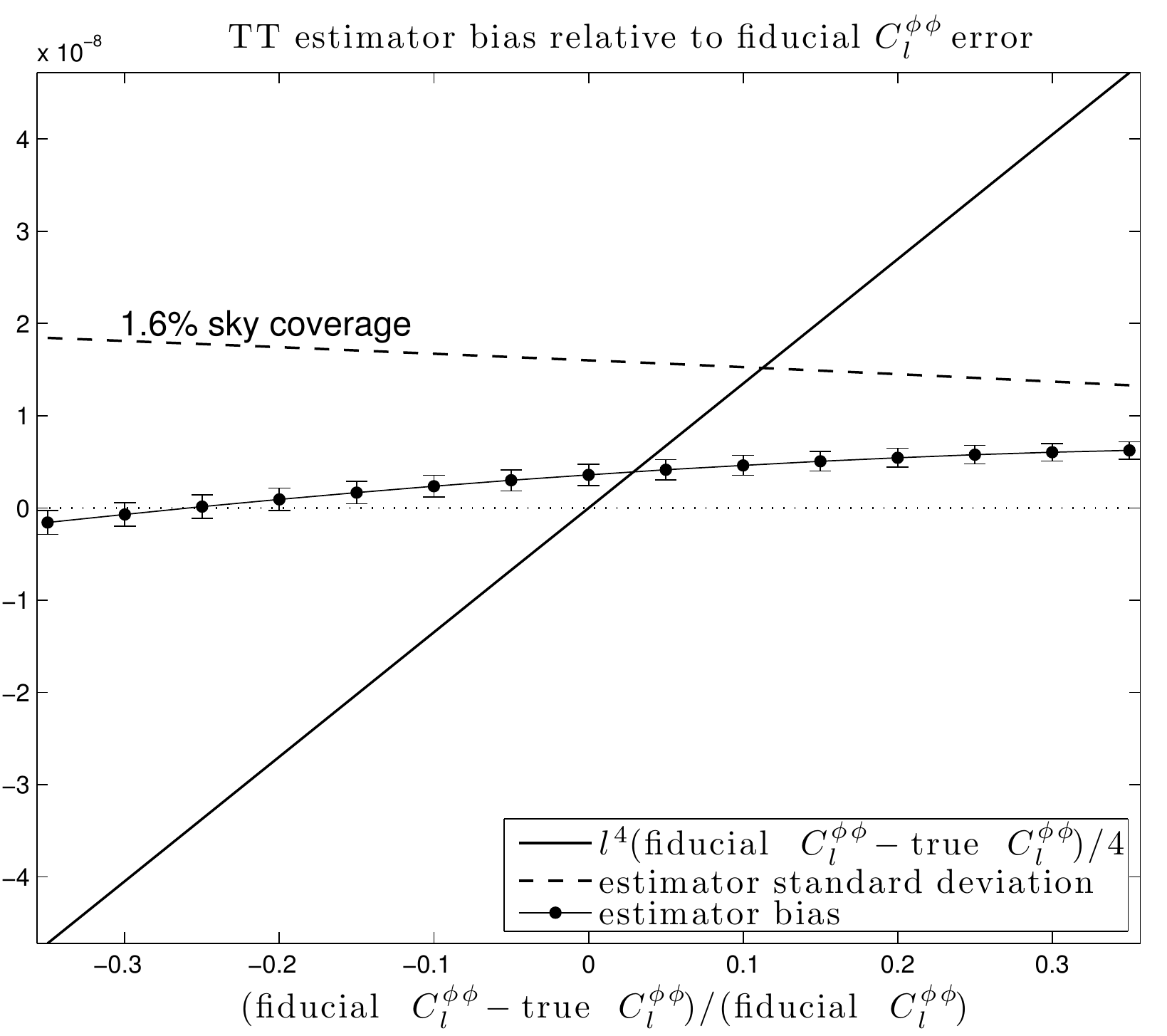}
\includegraphics[height = 2.8in]{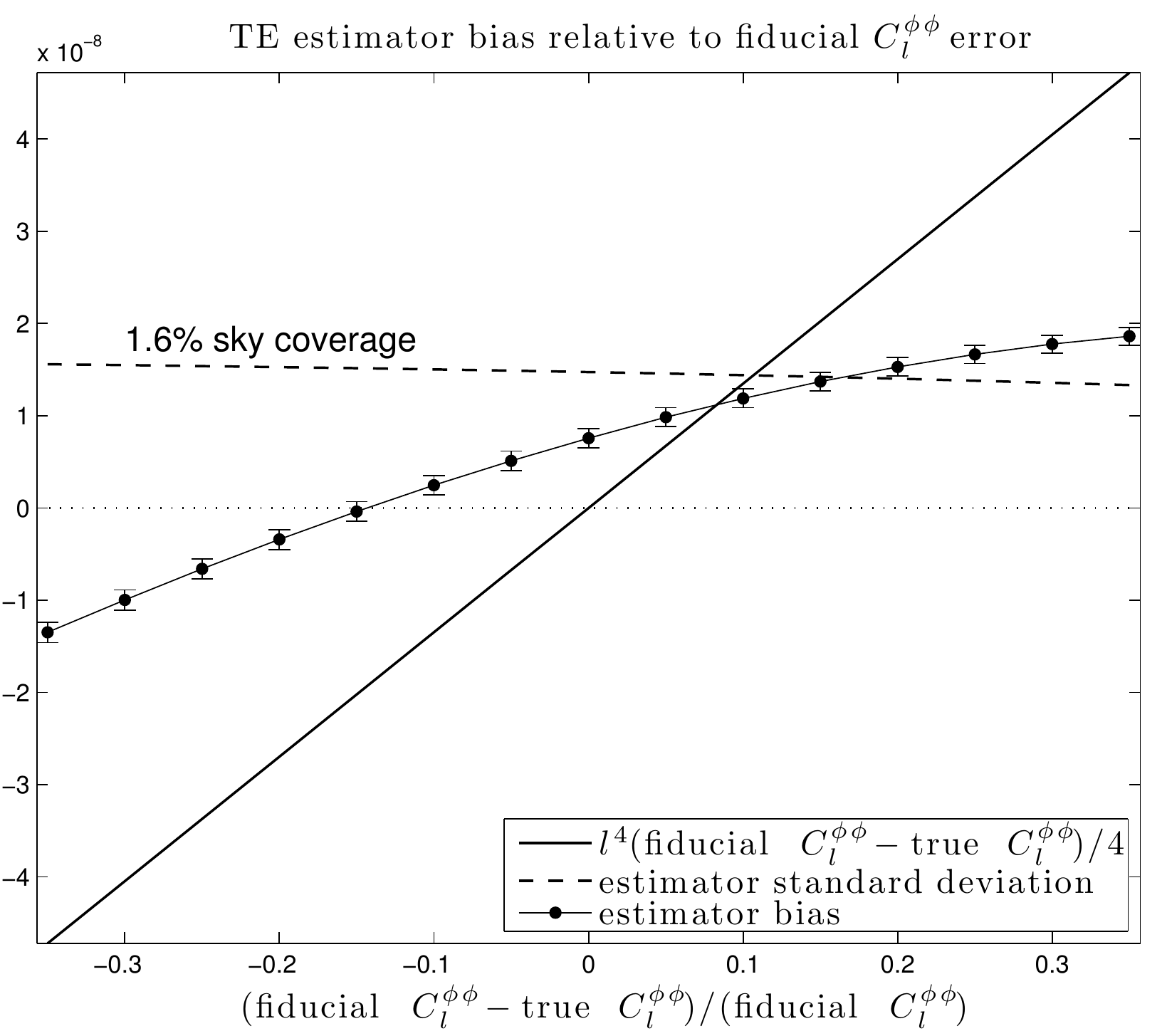}
\caption{
Spectral estimation bias averaged over the multiple bin $l\in[10,200]$ as a function of fiducial uncertainty. See Section \ref{sim} for details. \label{fig5}
 }
\end{figure*}

\section{Simulation}
\label{sim}
We perform two  simulation experiments under experimental conditions similar to those found in future ACTpol/SPTpol experiments.
The first simulation explores the bias terms  $ N^{(j)}_{\bl}$ and $\bigl\langle N^{(j)}_{\bl, \text{\tiny obs}}\bigr\rangle$ for the two quadratic estimators $\hat\phi^\xy_\bl$ and $\tilde \phi^\xy_\bl$. The results are summarized in figures \ref{fig1}, \ref{fig2} and \ref{fig3}. 
For these simulations, to compute the lensed weights in $\tilde \phi^\xy_\bl$, we use the same fiducial model for $C_\bl^{\phi\phi}$ as the simulation model for $C_\bl^{\phi\phi}$. In contrast, for the second set of simulations we explore the effect of  uncertainty in the fiducial model for $C_\bl^{\phi\phi}$ when computing $\tilde \phi^\xy_\bl$. The results are summarized in figures \ref{fig4}, \ref{fig5} and \ref{fig6}. The main conclusion of the first set of simulations is that although $\tilde \phi^\xy_\bl$ does reduce low $l$ estimation bias, this is accomplished by increasing $ N^{(1)}_{\bl}$ (or  $\bigl\langle N^{(1)}_{\bl, \text{\tiny obs}}\bigr\rangle$ in the case that one uses $N^{(0)}_{\bl, \text{\tiny obs}}$ instead of  $N^{(0)}_{\bl}$) to the point of canceling with $ N^{(2)}_{\bl}$ (or $\bigl\langle N^{(2)}_{\bl, \text{\tiny obs}}\bigr\rangle$ as the case may be) when the correct model for $C_\bl^{\phi\phi}$ is used to generate the lensed weights. 
The  second set of simulations show that this cancelation can be sensitive to the fiducial model for $C_\bl^{\phi\phi}$ depending on which estimator one uses: TE and EE are most sensitive, EB is least sensitive. At the end of this section we discuss the inferential implications for future experiments.

The cosmology used in our simulations are based on a flat,
power law $\Lambda$CDM cosmological model, with baryon density
$\Omega_b=0.044$; cold dark matter density $\Omega_\text{cdm}=0.21$;
cosmological constant density $\Omega_\Lambda=0.74$; Hubble parameter
$h=0.71$ in units of 100$\,$km$\,$s$^{-1}\,$Mpc$^{-1}$; primordial
scalar fluctuation amplitude $A_s(k=0.002\,$Mpc$^{-1}) = 2.45\times
10^{-9}$; scalar spectral index $n_s(k=0.002\,$Mpc$^{-1}) = 0.96$;
primordial helium abundance $Y_P=0.24$; and reionization optical depth
$\tau_r=0.088$. The CAMB code is used to generate the theoretical
power spectra \cite{CAMB}.

To construct the lensed CMB simulation  we first  generate a high resolution simulation of $\Theta(\bs x)$ and the gravitational potential $\phi(\bs x)$ on a periodic $25.6^\text{o} \times 25.6^\text{o}$ patch of the flat sky. The lensing operation is performed by taking the numerical gradient of $\phi$, then using linear interpolation to obtain the lensed field $\Theta(\bs x + \nabla \phi(\bs x))$.  
We down-sample the lensed field to obtain the desired 1.5 arcmin pixel resolution for the simulation output.  
 Finally, the observational noise  (with a standard deviation of $5\mu$K-arcmin on  $T$ and   $\sqrt{2} \times 5\mu$K-arcmin  on  $E$, $B$ and Gaussian beam FWHM=1.5 arcmin deconvolution)  is added in Fourier space.  For all of the simulations we assume a zero $B$ mode and a lensing potential $\phi$ which is uncorrelated with the CMB.
In contrast to the full lensing simulation, the pertabive expansions given in Section \ref{fft1} only require simulation of unlensed CMB fields at the low resolution 1.5 arcmin pixels. 

\subsection*{Figures \ref{fig1}, \ref{fig2} and \ref{fig3} }

Each plot in figures \ref{fig1}, \ref{fig2} and \ref{fig3} 
correspond to a different quadratic pairing $X,Y\in\{ T, E,B\}$ and shows the Monte Carlo approximations to   $l^4\bigl\langle N^{(1)}_{\bl, \text{\tiny obs}}\bigr\rangle/4$ and  $l^4\bigl\langle N^{(2)}_{\bl, \text{\tiny obs}}\bigr\rangle/4$ along with the expected value of the spectral density estimates  $l^4( \delta_0^{-1}  |\hat \phi_\bl^\xy  |^2 -N^{(0)}_{\bl,\text{\tiny obs}})/4$ and $l^4( \delta_0^{-1} |\tilde  \phi_\bl^\xy  |^2 -N^{(0)}_{\bl,\text{\tiny obs}})/4$ over different realizations of $\phi$, the CMB and the observational noise. Although not shown, the bias terms  $l^4 N^{(1)}_{\bl}/4$ were also computed, using  the coupling technique given in Section \ref{fft2}, and resulted in very similar plots (mostly indistinguishable above the Monte Carlo error).
   The spectral density estimates are computed from the all-order lensed simulations whereas the bias terms are computed using  perturbative expansions discussed in Section \ref{fft1}. The Monte Carlo approximations are based on 2000 independent  realizations for the $TT$, $EE$ and $EB$ estimators and 18000 independent realizations for the $TE$ and $TB$ estimators.  These estimates are then radially averaged on sliding concentric annuli with wavenumber bins of width $20$ to yield the plots shown in figures \ref{fig1}, \ref{fig2} and \ref{fig3}.  

 The main feature in these simulations is the large increase in $\bigl\langle N^{(1)}_{\bl, \text{\tiny obs}}\bigr\rangle$ bias at low $l$ for $\tilde \phi^\xy_\bl$ as compared to the corresponding quantity for $\hat\phi^\xy_\bl$, especially  for the EE, TE and TT estimators. 
In contrast, $\bigl\langle N^{(2)}_{\bl, \text{\tiny obs}}\bigr\rangle$ also increases in magnitude but to a lesser extent, enabling the cancelation with $\bigl\langle N^{(1)}_{\bl, \text{\tiny obs}}\bigr\rangle$. Since the terms $\bigl\langle N^{(1)}_{\bl, \text{\tiny obs}}\bigr\rangle$ and $\bigl\langle N^{(2)}_{\bl, \text{\tiny obs}}\bigr\rangle$ are not individually small but instead cancel, there 
is the potential for this cancelation to be offset when there is
 uncertainty in the fiducial model for $C_\bl^{\phi\phi}$ used to compute the lensing weights for the estimate $\tilde \phi^\xy_\bl$. In the next section we explore this sensitivity by analyzing the resulting estimation bias as a function of fiducial sensitivity.

\subsection*{Figures \ref{fig4}, \ref{fig5} and \ref{fig6} }

To explore the effect of fiducial uncertainty when computing the lensed weights in $\tilde \phi^\xy_\bl$ we fixed a fiducial model $C_\bl^{\phi\phi}$ used to compute the lensed weights,  then analyzed simulations under perturbations of $C_\bl^{\phi\phi}$. In particular, we considered simulation models which differ from the fiducial model by a maximum of $35$\% only in the multipole range $[10, 200]$. The simulation models are of the form $T_\bl C_\bl^{\phi\phi}$ where $T_\bl = 1$ when $l \notin [10,200]$ and $T_\bl = c$ when $l\in [10,200]$ where  $c$ ranges from $1.35$ to $0.65$.
For each  scalar $c$ we simulated 200 different CMB, lensing and noise fields.  At each simulation we recorded the estimation error, namely $l^4\bigl(\delta_0^{-1}  |\tilde \phi_\bl^\xy  |^2 -N^{(0)}_{\bl,\text{\tiny obs}}-T_\bl C_\bl^{\phi\phi}\bigr)/4$,  averaged over all frequencies  in the $l$ bin $[10,200]$. This error is then averaged all $200$ simulations to estimate the bias. This bias is then plotted for each quadratic pairing $X,Y\in\{ T, E,B\}$ in figures  \ref{fig4}, \ref{fig5} and \ref{fig6}.
The estimation bias is shown as `\textcolor{black}{$-\!\!\bullet\!\!-$}' with $1\sigma$ monte carlo error bars attached as a function of the fiducial uncertainty  $1-c$ in the bin $l\in[10,200]$.  The dashed line shows the standard deviation of  $l^4(\delta_0^{-1}  |\tilde \phi_\bl^\xy  |^2 -N^{(0)}_{\bl,\text{\tiny obs}})/4$ and the solid black line plots the fiducial error $(1-T_\bl) l^4C_\bl^{\phi\phi}/4$ averaged over $l\in[10,200]$.

\begin{figure}[ht]
\includegraphics[height = 2.8in]{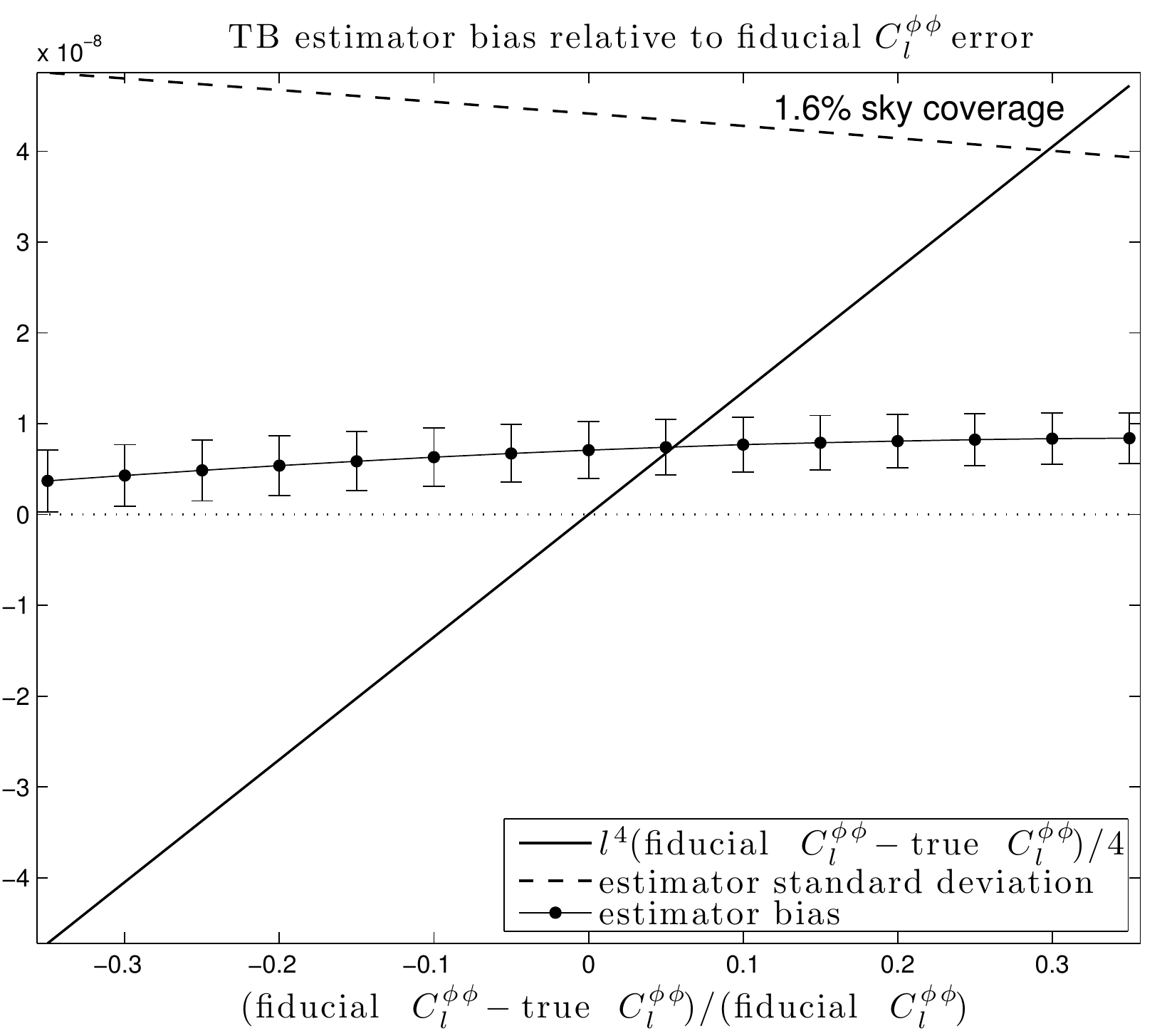}
\caption{
Spectral estimation bias averaged over the multiple bin $l\in[10,200]$ as a function of fiducial uncertainty. See Section \ref{sim} for details.  \label{fig6}
 }
\end{figure}

\subsubsection*{Does this effect future inference?}

Table \ref{tab} summarizes fiducial uncertainty (at 15\%),  estimation standard deviation, and bias range when using  $\tilde \phi_\bl^\xy$ to estimate the average power in  $ l^4C_\bl^{\phi\phi}/4$ over $l\in[10,200]$. 
Each row corresponds to a different quadratic pairing $X,Y\in\{T, E, B\}$. The second column shows the bias range corresponding to $15$\% fiducial uncertainty taken from figures \ref{fig4}, \ref{fig5} and \ref{fig6}. We list bias range, versus absolute bias,  since any baseline bias can be estimated with simulation and subsequently subtracted from any estimate. 
The third column shows $4\sigma$ full sky error bars where the estimation standard deviation is extrapolated from figures \ref{fig4}, \ref{fig5} and \ref{fig6} to full sky $\sigma$ by multiplying $\sqrt{f_\text{sky}}=0.126$. 
 Notice that, ignoring the bias, the $EE$ and $EB$ estimator have the power to constrain the fiducial uncertainty by a factor of $3$. However,  accounting for bias  this constraining power is mitigated, especially for estimators $EE$ and $TE$. In contrast, the $EB$ estimator bias can nearly ignored completely.

\begin{table}[h]
\begin{tabular}{c|c|c|c}
 $\tilde \phi_\bl^\xy$ &  $\sim$ bias range  &    $\sim$ full sky $4\sigma$ & fiducial error ($15\%$)\\
 \hline
TE &  $\pm 0.70$ &  $\pm 0.76$ & $\pm 2.02$ \\
EE &  $\pm 0.53$ & $\pm0.62$ & $\pm 2.02$ \\
TT &   $\pm 0.16$ & $\pm0.86$ &$\pm 2.02$ \\
TB &  $\pm 0.10$ & $\pm2.33$ &$\pm 2.02$ \\
EB &  $\pm 0.01$ & $\pm0.67$ & $\pm 2.02$
 \end{tabular}
\caption{ In units of spectral power, per $10^{-8}$, the above table shows bias range,  approximate full sky standard deviation and  fiducial error (at 15\%) for each quadratic estimate $\tilde \phi_\bl^{\xandy TE}$, $\tilde \phi_\bl^{\xandy EE}$, $\tilde \phi_\bl^{\xandy TT}$, $\tilde \phi_\bl^{\xandy TB}$ and $\tilde \phi_\bl^{\xandy EB}$.\label{tab}}
\end{table}

\section{Discussion}

 The state-of-the-art estimator of weak lensing, the quadratic estimator developed by Hu and Okomoto  \cite{Hu2001b, HuOka2002}, works in part through a delicate cancelation of terms in a Taylor expansion of the lensing effect on the CMB. In this paper we present two simulation based approaches for exploring the nature of this cancelation for both the CMB intensity and the polarization fields.
 In particular, we use these two simulation algorithms to 
 analyze the so called $N^{(1)}_\bl$ and $N^{(2)}_\bl$ bias for 
 two modifications of the  full set of quadratic temperature/polarization lensing estimators: one which  incorporates lensed rather than unlensed spectra into the estimator weights to mitigate the effect of higher order terms; 
 and  one which uses the observed lensed CMB fields to correct for the, so called, $N^{(0)}_\bl$ bias.
 Our simulation algorithms, which can simulate all higher order bias terms $N^{(j)}_\bl$ for $j\geq 0$, utilize an extension of the FFT techniques developed in  \cite{Hu2001b}. These FFT characterizations are key to making estimates of the $N^{(j)}_\bl$ fast and additionally provide fast (non-stochastic) algorithms for approximating the $N^{(0)}_\bl$ bias using the observed lensed CMB fields.

  In Section \ref{sim} we use our algorithm to analyze  the modified quadratic temperature/polarization lensing estimators for future ACTpol/SPTpol experiments. We find that the modified estimates  do reduce low $l$ estimation bias. However this is accomplished by effectively increasing the magnitude of $ N^{(1)}_{\bl}$ and  $ N^{(2)}_{\bl}$  to the point of cancelation when  the correct model for $C_\bl^{\phi\phi}$ is used to generate the lensed weights. We also demonstrate, through an analysis of estimator bias versus fiducial uncertainty,  that  this cancelation can be sensitive to the fiducial model for $C_\bl^{\phi\phi}$ depending on which estimator one uses: TE and EE are most sensitive, EB is least sensitive. For low $l$ estimation in future ACTpol/SPTpol experiments we conclude that the bias in the EB estimator can be effectively ignored. For the TE and the EE estimators, however, the bias does contribute significantly to projected error bars and may need to be corrected  to give the estimator inferential power beyond a nominal  fiducial uncertainty. 

\appendix
%
%

\section{FFT algorithms}
\label{fftappen}
In this section we derive the FFT algorithms  which allow fast simulation of the fields $\mathscr L^{(j)}_\bl$,  $\mathscr O^{(j)}_\bl$ and $\mathscr N^{(j)}_{\bl}$ as described in Section \ref{fft1}. 
 We begin with some notational conventions which greatly simplify the subsequent formulas. 
 First let $\varphi_{\bs k}$ denote the phase angle of frequency ${\bs k}$ and $\Delta \varphi \equiv \varphi_{\bk+\bl}-\varphi_{\bk}$. Also let  $l_a$ denote the $a^\text{th}$ coordinate of $\bl$.
 For any field $X\in\{T, E, B\}$ we let  
\begin{align*}
\hat X_\bk\equiv e^{i2\varphi_\bk} X_\bk, \qquad\check X_\bk\equiv e^{i4\varphi_\bk} X_\bk.
\end{align*}
Furthermore, we will utilize a super-script/sub-script notation to  denote multiplication/division by particular power spectra.
An example serves to illustrate the notation:
\begin{align*}
  \bigl(\check X_{\xandy TT} \bigr)_\bk &\equiv \check X_\bk \frac{1}{\widetilde C_{\bk,\text{\tiny exp}}^{\xandy TT}},\qquad
 \bigl(\check X_{\xandy TT, \xandy EE}^{\xandy TE}\bigr)_\bk \equiv \check X_\bk \frac{C_\bk^{\xandy TE}}{\widetilde C_{\bk,\text{\tiny exp}}^{\xandy TT}\widetilde C_{\bk,\text{\tiny exp}}^{\xandy EE}}\\
  \bigl(\check X_{\xandy TT}^{\xandy TE}\bigr)_\bk &\equiv \check X_\bk \frac{C_\bk^{\xandy TE}}{\widetilde C_{\bk,\text{\tiny exp}}^{\xandy TT}},
 \qquad
  \bigl(\check X_{\xandy TT, \xandy EE}^{\xandy TE,\xandy TT}\bigr)_\bk \equiv\check X_\bk \frac{C_\bk^{\xandy TE} C_\bk^{\xandy TT}}{\widetilde C_{\bk,\text{\tiny exp}}^{\xandy TT}\widetilde C_{\bk,\text{\tiny exp}}^{\xandy EE}}.
  \end{align*}
Notice that the above denominators always use lensed spectra with experimental noise, whereas the numerators always use unlensed spectra. In doing so, the  formulas found in claims 1 through 5 below can be used for fast algorithms for the quadratic estimate $\hat\phi_\bl^\xy$. 
To obtain the corresponding formulas for the modified quadratic estimate $\tilde \phi^\xy_\bl$ one simply needs to replace the unlensed spectra in the numerator of the above notation, with the lensed spectra (but without experimental noise).


\begin{claim}[TT estimator]
  If $X_\bk$ and $Y_\bk$ are complex functions such that $X_{-\bk}=X^*_\bk$ and $Y_{-\bk}=Y^*_\bk$ then 
  \begin{align*}
 \int \frac{d^2\bs k}{2\pi} g^{\xandy{T}{T}}_{\bs l,\bs k\,}   X_{\bs k+\bs l}  Y^*_{\bs k}  &=  -i  l_a  \int \frac{d^2\bx}{2\pi} e^{-i\bl\cdot\bx} \\
  &\times\Bigl(  [\nabla^aX_{\xandy TT}^{\xandy TT}(\bx)][Y_{\xandy TT}(\bx)] \\
  &\quad+   [\nabla^aY_{\xandy TT}^{\xandy TT}(\bx)][X_{\xandy TT}(\bx)]  \Bigr) ;
 \end{align*}
 \begin{align*}
 \int \frac{d^2\bk}{(2\pi)^2}  
   |g^\text{\tiny T\!T}_{\bl,\bk\,}|^2
   X_{\bk+\bl} Y^*_{\bk}
&= -\frac{ l_a l_b}{2\pi} \int \frac{d^2 \bx}{2\pi}  e^{-i\bl\cdot\bx} \\
&\times \Bigl( 2[\nabla ^a   X^{\xandy TT}_{\xandy TT, \xandy TT} (\bx)][\nabla ^b  Y^{\xandy TT}_{\xandy TT, \xandy TT} (\bx)]\nonumber \\
 &\quad+  [X^{}_{\xandy TT, \xandy TT}(\bx)] [\nabla^a\nabla ^b Y^{\xandy TT,\xandy TT}_{\xandy TT, \xandy TT}(\bx)]\nonumber\\
 &\quad+ [Y^{}_{\xandy TT, \xandy TT}(\bx)] [\nabla^a\nabla ^b  X^{\xandy TT,\xandy TT}_{\xandy TT, \xandy TT}(\bx)]  \Bigr) ;\nonumber
\end{align*}
\begin{align*}
 \int \frac{d^2 \bk}{(2\pi)^2}  g^\text{\tiny $T\!T$}_{\bl,\bk\,}   g^\text{\tiny $T\!T$}_{\bl,-\bk-\bl\,}   X_{\bk+\bl} Y^*_{\bk} = \int \frac{d^2\bk}{(2\pi)^2}  
   |g^\text{\tiny T\!T}_{\bl,\bk\,}|^2
  X_{\bk+\bl} Y^*_{\bk},
\end{align*}
where $g^{\xandy TT}_{\bl,\bk\,} \equiv 2\pi{f^{\xandy TT}_{\bl, \bk}}[\widetilde C^{\xandy TT}_{\bk+\bl,\text{\tiny exp}} \widetilde C^{\xandy TT}_{\bk,\text{\tiny exp}} ]^{-1} $ and  $f^\text{\tiny $TT$}_{\bl, \bk}  \equiv \frac{1}{2\pi} [\bl \cdot (\bk + \bl) C^\text{\tiny $TT$}_{\bk+\bl}  -\bl \cdot \bk  C^\text{\tiny $TT$}_{\bk}]$.
 \end{claim}

\if\Ver\LongVer{ 
{\flushleft\textcolor{blue}{$\downarrow$---------begin long version---------}}\newline
\onecolumngrid

\textcolor{blue}{One thing to note is that all the lensed spectra in the following proofs should be changed to lensed spectra with experimental noise.}
\begin{proof}
For the estimator estimator $\hat\phi^{\xandy{T}{T}}$ notice that $f^\text{\tiny $TT$}_{\bl, \bk}  = \frac{1}{2\pi} [\bl \cdot (\bk + \bl) C^\text{\tiny $TT$}_{\bk+\bl}  -\bl \cdot \bk  C^\text{\tiny $TT$}_{\bk}]$. Since $\gw_{\bs l, \bs k} = 2\pi{f^\xy_{\bl,\bk}}[\widetilde C^\text{\tiny X\!X}_{\bk+\bl,\text{\tiny exp}} \widetilde C^\text{\tiny $Y\!Y$}_{\bk,\text{\tiny exp}} ]^{-1} $ we have that
\begin{align*}
 \int \frac{d^2\bs k}{2\pi} g^{\xandy{T}{T}}_{\bs l, \bs k}   X_{\bs k+\bs l}  Y^*_{\bs k} & =    \int \frac{d^2\bs k}{2\pi}   \Bigl[ \bl\cdot (\bk + \bl) C^\text{\tiny $TT$}_{\bk+\bl}  - \bl\cdot \bk  C^\text{\tiny $TT$}_{\bk}\Bigr]  \frac{ X_{\bs k+\bs l}}{\widetilde C^\text{\tiny T\!T}_{\bk+\bl,\text{\tiny exp}}}  \frac{Y^*_{\bs k} }{\widetilde C^\text{\tiny $T\!T$}_{\bk,\text{\tiny exp}}} \\
 &=  -i l_a \int \frac{d^2\bs k}{2\pi}   \Bigl[ i (k^a + l^a) C^\text{\tiny $TT$}_{\bk+\bl}  - i k^a  C^\text{\tiny $T\!T$}_{\bk}\Bigr]  \frac{ X_{\bs k+\bs l}}{\widetilde C^\text{\tiny T\!T}_{\bk+\bl,\text{\tiny exp}}}  \frac{Y^*_{\bs k} }{\widetilde C^\text{\tiny $T\!T$}_{\bk,\text{\tiny exp}}}.
 \end{align*}
 Now
\[
\int \frac{d^2\bs k}{2\pi} \underbrace{\bigl[i(k^a + l^a)C^\text{\tiny TT}_{\bk+\bl}  \bigr] \frac{ X_{\bs k +\bs l}}{\widetilde C^\text{\tiny $T\!T$}_{\bk+\bl,\text{\tiny exp}}}}_{ [\nabla^a X_{\xandy TT}^{\xandy TT}]_{\bk+\bl}} \underbrace{\frac{ Y^*_\bk}{\widetilde C^\text{\tiny $T\!T$}_{\bk,\text{\tiny exp}}}}_{ [Y_{\xandy TT}]^*_\bk}= \int \frac{d^2\bx}{2\pi} e^{-i\bl\cdot\bx}  [\nabla^a X_{\xandy TT}^{\xandy TT}(\bx)] [Y_{\xandy TT}(\bx)].
\]
which follows from Lemma \ref{LEM1}. Also notice 
\[
\int \frac{d^2\bs k}{2\pi} \underbrace{ \frac{ X_{\bs k +\bs l}}{\widetilde C^\text{\tiny $T\!T$}_{\bk+\bl,\text{\tiny exp}}}}_{  [ X_{\xandy TT}]_{\bk+\bl}} \underbrace{\bigl[ik^a C^\text{\tiny $T\!T$}_{\bk}  \bigr]^*\frac{ Y^*_\bk}{\widetilde C^\text{\tiny $T\!T$}_{\bk,\text{\tiny exp}}}}_{  [\nabla^a Y_{\xandy TT}^{\xandy TT}]^*_\bk}= \int \frac{d^2\bx}{2\pi} e^{-i\bl\cdot\bx} [\nabla^a Y_{\xandy TT}^{\xandy TT}(\bx)]  [ X_{\xandy TT}(\bx)].
\]
which follows from Lemma \ref{LEM1}.

To finish notice that
\begin{align*}
\int \frac{d^2\bk}{(2\pi)^2}  
   |g^\text{\tiny T\!T}_{\bl,\bk}|^2
  X_{\bk+\bl} Y^*_{\bk}
 &=  \int \frac{d^2 \bk}{(2\pi)^2}  \frac{ \bigl[\bl\cdot(\bk + \bl)C^\text{\tiny $T\!T$}_{\bk+\bl} - \bl\cdot\bk C^\text{\tiny $T\!T$}_\bk  \bigr]^2}{[\widetilde C^\text{\tiny T\!T}_{\bk+\bl,\text{\tiny exp}}]^2 [\widetilde C^\text{\tiny $T\!T$}_{\bk,\text{\tiny exp}}]^2}{X_{\bk+\bl}}{Y^*_\bk}\\
&=  - l_a l_b \int \frac{d^2 \bk}{(2\pi)^2}  
\underbrace{\frac{i(k^a+l^a) i(k^b+l^b) [C^\text{\tiny $T\!T$}_{\bk+\bl}]^2 X_{\bk +\bl}}{[\widetilde C^\text{\tiny T\!T}_{\bk+\bl,\text{\tiny exp}}]^2 } }_{ [\nabla^a\nabla ^b  X^{\xandy TT,\xandy TT}_{\xandy TT, \xandy TT}]_{\bk+\bl}}
\underbrace{\frac{Y^*_\bk}{ [\widetilde C^\text{\tiny $T\!T$}_{\bk,\text{\tiny exp}}]^2}}_{ [ Y^{}_{\xandy TT, \xandy TT}]^*_\bk}
-l_a l_b \int \frac{d^2 \bk}{(2\pi)^2}
\underbrace{\frac{X_{\bk+\bl}}{[\widetilde C^\text{\tiny T\!T}_{\bk+\bl,\text{\tiny exp}}]^2}}_{  [X^{}_{\xandy TT, \xandy TT}]_{\bk+\bl}}  
\underbrace{\frac{  ik^a ik^b \bigl[ C^\text{\tiny $T\!T$}_\bk  \bigr]^2Y^*_\bk}{ [\widetilde C^\text{\tiny $T\!T$}_{\bk,\text{\tiny exp}}]^2 }}_{  [\nabla^a\nabla ^b  Y^{\xandy TT,\xandy TT}_{\xandy TT, \xandy TT}]^*_{\bk}} \\
&\qquad\qquad\qquad\qquad
-2 l_a l_b \int \frac{d^2 \bk}{(2\pi)^2} 
\underbrace{\frac{ i(k^a + l^a)C^\text{\tiny $T\!T$}_{\bk+\bl} X_{\bk+\bl}}{[\widetilde C^\text{\tiny T\!T}_{\bk+\bl,\text{\tiny exp}}]^2}}_{  [\nabla^a  X^{\xandy TT}_{\xandy TT, \xandy TT}]_{\bk+\bl}}
\underbrace{\frac{ - ik^b C^\text{\tiny $T\!T$}_\bk Y^*_\bk }{ [\widetilde C^\text{\tiny $T\!T$}_{\bk,\text{\tiny exp}}]^2}} _{  [\nabla^b  Y^{\xandy TT}_{\xandy TT, \xandy TT}]^*_{\bk}}.
\end{align*}
QED.
\end{proof}
{\flushleft\textcolor{blue}{$\uparrow$------------end long version---------}}\newline
} \fi


\begin{claim}[TE estimator]
  If $X_\bk$ and $Y_\bk$ are complex functions such that $X_{-\bk}=X^*_\bk$ and $Y_{-\bk}=Y^*_\bk$ then 
\begin{align*}
 \int \frac{d^2\bs k}{2\pi} g^{\xandy{T}{E}}_{\bs l,\bs k\,}   X_{\bs k+\bs l}  Y^*_{\bs k}  &=  -i  l_a   \int \frac{d^2\bx}{2\pi} e^{-i\bl\cdot\bx} \\
 &\times \Bigl(  Re\{ [\nabla^a \hat X^{\xandy TE}_{\xandy TT}(\bx) ]  [\hat Y^{}_{\xandy EE}(\bx) ]^*\} \\
 &\quad+ [\nabla^a Y^{\xandy TE}_{\xandy EE}(\bx) ]  [X^{}_{\xandy TT}(\bx) ]   \Bigr);
 \end{align*}
 \begin{align*}
 \int \frac{d^2\bk}{(2\pi)^2}  
  & |g^\text{\xandy{T}{E}}_{\bl,\bk\,}|^2
   X_{\bk+\bl} Y^*_{\bk}
= -\frac{ l_a l_b}{4\pi} \int \frac{d^2 \bx}{2\pi}  e^{-i\bl\cdot\bx} \\
&\times\Bigl( 
[ Y^{}_{\xandy EE, \xandy EE}(\bx) ]
[\nabla^a \nabla^b  X^{\xandy TE, \xandy TE }_{\xandy TT, \xandy TT}(\bx) ]\nonumber\\
&\quad +Re\{
 [\check Y^{ }_{\xandy EE, \xandy EE}(\bx) ]
 [\nabla^a \nabla^b \check X^{\xandy TE, \xandy TE }_{\xandy TT, \xandy TT}(\bx) ]^* \}\nonumber\\
 &\quad+ 2 
 [  X^{}_{\xandy TT, \xandy TT}(\bx) ]
 [\nabla^a \nabla^b Y^{\xandy TE, \xandy TE }_{\xandy EE, \xandy EE}(\bx) ] \nonumber\\
 &\quad+ Re\{ 4 
 [\nabla^a \hat X^{ \xandy TE }_{\xandy TT, \xandy TT}(\bx) ]
  [\nabla^b \hat Y^{ \xandy TE }_{\xandy EE, \xandy EE}(\bx) ]^*
 \}\Bigr);
\end{align*}
\begin{align*}
 \int \frac{d^2 \bk}{(2\pi)^2}  g^\text{\tiny $T\!E$}_{\bl,\bk\,}  & g^\text{\tiny $T\!E$}_{\bl,-\bk-\bl\,}   X_{\bk+\bl} Y^*_{\bk} = -\frac{ l_a l_b}{4\pi} \int \frac{d^2 \bx}{2\pi}  e^{-i\bl\cdot\bx}\\
 &\times \Bigl(Re\{ 2 
 [\nabla^a \nabla^b \hat X^{\xandy TE, \xandy TE }_{\xandy TT, \xandy EE}(\bx) ]
[\hat Y^{ }_{\xandy TT, \xandy EE}(\bx) ]^*\} \nonumber \\
 &\quad+Re\{2 
  [\nabla^a \nabla^b \hat Y^{\xandy TE, \xandy TE }_{\xandy TT, \xandy EE}(\bx) ]
   [ \hat X^{ }_{\xandy TT, \xandy EE}(\bx) ]^*
  \}\nonumber\\
 &\quad+  3 
 [\nabla^a X^{ \xandy TE }_{\xandy TT, \xandy EE}(\bx) ]
 [ \nabla^b  Y^{ \xandy TE }_{\xandy TT, \xandy EE}(\bx) ]\nonumber\\
 &\quad+ Re\{
 [\nabla^b \check Y^{ \xandy TE }_{\xandy TT, \xandy EE}(\bx) ]
 [\nabla^a \check X^{\xandy TE }_{\xandy TT, \xandy EE}(\bx) ] ^* \}  \Bigr),
\end{align*}
where $g^{\xandy TE}_{\bl,\bk\,} \equiv 2\pi{f^{\xandy TE}_{\bl, \bk}}[\widetilde C^{\xandy TT}_{\bk+\bl,\text{\tiny exp}} \widetilde C^{\xandy EE}_{\bk,\text{\tiny exp}} ]^{-1} $ and  $f^\text{\tiny $T\!E$}_{\bl, \bk}   \equiv \frac{1}{2\pi} [\bl \cdot (\bk + \bl)  \cos(2\Delta\varphi) C^\text{\tiny $T\!E$}_{\bk+\bl}  -\bl \cdot \bk  C^\text{\tiny $T\!E$}_{\bk}]$.
 \end{claim}

\if\Ver\LongVer{ 
{\flushleft\textcolor{blue}{$\downarrow$---------begin long version---------}}\newline

\begin{proof}
For the estimator estimator $\hat\phi^{\xandy{T}{E}}$ notice that if we set $\Delta \varphi \equiv \varphi_{\bk+\bl}-\varphi_{\bk}$ then $f^\text{\tiny $T\!E$}_{\bl, \bk}   = \frac{1}{2\pi} [\bl \cdot (\bk + \bl)  \cos(2\Delta \varphi) C^\text{\tiny $T\!E$}_{\bk+\bl}  -\bl \cdot \bk  C^\text{\tiny $T\!E$}_{\bk}]$. Since $\gw_{\bs l, \bs k} = 2\pi{f^\xy_{\bl,\bk}}[\widetilde C^\text{\tiny X\!X}_{\bk+\bl,\text{\tiny exp}} \widetilde C^\text{\tiny $Y\!Y$}_{\bk,\text{\tiny exp}} ]^{-1} $ we have that
\begin{align*}
 \int \frac{d^2\bs k}{2\pi} g^{\xandy{T}{E}}_{\bs l, \bs k}   X_{\bs k+\bs l}  Y^*_{\bs k} & =    \int \frac{d^2\bs k}{2\pi}   \Bigl[\bl \cdot (\bk + \bl)  \cos(2\Delta \varphi) C^\text{\tiny $T\!E$}_{\bk+\bl}  -\bl \cdot \bk  C^\text{\tiny $T\!E$}_{\bk}\Bigr]  \frac{ X_{\bs k+\bs l}}{\widetilde C^\text{\tiny T\!T}_{\bk+\bl,\text{\tiny exp}}}  \frac{Y^*_{\bs k} }{\widetilde C^\text{\tiny $E\!E$}_{\bk,\text{\tiny exp}}} \\
 &=  -i l_a \int \frac{d^2\bs k}{2\pi}   \Bigl[i (k^a + l^a)  \cos(2\Delta \varphi) C^\text{\tiny $T\!E$}_{\bk+\bl}  - i k^a  C^\text{\tiny $T\!E$}_{\bk}\Bigr] \frac{ X_{\bs k+\bs l}}{\widetilde C^\text{\tiny T\!T}_{\bk+\bl,\text{\tiny exp}}}  \frac{Y^*_{\bs k} }{\widetilde C^\text{\tiny $E\!E$}_{\bk\text{\tiny exp}}}.
 \end{align*}
First notice

\[
 -i l_a\int \frac{d^2\bs k}{2\pi} \cos(2\Delta\varphi)
 \underbrace{   i(k^a + l^a)C^\text{\tiny $T\!E$}_{\bk+\bl}  \frac{X_{\bs k +\bs l}}{\widetilde C^\text{\tiny $T\!T$}_{\bk+\bl,\text{\tiny exp}}}}
 _{  [\nabla^a  X^{ \xandy TE }_{\xandy TT}]_{\bk+\bl}} 
 \underbrace{\frac{Y^*_\bk}{\widetilde C^\text{\tiny $E\!E$}_{\bk,\text{\tiny exp}}}}
 _{ [ Y^{}_{\xandy EE}]^*_\bk}
 =  -i l_a\int \frac{d^2\bx}{2\pi} e^{-i\bl\cdot\bx} Re\bigl\{  [\nabla^a \hat X^{ \xandy TE }_{\xandy TT}(\bx)]  [ \hat Y^{}_{\xandy EE}(\bx)]^* \bigr\}
\]
 by Lemma \ref{LEM2}. Also notice that 
\[
 -i l_a\int \frac{d^2\bs k}{2\pi}
 \underbrace{ \frac{X_{\bs k +\bs l}}{\widetilde C^\text{\tiny $T\!T$}_{\bk+\bl,\text{\tiny exp}}}}
 _{ [X^{}_{\xandy TT}]_{\bk+\bl}}  
 \underbrace{[  ik^a C^\text{\tiny $T\!E$}_\bk  ]^*\frac{Y^*_\bk}{\widetilde C^\text{\tiny $E\!E$}_{\bk,\text{\tiny exp}}}}
 _{  [\nabla^a  Y^{ \xandy TE }_{\xandy EE}]^*_\bk}=  -i l_a\int \frac{d^2\bx}{2\pi} e^{-i\bl\cdot\bx} [X^{}_{\xandy TT}(\bx)] [\nabla^a  Y^{ \xandy TE }_{\xandy EE}(\bx)] 
\]
which follows by Lemma \ref{LEM1}.

To show the next equality notice that  TE notice that $\int \frac{d^2\bk}{(2\pi)^2}  
   |g^\text{\tiny T\!E}_{\bl,\bk}|^2
  X_{\bk+\bl} Y^*_{\bk} =I+ II + III$ where
\begin{align*}
I &= \int \frac{d^2 \bk}{(2\pi)^2}  \frac{[\bl\cdot (\bk + \bl)C^\text{\tiny TE}_{\bk+\bl} ]^2 \cos^2(2\Delta \varphi)}{[\widetilde C^\text{\tiny TT}_{\bk+\bl,\text{\tiny exp}}]^2 [\widetilde C^\text{\tiny EE}_{\bk,\text{\tiny exp}}]^2}  X_{\bk+\bl} Y^*_{\bk}\\
&= -\frac{l_a l_b}{2\pi} \int \frac{d^2 \bk}{2\pi}
 \underbrace{\frac{i(k^a + l^a)i(k^b + l^b)[C^\text{\tiny TE}_{\bk+\bl} ]^2X_{\bk+\bl}}{[\widetilde C^\text{\tiny TT}_{\bk+\bl,\text{\tiny exp}}]^2}}
 _{ [\nabla^a \nabla^b X^{\xandy TE, \xandy TE }_{\xandy TT, \xandy TT} ]_{\bk+\bl}}
 \underbrace{\frac{ Y^*_\bk }{[\widetilde C^\text{\tiny EE}_{\bk,\text{\tiny exp}}]^2}}
 _{ [  Y_{\xandy EE, \xandy EE}]^*_\bk}
  \cos^2(2\Delta \varphi) \\
&=  -\frac{l_a l_b}{4\pi}\int \frac{d^2\bx}{2\pi} e^{-i\bl\cdot\bx}\left(
[  Y^{}_{\xandy EE, \xandy EE} (\bx)][\nabla^a \nabla^b X^{\xandy TE, \xandy TE }_{\xandy TT, \xandy TT} (\bx)]
 +Re\{[ \check Y^{}_{\xandy EE, \xandy EE} (\bx)][\nabla^a \nabla^b \check X^{\xandy TE, \xandy TE }_{\xandy TT, \xandy TT} (\bx)]^* \}\right) 
\end{align*}
also
\begin{align*}
II &= \int \frac{d^2 \bk}{(2\pi)^2}  \frac{[\bl\cdot \bk C^\text{\tiny T\!E}_{\bk} ]^2 }{[\widetilde C^\text{\tiny TT}_{\bk+\bl,\text{\tiny exp}}]^2 [\widetilde C^\text{\tiny EE}_{\bk,\text{\tiny exp}}]^2}  X_{\bk+\bl} Y^*_{\bk}\\
&= - \frac{l_a l_b}{2\pi} \int \frac{d^2 \bk}{2\pi}
 \underbrace{\frac{X_{\bk+\bl}}{[\widetilde C^\text{\tiny TT}_{\bk+\bl,\text{\tiny exp}}]^2}}
 _{[ X^{}_{\xandy TT, \xandy TT} ]_{\bk+\bl}}
 \underbrace{\frac{  ik^a ik^b \bigl[ C^\text{\tiny TE}_\bk  \bigr]^2  Y^*_\bk }{[\widetilde C^\text{\tiny EE}_{\bk,\text{\tiny exp}}]^2}}
 _{[\nabla^a \nabla^b  Y^{\xandy TE, \xandy TE }_{\xandy EE, \xandy EE} ]^*_\bk} \\
&= -\frac{l_a l_b}{4\pi} \int \frac{d^2\bx}{2\pi} e^{-i\bl\cdot\bx}\left( 2 [ X^{}_{\xandy TT, \xandy TT} (\bx)][\nabla^a \nabla^b  Y^{\xandy TE, \xandy TE }_{\xandy EE, \xandy EE}(\bx)] \right) 
\end{align*}
and

\begin{align*}
III &= \int \frac{d^2 \bk}{(2\pi)^2}  \frac{-2[\bl\cdot \bk C^\text{\tiny TE}_{\bk} ][\bl\cdot (\bk + \bl)C^\text{\tiny TE}_{\bk+\bl} ]\cos(2\Delta\varphi) }{[\widetilde C^\text{\tiny TT}_{\bk+\bl,\text{\tiny exp}}]^2 [\widetilde C^\text{\tiny EE}_{\bk,\text{\tiny exp}}]^2}  X_{\bk+\bl} Y^*_{\bk}\\
&=  -\frac{2l_al_b}{2\pi}\int \frac{d^2 \bk}{2\pi}
 \underbrace{\frac{i(k^a + l^a)C^\text{\tiny TE}_{\bk+\bl} X_{\bk+\bl}}{[\widetilde C^\text{\tiny TT}_{\bk+\bl,\text{\tiny exp}}]^2}}
 _{[\nabla^a  X^{ \xandy TE }_{\xandy TT, \xandy TT} ]_{\bk+\bl}}
 \underbrace{\frac{ -ik^b C^\text{\tiny TE}_{\bk}  Y^*_\bk }{[\widetilde C^\text{\tiny EE}_{\bk,\text{\tiny exp}}]^2}}
 _{[\nabla^b  Y^{\xandy TE }_{\xandy EE, \xandy EE} ] ^*_\bk} \cos(2\Delta\varphi) \\
&=  -\frac{l_al_b}{4\pi} \int \frac{d^2\bx}{2\pi} e^{-i\bl\cdot\bx}Re\left\{ 4 [\nabla^a \hat X^{ \xandy TE }_{\xandy TT, \xandy TT} (\bx)][\nabla^b  \hat Y^{\xandy TE }_{\xandy EE, \xandy EE}(\bx)]^* \right\} \\
\end{align*}

Finally notice that 
To show the results for TE notice that $\int \frac{d^2\bk}{(2\pi)^2}  
   g^\text{\tiny T\!E}_{\bl,\bk\,}g^\text{\tiny T\!E}_{\bl,-\bk-\bl }
  X_{\bk+\bl} Y^*_{\bk} = I+ II + III+ IV$ where
\begin{align*}
I &= \int \frac{d^2 \bk}{(2\pi)^2}  \frac{[\bl\cdot (\bk + \bl)C^\text{\tiny TE}_{\bk+\bl} ]^2 \cos(2\Delta\varphi)}{\widetilde C^\text{\tiny TT}_{\bk+\bl,\text{\tiny exp}} \widetilde C^\text{\tiny EE}_{\bk+\bl,\text{\tiny exp}}\widetilde C^\text{\tiny TT}_{\bk,\text{\tiny exp}} \widetilde C^\text{\tiny EE}_{\bk}}  X_{\bk+\bl,\text{\tiny exp}} Y^*_{\bk}\\
&= - l_a l_b \int \frac{d^2 \bk}{(2\pi)^2}
 \underbrace{\frac{i(k^a + l^a)i(k^b + l^b)[C^\text{\tiny TE}_{\bk+\bl} ]^2X_{\bk+\bl}}{\widetilde C^\text{\tiny TT}_{\bk+\bl,\text{\tiny exp}} \widetilde C^\text{\tiny EE}_{\bk+\bl,\text{\tiny exp}}}}
 _{[\nabla^a \nabla^b X^{\xandy TE, \xandy TE }_{\xandy TT, \xandy EE} ]_{\bk+\bl}}
 \underbrace{\frac{ Y^*_\bk }{\widetilde C^\text{\tiny TT}_{\bk} \widetilde C^\text{\tiny EE}_{\bk}}}
 _{[ Y^{ }_{\xandy TT, \xandy EE} ]^*_\bk} \cos(2\Delta\varphi) \\
&= -\frac{l_a l_b}{4\pi} \int \frac{d^2\bx}{2\pi} e^{-i\bl\cdot\bx}Re\{ 2 
[\nabla^a \nabla^b\hat  X^{\xandy TE, \xandy TE }_{\xandy TT, \xandy EE}(\bx)]
[ \hat Y^{ }_{\xandy TT, \xandy EE}(\bx)]^* \}
\end{align*}
For the second term
\begin{align*}
II &= \int \frac{d^2 \bk}{(2\pi)^2}  \frac{[\bl\cdot \bk C^\text{\tiny TE}_{\bk} ]^2 \cos(2\Delta\varphi)}{\widetilde C^\text{\tiny TT}_{\bk+\bl,\text{\tiny exp}} \widetilde C^\text{\tiny EE}_{\bk+\bl,\text{\tiny exp}}\widetilde C^\text{\tiny TT}_{\bk,\text{\tiny exp}} \widetilde C^\text{\tiny EE}_{\bk,\text{\tiny exp}}}  X_{\bk+\bl} Y^*_{\bk}\\
&= - l_a l_b \int \frac{d^2 \bk}{(2\pi)^2}
 \underbrace{\frac{X_{\bk+\bl}}{\widetilde C^\text{\tiny TT}_{\bk+\bl,\text{\tiny exp}} \widetilde C^\text{\tiny EE}_{\bk+\bl,\text{\tiny exp}}}}
 _{[X^{}_{\xandy TT, \xandy EE} ]_{\bk+\bl}}
 \underbrace{\frac{ i(k^a)i(k^b)[C^\text{\tiny TE}_{\bk} ]^2Y^*_\bk }{\widetilde C^\text{\tiny TT}_{\bk,\text{\tiny exp}} \widetilde C^\text{\tiny EE}_{\bk,\text{\tiny exp}}}}
 _{[\nabla^a \nabla^b  Y^{\xandy TE, \xandy TE }_{\xandy TT, \xandy EE} ]^*_\bk} \cos(2\Delta\varphi) \\
&= -\frac{l_a l_b}{4\pi} \int \frac{d^2\bx}{2\pi} e^{-i\bl\cdot\bx}Re\{2 
[\hat X^{}_{\xandy TT, \xandy EE}(\bx)]
[\nabla^a \nabla^b \hat Y^{\xandy TE, \xandy TE }_{\xandy TT, \xandy EE} (\bx)]^* \}.
\end{align*}
Now for the term $III$ we have
\begin{align*}
III &= -\int \frac{d^2 \bk}{(2\pi)^2}  \frac{[\bl\cdot \bk C^\text{\tiny TE}_{\bk} ][\bl\cdot (\bk + \bl)C^\text{\tiny TE}_{\bk+\bl} ] }{\widetilde C^\text{\tiny TT}_{\bk+\bl,\text{\tiny exp}} \widetilde C^\text{\tiny EE}_{\bk+\bl,\text{\tiny exp}}\widetilde C^\text{\tiny TT}_{\bk} \widetilde C^\text{\tiny EE}_{\bk,\text{\tiny exp}}}  X_{\bk+\bl} Y^*_{\bk}\\
&= - l_a l_b \int \frac{d^2 \bk}{(2\pi)^2}
 \underbrace{\frac{i(k^a+l^a)[C^\text{\tiny TE}_{\bk+\bl} ]X_{\bk+\bl}}{\widetilde C^\text{\tiny TT}_{\bk+\bl,\text{\tiny exp}} \widetilde C^\text{\tiny EE}_{\bk+\bl,\text{\tiny exp}}}}
 _{[\nabla^a  X^{ \xandy TE }_{\xandy TT, \xandy EE} ]_{\bk+\bl}}
 \underbrace{\frac{ -i(k^b)[C^\text{\tiny TE}_{\bk} ]Y^*_\bk }{\widetilde C^\text{\tiny TT}_{\bk,\text{\tiny exp}} \widetilde C^\text{\tiny EE}_{\bk,\text{\tiny exp}}}}
 _{[ \nabla^b  Y^{ \xandy TE }_{\xandy TT, \xandy EE} ]^*_\bk}  \\
&= -\frac{l_a l_b}{4\pi} \int \frac{d^2\bx}{2\pi} e^{-i\bl\cdot\bx}\left( 2[\nabla^a  X^{ \xandy TE }_{\xandy TT, \xandy EE} (\bx)][ \nabla^b  Y^{ \xandy TE }_{\xandy TT, \xandy EE}(\bx)] \right) 
\end{align*}
Finally for the last term $IV$ we have
\begin{align*}
IV &= -\int \frac{d^2 \bk}{(2\pi)^2}  \frac{[\bl\cdot \bk C^\text{\tiny TE}_{\bk} ][\bl\cdot (\bk + \bl)C^\text{\tiny TE}_{\bk+\bl} ] \cos^2(2\Delta\varphi )}{\widetilde C^\text{\tiny TT}_{\bk+\bl,\text{\tiny exp}} \widetilde C^\text{\tiny EE}_{\bk+\bl,\text{\tiny exp}}\widetilde C^\text{\tiny TT}_{\bk,\text{\tiny exp}} \widetilde C^\text{\tiny EE}_{\bk,\text{\tiny exp}}}  X_{\bk+\bl,\text{\tiny exp}} Y^*_{\bk}\\
&= - l_a l_b \int \frac{d^2 \bk}{(2\pi)^2}
 \underbrace{\frac{i(k^a+l^a)[C^\text{\tiny TE}_{\bk+\bl} ]X_{\bk+\bl}}{\widetilde C^\text{\tiny TT}_{\bk+\bl,\text{\tiny exp}} \widetilde C^\text{\tiny EE}_{\bk+\bl,\text{\tiny exp}}}}
 _{[\nabla^a  X^{ \xandy TE }_{\xandy TT, \xandy EE} ]_{\bk+\bl}}
 \underbrace{\frac{ -i(k^b)[C^\text{\tiny TE}_{\bk} ]Y^*_\bk }{\widetilde C^\text{\tiny TT}_{\bk,\text{\tiny exp}} \widetilde C^\text{\tiny EE}_{\bk,\text{\tiny exp}}}}
 _{[ \nabla^b  Y^{ \xandy TE }_{\xandy TT, \xandy EE} ]^*_\bk} \cos^2(2\Delta\varphi)  \\
&= -\frac{l_a l_b}{4\pi} \int \frac{d^2\bx}{2\pi} e^{-i\bl\cdot\bx}\left( [\nabla^a  X^{ \xandy TE }_{\xandy TT, \xandy EE} (\bx)][ \nabla^b  Y^{ \xandy TE }_{\xandy TT, \xandy EE}(\bx)]+ Re\{[\nabla^a  \check X^{ \xandy TE }_{\xandy TT, \xandy EE} (\bx)][ \nabla^b  \check Y^{ \xandy TE }_{\xandy TT, \xandy EE}(\bx)]^* \} \right).
\end{align*}
QED.
 
\end{proof}
{\flushleft\textcolor{blue}{$\uparrow$------------end long version---------}}\newline
} \fi


\begin{claim}[TB estimator]
  If $X_\bk$ and $Y_\bk$ are complex functions such that $X_{-\bk}=X^*_\bk$ and $Y_{-\bk}=Y^*_\bk$ then
  \begin{align*}
 \int \frac{d^2\bs k}{2\pi} g^{\xandy{T}{B}}_{\bs l,\bs k\,}   X_{\bs k+\bs l}  Y^*_{\bs k} & =  -i  l_a   \int \frac{d^2\bx}{2\pi} e^{-i\bl\cdot\bx} \\
 &\quad\times Im\{ [\nabla^a\hat  X^{\xandy TE}_{\xandy TT}(\bx)] [\hat Y_{\xandy BB}(\bx)]^*\};
 \end{align*}
 \begin{align*}
 \int \frac{d^2\bk}{(2\pi)^2}  
   &|g^{\xandy TB}_{\bl,\bk\,}|^2
   X_{\bk+\bl} Y^*_{\bk}
= -\frac{ l_a l_b}{4\pi} \int \frac{d^2 \bx}{2\pi}  e^{-i\bl\cdot\bx}\\
&\times \Bigl(  [Y_{\xandy BB, \xandy BB}(\bx)] [\nabla^a\nabla^b X^{\xandy TE, \xandy TE}_{\xandy TT, \xandy TT}(\bx)] \nonumber\\
&\quad -Re\{ [\check Y_{\xandy BB, \xandy BB}(\bx)] [\nabla^a\nabla^b \check X^{\xandy TE, \xandy TE}_{\xandy TT, \xandy TT}(\bx)]^* \} \Bigr);
\end{align*}
\begin{align*}
 \int \frac{d^2 \bk}{(2\pi)^2}  g^{\xandy TB}_{\bl,\bk\,} &  g^{\xandy  T B}_{\bl,-\bk-\bl\,}   X_{\bk+\bl} Y^*_{\bk} = -\frac{ l_a l_b}{4\pi} \int \frac{d^2 \bx}{2\pi}  e^{-i\bl\cdot\bx} \\
 &\times\Bigl(  [\nabla^a X^{\xandy TE}_{\xandy TT, \xandy BB}(\bx)][\nabla^b Y^{\xandy TE}_{\xandy TT, \xandy BB}(\bx)]\nonumber \\
 &\quad-Re\{[\nabla^a  \check X^{\xandy TE}_{\xandy TT, \xandy BB}(\bx)][\nabla^b\check Y^{\xandy TE}_{\xandy TT, \xandy BB}(\bx)]^* \} \Bigr),
\end{align*}
where $g^{\xandy TB}_{\bl,\bk\,} \equiv 2\pi{f^{\xandy TB}_{\bl, \bk}}[\widetilde C^{\xandy TT}_{\bk+\bl,\text{\tiny exp}} \widetilde C^{\xandy BB}_{\bk,\text{\tiny exp}} ]^{-1} $ and  $f^{\xandy TB}_{\bl, \bk}   \equiv  \frac{1}{2\pi} [\bl \cdot (\bk + \bl) C^{\xandy TE}_{\bk+\bl} ]  \sin(2\Delta\varphi)$.
 \end{claim}

\if\Ver\LongVer{ 
{\flushleft\textcolor{blue}{$\downarrow$---------begin long version---------}}\newline

\begin{proof}

For the estimator estimator $\hat\phi^{\xandy{T}{B}}$ notice that if we set $\Delta \varphi \equiv \varphi_{\bk+\bl}-\varphi_{\bk}$ then $f^{\xandy TB}_{\bl, \bk}   =  \frac{1}{2\pi} [\bl \cdot (\bk + \bl) C^{\xandy TE}_{\bk+\bl} ]  \sin(2\Delta\varphi)$. Since $\gw_{\bs l, \bs k} = 2\pi{f^\xy_{\bl,\bk}}[\widetilde C^\text{\tiny X\!X}_{\bk+\bl,\text{\tiny exp}} \widetilde C^\text{\tiny $Y\!Y$}_{\bk,\text{\tiny exp}} ]^{-1} $ we have that
\begin{align*}
 \int \frac{d^2\bs k}{2\pi} g^{\xandy{T}{B}}_{\bs l, \bs k}   X_{\bs k+\bs l}  Y^*_{\bs k} & =    \int \frac{d^2\bs k}{2\pi}   \Bigl[\bl \cdot (\bk + \bl) C^{\xandy TE}_{\bk+\bl}   \sin(2\Delta\varphi) \Bigr]  \frac{ X_{\bs k+\bs l}}{\widetilde C^\text{\tiny T\!T}_{\bk+\bl,\text{\tiny exp}}}  \frac{Y^*_{\bs k} }{\widetilde C^\text{\tiny $B\!B$}_{\bk,\text{\tiny exp}}} \\
 &=  -i l_a \int \frac{d^2\bs k}{2\pi}   \Bigl[i (k^a + l^a) C^{\xandy TE}_{\bk+\bl}   \sin(2\Delta\varphi) \Bigr] \frac{ X_{\bs k+\bs l}}{\widetilde C^\text{\tiny T\!T}_{\bk+\bl,\text{\tiny exp}}}  \frac{Y^*_{\bs k} }{\widetilde C^\text{\tiny $B\!B$}_{\bk,\text{\tiny exp}}}\\
  &=  -i l_a \int \frac{d^2\bs k}{2\pi}  \underbrace{ i (k^a + l^a) C^{\xandy TE}_{\bk+\bl}    \frac{ X_{\bs k+\bs l}}{\widetilde C^\text{\tiny T\!T}_{\bk+\bl,\text{\tiny exp}}}  }
  _{[\nabla^a X^{\xandy TE}_{\xandy TT}]_{\bk+\bl}}
  \underbrace{\frac{Y^*_{\bs k} }{\widetilde C^\text{\tiny $B\!B$}_{\bk,\text{\tiny exp}}}}
  _{[Y_{\xandy BB}]^*_\bk} 
  \sin(2\Delta\varphi)\\
 &=  -i  l_a   \int \frac{d^2\bx}{2\pi} e^{-i\bl\cdot\bx}  Im\{ [\nabla^a\hat  X^{\xandy TE}_{\xandy TT}(\bx)] [\hat Y_{\xandy BB}(\bx)]^*\}
 \end{align*}

Also notice that
\begin{align*}
\int \frac{d^2\bk}{(2\pi)^2}  
   |g^{\xandy TB}_\bl(\bk)|^2
  X_{\bk+\bl} Y^*_{\bk} &= \int \frac{d^2 \bk}{(2\pi)^2}  \frac{[\bl\cdot (\bk + \bl)C^{\xandy TE}_{\bk+\bl} ]^2 \sin^2(2\Delta\varphi)}{[\widetilde C^{\xandy TT}_{\bk+\bl,\text{\tiny exp}}]^2 [\widetilde C^{\xandy BB}_{\bk,\text{\tiny exp}}]^2}  X_{\bk+\bl} Y^*_{\bk}\\
&= - l_a l_b \int \frac{d^2 \bk}{(2\pi)^2}
 \underbrace{\frac{i(k^a + l^a)i(k^b + l^b)[C^{\xandy TE}_{\bk+\bl} ]^2X_{\bk+\bl}}{[\widetilde C^{\xandy TT}_{\bk+\bl,\text{\tiny exp}}]^2}}
 _{ [\nabla^a\nabla^b X^{\xandy TE, \xandy TE}_{\xandy TT, \xandy TT}]_{\bk+\bl}}
 \underbrace{\frac{ Y^*_\bk }{[\widetilde C^\text{\tiny BB}_{\bk,\text{\tiny exp}}]^2}}
 _{[Y_{\xandy BB, \xandy BB}]^*_\bk} 
 \sin^2(2\Delta\varphi) \\
 &= -l_a l_b \frac{1}{4\pi}\int \frac{d^2\bx}{2\pi} e^{-i\bl\cdot\bx}\left( [Y_{\xandy BB, \xandy BB}(\bx)] [\nabla^a\nabla^b X^{\xandy TE, \xandy TE}_{\xandy TT, \xandy TT}(\bx)]-Re\{ [\check Y_{\xandy BB, \xandy BB}(\bx)] [\nabla^a\nabla^b \check X^{\xandy TE, \xandy TE}_{\xandy TT, \xandy TT}(\bx)]^* \}\right) 
 \end{align*}
Also  notice that
\begin{align*}
\int \frac{d^2\bk}{(2\pi)^2}  
  g^{\xandy TB}_\bl(\bk)g^{\xandy TB}_\bl(-\bk-\bl)
  X_{\bk+\bl} Y^*_{\bk} &= \int \frac{d^2 \bk}{(2\pi)^2} 
   \frac{ \bl\cdot (\bk + \bl)\bl\cdot (-\bk) C^\text{\tiny TE}_{\bk+\bl} C^\text{\tiny TE}_{\bk} (-1) \sin^2(2\Delta\varphi)}{\widetilde C^\text{\tiny TT}_{\bk+\bl,\text{\tiny exp}} \widetilde C^\text{\tiny BB}_{\bk,\text{\tiny exp}}\widetilde C^\text{\tiny TT}_{\bk} \widetilde C^\text{\tiny BB}_{\bk+\bl,\text{\tiny exp}}  }  X_{\bk+\bl} Y^*_{\bk}\\
&=  l_a l_b \int \frac{d^2 \bk}{(2\pi)^2}
 \underbrace{\frac{i(k^a + l^a) C^\text{\tiny TE}_{\bk+\bl}  X_{\bk+\bl}}{\widetilde C^\text{\tiny TT}_{\bk+\bl,\text{\tiny exp}}  \widetilde C^\text{\tiny BB}_{\bk+\bl,\text{\tiny exp}}}}_{[\nabla^a X^{\xandy TE}_{\xandy TT, \xandy BB}]_{\bk+\bl}}
 \underbrace{\frac{-ik^b C^\text{\tiny TE}_{\bk}  Y^*_{\bk}}{\widetilde C^\text{\tiny TT}_{\bk,\text{\tiny exp}}  \widetilde C^\text{\tiny BB}_{\bk,\text{\tiny exp}}}}
 _{[\nabla^b Y^{\xandy TE}_{\xandy TT, \xandy BB}]_{\bk}^*}
 \sin^2(2\Delta\varphi) \\
&= l_a l_b \frac{1}{4\pi}\int \frac{d^2\bx}{2\pi} e^{-i\bl\cdot\bx}\left([\nabla^a X^{\xandy TE}_{\xandy TT, \xandy BB}(\bx)][\nabla^b Y^{\xandy TE}_{\xandy TT, \xandy BB}(\bx)]-Re\{[\nabla^a  \check X^{\xandy TE}_{\xandy TT, \xandy BB}(\bx)][\nabla^b\check Y^{\xandy TE}_{\xandy TT, \xandy BB}(\bx)]^* \}\right).
\end{align*}
QED.
 
\end{proof}
{\flushleft\textcolor{blue}{$\uparrow$------------end long version---------}}\newline
} \fi


\begin{claim}[EB estimator]
  If $X_\bk$ and $Y_\bk$ are complex functions such that $X_{-\bk}=X^*_\bk$ and $Y_{-\bk}=Y^*_\bk$ then 
\begin{align*}
 \int \frac{d^2\bs k}{2\pi} g^{\xandy{E}{B}}_{\bs l,\bs k\,}   X_{\bs k+\bs l}  Y^*_{\bs k}  &=  -i  l_a   \int \frac{d^2\bx}{2\pi} e^{-i\bl\cdot\bx} \\
 &\quad\times Im\{ [\nabla^a\hat  X^{\xandy EE}_{\xandy EE}(\bx)] [\hat Y_{\xandy BB}(\bx)]^*\};
 \end{align*}
 \begin{align*}
 \int \frac{d^2\bk}{(2\pi)^2}  
   |g^{\xandy EB}_{\bl,\bk\,}|^2
   &X_{\bk+\bl} Y^*_{\bk}
= -\frac{ l_a l_b}{4\pi} \int \frac{d^2 \bx}{2\pi}  e^{-i\bl\cdot\bx} \\
&\times\Bigl(  [Y_{\xandy BB, \xandy BB}(\bx)] [\nabla^a\nabla^b X^{\xandy EE, \xandy EE}_{\xandy EE, \xandy EE}(\bx)] \nonumber\\
&\quad -Re\{ [\check Y_{\xandy BB, \xandy BB}(\bx)] [\nabla^a\nabla^b \check X^{\xandy EE, \xandy EE}_{\xandy EE, \xandy EE}(\bx)]^* \} \Bigr);
\end{align*}
\begin{align*}
 \int \frac{d^2 \bk}{(2\pi)^2}  g^{\xandy EB}_{\bl,\bk\,} &  g^{\xandy  E B}_{\bl,-\bk-\bl\,}   X_{\bk+\bl} Y^*_{\bk} = -\frac{ l_a l_b}{4\pi} \int \frac{d^2 \bx}{2\pi}  e^{-i\bl\cdot\bx} \\
 &\times \Bigl(  [\nabla^a X^{\xandy EE}_{\xandy EE, \xandy BB}(\bx)][\nabla^b Y^{\xandy EE}_{\xandy EE, \xandy BB}(\bx)]\nonumber \\
 &\quad-Re\{[\nabla^a  \check X^{\xandy EE}_{\xandy EE, \xandy BB}(\bx)][\nabla^b\check Y^{\xandy EE}_{\xandy EE, \xandy BB}(\bx)]^* \} \Bigr),
\end{align*}
where $g^{\xandy EB}_{\bl,\bk\,} \equiv 2\pi{f^{\xandy EB}_{\bl, \bk}}[\widetilde C^{\xandy EE}_{\bk+\bl,\text{\tiny exp}} \widetilde C^{\xandy BB}_{\bk,\text{\tiny exp}} ]^{-1} $ and $f^{\xandy EB}_{\bl, \bk}   \equiv \frac{1}{2\pi} [\bl \cdot (\bk + \bl) C^{\xandy EE}_{\bk+\bl} ]  \sin(2\Delta\varphi)$.
 \end{claim}

\if\Ver\LongVer{ 
{\flushleft\textcolor{blue}{$\downarrow$---------begin long version---------}}\newline

\begin{proof}

For the estimator estimator $\hat\phi^{\xandy{E}{B}}$ notice that if we set $\Delta \varphi \equiv \varphi_{\bk+\bl}-\varphi_{\bk}$ then $f^{\xandy EB}_{\bl, \bk}   = \frac{1}{2\pi} [\bl \cdot (\bk + \bl) C^{\xandy EE}_{\bk+\bl} ]  \sin(2\Delta\varphi)$. Since $\gw_{\bs l, \bs k} = 2\pi{f^\xy_{\bl,\bk}}[\widetilde C^\text{\tiny X\!X}_{\bk+\bl,\text{\tiny exp}} \widetilde C^\text{\tiny $Y\!Y$}_{\bk,\text{\tiny exp}} ]^{-1} $ we have that
\begin{align*}
 \int \frac{d^2\bs k}{2\pi} g^{\xandy{E}{B}}_{\bs l, \bs k}   X_{\bs k+\bs l}  Y^*_{\bs k} & =    \int \frac{d^2\bs k}{2\pi}   \Bigl[\bl \cdot (\bk + \bl) C^{\xandy EE}_{\bk+\bl}   \sin(2\Delta\varphi) \Bigr]  \frac{ X_{\bs k+\bs l}}{\widetilde C^{\xandy  EE}_{\bk+\bl,\text{\tiny exp}}}  \frac{Y^*_{\bs k} }{\widetilde C^\text{\tiny $B\!B$}_{\bk,\text{\tiny exp}}} \\
 &=  -i l_a \int \frac{d^2\bs k}{2\pi}   \Bigl[i (k^a + l^a) C^{\xandy EE}_{\bk+\bl}   \sin(2\Delta\varphi) \Bigr] \frac{ X_{\bs k+\bs l}}{\widetilde C^{\xandy EE}_{\bk+\bl,\text{\tiny exp}}}  \frac{Y^*_{\bs k} }{\widetilde C^\text{\tiny $B\!B$}_{\bk,\text{\tiny exp}}}\\
  &=  -i l_a \int \frac{d^2\bs k}{2\pi}  \underbrace{ i (k^a + l^a) C^{\xandy EE}_{\bk+\bl}    \frac{ X_{\bs k+\bs l}}{\widetilde C^{\xandy EE}_{\bk+\bl,\text{\tiny exp}}}  }
  _{[\nabla^a X^{\xandy EE}_{\xandy EE}]_{\bk+\bl}}
  \underbrace{\frac{Y^*_{\bs k} }{\widetilde C^\text{\tiny $B\!B$}_{\bk,\text{\tiny exp}}}}
  _{[Y_{\xandy BB}]^*_\bk} 
  \sin(2\Delta\varphi)\\
 &=  -i  l_a   \int \frac{d^2\bx}{2\pi} e^{-i\bl\cdot\bx}  Im\{ [\nabla^a\hat  X^{\xandy EE}_{\xandy EE}(\bx)] [\hat Y_{\xandy BB}(\bx)]^*\}
 \end{align*}

Also notice that
\begin{align*}
\int \frac{d^2\bk}{(2\pi)^2}  
   |g^{\xandy EB}_\bl(\bk)|^2
  X_{\bk+\bl} Y^*_{\bk} &= \int \frac{d^2 \bk}{(2\pi)^2}  \frac{[\bl\cdot (\bk + \bl)C^{\xandy EE}_{\bk+\bl} ]^2 \sin^2(2\Delta\varphi)}{[\widetilde C^{\xandy EE}_{\bk+\bl,\text{\tiny exp}}]^2 [\widetilde C^{\xandy BB}_{\bk,\text{\tiny exp}}]^2}  X_{\bk+\bl} Y^*_{\bk}\\
&= - l_a l_b \int \frac{d^2 \bk}{(2\pi)^2}
 \underbrace{\frac{i(k^a + l^a)i(k^b + l^b)[C^{\xandy 
 EE}_{\bk+\bl} ]^2X_{\bk+\bl}}{[\widetilde C^{\xandy EE}_{\bk+\bl,\text{\tiny exp}}]^2}}
 _{ [\nabla^a\nabla^b X^{\xandy EE, \xandy EE}_{\xandy EE, \xandy EE}]_{\bk+\bl}}
 \underbrace{\frac{ Y^*_\bk }{[\widetilde C^\text{\tiny BB}_{\bk,\text{\tiny exp}}]^2}}
 _{[Y_{\xandy BB, \xandy BB}]^*_\bk} 
 \sin^2(2\Delta\varphi) \\
 &= -l_a l_b \frac{1}{4\pi}\int \frac{d^2\bx}{2\pi} e^{-i\bl\cdot\bx}\left( [Y_{\xandy BB, \xandy BB}(\bx)] [\nabla^a\nabla^b X^{\xandy EE, \xandy EE}_{\xandy EE, \xandy EE}(\bx)]-Re\{ [\check Y_{\xandy BB, \xandy BB}(\bx)] [\nabla^a\nabla^b \check X^{\xandy EE, \xandy EE}_{\xandy EE, \xandy EE}(\bx)]^* \}\right) 
 \end{align*}
Also  notice that
\begin{align*}
\int \frac{d^2\bk}{(2\pi)^2}  
  g^{\xandy EB}_\bl(\bk)g^{\xandy EB}_\bl(-\bk-\bl)
  X_{\bk+\bl} Y^*_{\bk} &= \int \frac{d^2 \bk}{(2\pi)^2} 
   \frac{ \bl\cdot (\bk + \bl)\bl\cdot (-\bk) C^\text{\tiny EE}_{\bk+\bl} C^\text{\tiny EE}_{\bk} (-1) \sin^2(2\Delta\varphi)}{\widetilde C^\text{\tiny EE}_{\bk+\bl,\text{\tiny exp}} \widetilde C^\text{\tiny BB}_{\bk,\text{\tiny exp}}\widetilde C^\text{\tiny EE}_{\bk,\text{\tiny exp}} \widetilde C^\text{\tiny BB}_{\bk+\bl,\text{\tiny exp}}  }  X_{\bk+\bl} Y^*_{\bk}\\
&=  l_a l_b \int \frac{d^2 \bk}{(2\pi)^2}
 \underbrace{\frac{i(k^a + l^a) C^\text{\tiny EE}_{\bk+\bl}  X_{\bk+\bl}}{\widetilde C^{\xandy EE}_{\bk+\bl,\text{\tiny exp}}  \widetilde C^\text{\tiny BB}_{\bk+\bl,\text{\tiny exp}}}}_{[\nabla^a X^{\xandy EE}_{\xandy EE, \xandy BB}]_{\bk+\bl}}
 \underbrace{\frac{-ik^b C^\text{\tiny EE}_{\bk}  Y^*_{\bk}}{\widetilde C^\text{\tiny EE}_{\bk,\text{\tiny exp}}  \widetilde C^\text{\tiny BB}_{\bk,\text{\tiny exp}}}}
 _{[\nabla^b Y^{\xandy EE}_{\xandy EE, \xandy BB}]_{\bk}^*}
 \sin^2(2\Delta\varphi) \\
&= l_a l_b \frac{1}{4\pi}\int \frac{d^2\bx}{2\pi} e^{-i\bl\cdot\bx}\left([\nabla^a X^{\xandy EE}_{\xandy EE, \xandy BB}(\bx)][\nabla^b Y^{\xandy EE}_{\xandy EE, \xandy BB}(\bx)]-Re\{[\nabla^a  \check X^{\xandy EE}_{\xandy EE, \xandy BB}(\bx)][\nabla^b\check Y^{\xandy EE}_{\xandy EE, \xandy BB}(\bx)]^* \}\right).
\end{align*}
QED.
 
\end{proof}
{\flushleft\textcolor{blue}{$\uparrow$------------end long version---------}}\newline
} \fi


\begin{claim}[EE estimator]
  If $X_\bk$ and $Y_\bk$ are complex functions such that $X_{-\bk}=X^*_\bk$ and $Y_{-\bk}=Y^*_\bk$ then  
\begin{align*}
 \int \frac{d^2\bs k}{2\pi} g^{\xandy{E}{E}}_{\bs l,\bs k\,}   X_{\bs k+\bs l}  Y^*_{\bs k}  &=  -i  l_a   \int \frac{d^2\bx}{2\pi} e^{-i\bl\cdot\bx} \\
 &\times\Bigl(  Re\{ [\nabla^a\hat  X^{\xandy EE}_{\xandy EE}(\bx)] [\hat Y_{\xandy EE}(\bx)]^*\}  \nonumber\\
 &\quad +  Re\{ [\nabla^a\hat  Y^{\xandy EE}_{\xandy EE}(\bx)] [\hat X_{\xandy EE}(\bx)]^*\} 
 \Bigr);
  \end{align*}
 \begin{align*}
 \int \frac{d^2\bk}{(2\pi)^2}  
   |g^{\xandy EE}_{\bl,\bk\,}|^2
  & X_{\bk+\bl} Y^*_{\bk}
= -\frac{ l_a l_b}{4\pi} \int \frac{d^2 \bx}{2\pi}  e^{-i\bl\cdot\bx} \\
 &\times\Bigl( 
  [\nabla^a\nabla ^b  X^{\xandy EE,\xandy EE}_{\xandy EE, \xandy EE}(\bx)]  [ Y^{}_{\xandy EE, \xandy EE}(\bx)]\nonumber  \\
 &\quad+Re\{ [\nabla^a\nabla ^b \check X^{\xandy EE,\xandy EE}_{\xandy EE, \xandy EE}(\bx)]  [\check Y^{}_{\xandy EE, \xandy EE}(\bx)]^*  \}  \nonumber\\
 &\quad+  [X^{}_{\xandy EE, \xandy EE}(\bx)]  [\nabla^a\nabla ^b  Y^{\xandy EE,\xandy EE}_{\xandy EE, \xandy EE}(\bx)] \nonumber \\
 &\quad+Re\{ [\check X^{}_{\xandy EE, \xandy EE}(\bx)]  [\nabla^a\nabla ^b  \check Y^{\xandy EE,\xandy EE}_{\xandy EE, \xandy EE}(\bx)]^*  \} \nonumber \\
 &\quad+  2 [\nabla^a  X^{\xandy EE}_{\xandy EE, \xandy EE}(\bx)][\nabla^b  Y^{\xandy EE}_{\xandy EE, \xandy EE}(\bx)]  \nonumber\\
 &\quad+Re\{  2 [\nabla^a  \check X^{\xandy EE}_{\xandy EE, \xandy EE}(\bx)][\nabla^b  \check Y^{\xandy EE}_{\xandy EE, \xandy EE}(\bx)]^*  \}
 \Bigr);
\end{align*}
\begin{align*}
 \int \frac{d^2 \bk}{(2\pi)^2}  g^{\xandy EE}_{\bl,\bk\,}   g^{\xandy  EE}_{\bl,-\bk-\bl\,}   X_{\bk+\bl} Y^*_{\bk} =  \int \frac{d^2\bk}{(2\pi)^2}  
   |g^{\xandy EE}_{\bl,\bk\,}|^2
   X_{\bk+\bl} Y^*_{\bk},
\end{align*}
where $g^{\xandy EE}_{\bl,\bk\,} \equiv 2\pi{f^{\xandy EE}_{\bl, \bk}}[\widetilde C^{\xandy EE}_{\bk+\bl,\text{\tiny exp}} \widetilde C^{\xandy EE}_{\bk,\text{\tiny exp}} ]^{-1} $ and $f^{\xandy EE}_{\bl, \bk}   \equiv \frac{1}{2\pi} [\bl \cdot (\bk + \bl) C^\text{\tiny $EE$}_{\bk+\bl}  -\bl \cdot \bk  C^\text{\tiny $EE$}_{\bk}] \cos(2\Delta\varphi) $.
 \end{claim}

\if\Ver\LongVer{ 
{\flushleft\textcolor{blue}{$\downarrow$---------begin long version---------}}\newline

\begin{proof}

For the estimator estimator $\hat\phi^{\xandy{E}{E}}$ notice that if we set $\Delta \varphi \equiv \varphi_{\bk+\bl}-\varphi_{\bk}$ then $f^{\xandy EE}_{\bl, \bk}   =\frac{1}{2\pi} [\bl \cdot (\bk + \bl) C^\text{\tiny $EE$}_{\bk+\bl}  -\bl \cdot \bk  C^\text{\tiny $EE$}_{\bk}] \cos(2\Delta\varphi) $. Since $\gw_{\bs l, \bs k} = 2\pi{f^\xy_{\bl,\bk}}[\widetilde C^\text{\tiny X\!X}_{\bk+\bl,\text{\tiny exp}} \widetilde C^\text{\tiny $Y\!Y$}_{\bk,\text{\tiny exp}} ]^{-1} $ we have that
\begin{align*}
 \int \frac{d^2\bs k}{2\pi} g^{\xandy{E}{E}}_{\bs l, \bs k}   X_{\bs k+\bs l}  Y^*_{\bs k} & =    \int \frac{d^2\bs k}{2\pi} [\bl \cdot (\bk + \bl) C^\text{\tiny $EE$}_{\bk+\bl}  -\bl \cdot \bk  C^\text{\tiny $EE$}_{\bk}] \cos(2\Delta\varphi)  \frac{ X_{\bs k+\bs l}}{\widetilde C^{\xandy  EE}_{\bk+\bl,\text{\tiny exp}}}  \frac{Y^*_{\bs k} }{\widetilde C^{\xandy EE}_{\bk,\text{\tiny exp}}} \\
 &=  -i l_a \int \frac{d^2\bs k}{2\pi}   \Bigl[i (k^a + l^a) C^{\xandy EE}_{\bk+\bl}   \cos(2\Delta\varphi) \Bigr] \frac{ X_{\bs k+\bs l}}{\widetilde C^{\xandy EE}_{\bk+\bl,\text{\tiny exp}}}  \frac{Y^*_{\bs k} }{\widetilde C^{\xandy EE}_{\bk,\text{\tiny exp}}}
  +i l_a \int \frac{d^2\bs k}{2\pi}   \Bigl[i k^a  C^{\xandy EE}_{\bk}   \cos(2\Delta\varphi) \Bigr] \frac{ X_{\bs k+\bs l}}{\widetilde C^{\xandy EE}_{\bk+\bl,\text{\tiny exp}}}  \frac{Y^*_{\bs k} }{\widetilde C^{\xandy EE}_{\bk,\text{\tiny exp}}}
 \\
  &=  -i l_a \int \frac{d^2\bs k}{2\pi}  \underbrace{ i (k^a + l^a) C^{\xandy EE}_{\bk+\bl}    \frac{ X_{\bs k+\bs l}}{\widetilde C^{\xandy EE}_{\bk+\bl,\text{\tiny exp}}}  }
  _{[\nabla^a X^{\xandy EE}_{\xandy EE}]_{\bk+\bl}}
  \underbrace{\frac{Y^*_{\bs k} }{\widetilde C^{\xandy EE}_{\bk,\text{\tiny exp}}}}
  _{[Y_{\xandy EE}]^*_\bk} 
  \cos(2\Delta\varphi) 
  - i l_a \int \frac{d^2\bs k}{2\pi}  \underbrace{   \frac{ X_{\bs k+\bs l}}{\widetilde C^{\xandy EE}_{\bk+\bl,\text{\tiny exp}}}  }
  _{[ X^{}_{\xandy EE}]_{\bk+\bl}}
  \underbrace{\frac{ (i k^a C^{\xandy EE}_{\bk} Y_\bk)^* }{\widetilde C^{\xandy EE}_{\bk,\text{\tiny exp}}}}
  _{[\nabla^aY^{\xandy EE}_{\xandy EE}]^*_\bk} 
  \cos(2\Delta\varphi)\\
 &=  -i  l_a   \int \frac{d^2\bx}{2\pi} e^{-i\bl\cdot\bx} \Bigl( Re\{ [\nabla^a\hat  X^{\xandy EE}_{\xandy EE}(\bx)] [\hat Y_{\xandy EE}(\bx)]^*\} +  Re\{ [\nabla^a\hat  Y^{\xandy EE}_{\xandy EE}(\bx)] [\tilde X_{\xandy EE}(\bx)]^*\} \Bigr)
 \end{align*}

To finish notice that
\begin{align*}
\int \frac{d^2\bk}{(2\pi)^2}  
   |g^{\xandy EE}_{\bl,\bk}|^2
  X_{\bk+\bl} Y^*_{\bk}
 &=  \int \frac{d^2 \bk}{(2\pi)^2}  \frac{ \bigl[\bl\cdot(\bk + \bl)C^{\xandy EE}_{\bk+\bl} - \bl\cdot\bk C^{\xandy EE}_\bk  \bigr]^2 \cos^2(2\Delta\varphi) }{[\widetilde C^{\xandy EE}_{\bk+\bl,\text{\tiny exp}}]^2 [\widetilde C^{\xandy EE}_{\bk,\text{\tiny exp}}]^2}{X_{\bk+\bl}}{Y^*_\bk}\\
&= I + II + III
\end{align*}
where
\begin{align*}
I&= - l_a l_b \int \frac{d^2 \bk}{(2\pi)^2}  
\underbrace{\frac{i(k^a+l^a) i(k^b+l^b) [C^{\xandy EE}_{\bk+\bl}]^2 X_{\bk +\bl}}{[\widetilde C^{\xandy EE}_{\bk+\bl,\text{\tiny exp}}]^2 } }_{ [\nabla^a\nabla ^b  X^{\xandy EE,\xandy EE}_{\xandy EE, \xandy EE}]_{\bk+\bl}}
\underbrace{\frac{Y^*_\bk}{ [\widetilde C^{\xandy EE}_{\bk,\text{\tiny exp}}]^2}}_{ [ Y^{}_{\xandy EE, \xandy EE}]^*_\bk} \cos^2(2\Delta\varphi) \\
&=  - l_a l_b\frac{1}{4\pi} \int \frac{d^2 \bx}{2\pi}  e^{-i\bl\cdot\bx} \Bigl(  [\nabla^a\nabla ^b  X^{\xandy EE,\xandy EE}_{\xandy EE, \xandy EE}(\bx)]  [ Y^{}_{\xandy EE, \xandy EE}(\bx)]  +Re\{ [\nabla^a\nabla ^b \check X^{\xandy EE,\xandy EE}_{\xandy EE, \xandy EE}(\bx)]  [\check Y^{}_{\xandy EE, \xandy EE}(\bx)]^*  \}\Bigr)
\end{align*}
and
\begin{align*}
II&=-l_a l_b \int \frac{d^2 \bk}{(2\pi)^2}
\underbrace{\frac{X_{\bk+\bl}}{[\widetilde C^{\xandy EE}_{\bk+\bl,\text{\tiny exp}}]^2}}_{  [X^{}_{\xandy EE, \xandy EE}]_{\bk+\bl}}  
\underbrace{\frac{  ik^a ik^b \bigl[ C^{\xandy EE}_\bk  \bigr]^2Y^*_\bk}{ [\widetilde C^{\xandy EE}_{\bk,\text{\tiny exp}}]^2 }}_{  [\nabla^a\nabla ^b  Y^{\xandy EE,\xandy EE}_{\xandy EE, \xandy EE}]^*_{\bk}}  \cos^2(2\Delta\varphi) \\
&=  - l_a l_b\frac{1}{4\pi} \int \frac{d^2 \bx}{2\pi}  e^{-i\bl\cdot\bx} \Bigl(  [X^{}_{\xandy EE, \xandy EE}(\bx)]  [\nabla^a\nabla ^b  Y^{\xandy EE,\xandy EE}_{\xandy EE, \xandy EE}(\bx)] +Re\{ [\check X^{}_{\xandy EE, \xandy EE}(\bx)]  [\nabla^a\nabla ^b  \check Y^{\xandy EE,\xandy EE}_{\xandy EE, \xandy EE}(\bx)]^*  \}\Bigr)
\end{align*}
and
\begin{align*}
III&=- 2 l_a l_b \int \frac{d^2 \bk}{(2\pi)^2} 
\underbrace{\frac{ i(k^a + l^a)C^{\xandy EE}_{\bk+\bl} X_{\bk+\bl}}{[\widetilde C^{\xandy EE}_{\bk+\bl,\text{\tiny exp}}]^2}}_{  [\nabla^a  X^{\xandy EE}_{\xandy EE, \xandy EE}]_{\bk+\bl}}
\underbrace{\frac{ - ik^b C^{\xandy EE}_\bk Y^*_\bk }{ [\widetilde C^{\xandy EE}_{\bk,\text{\tiny exp}}]^2}} _{  [\nabla^b  Y^{\xandy EE}_{\xandy EE, \xandy EE}]^*_{\bk}} \cos^2(2\Delta\varphi) \\
&=  - l_a l_b\frac{1}{4\pi} \int \frac{d^2 \bx}{2\pi}  e^{-i\bl\cdot\bx} \Bigl(  2 [\nabla^a  X^{\xandy EE}_{\xandy EE, \xandy EE}(\bx)][\nabla^b  Y^{\xandy EE}_{\xandy EE, \xandy EE}(\bx)]  +Re\{  2 [\nabla^a  \check X^{\xandy EE}_{\xandy EE, \xandy EE}(\bx)][\nabla^b  \check Y^{\xandy EE}_{\xandy EE, \xandy EE}(\bx)]^*  \}\Bigr)
\end{align*}

\end{proof}

{\flushleft\textcolor{blue}{$\uparrow$------------end long version---------}}\newline
} \fi

\if\Ver\LongVer{ 
{\flushleft\textcolor{blue}{$\downarrow$---------begin long version---------}}\newline


\begin{lemma}
\label{LEM1}
For any functions $A_\bl$ and $B_\bl$
\begin{equation}
\int \frac{d^2\bs k}{2\pi}  A_{\bs k +\bs l} B^*_\bk= \int \frac{d^2\bx}{2\pi} e^{-i\bl\cdot\bx} A(\bx) B^*(\bx) 
\end{equation}
\end{lemma}
\begin{proof}
{
\footnotesize
\begin{align*}
\int \frac{d^2\bs k}{2\pi}  A_{\bs k +\bs l} B^*_\bk &= \int \frac{d^2\bs k}{2\pi}   B^*_\bk \int \frac{d^2\bx}{2\pi} e^{-i (\bs k +\bs l)\cdot \bx} A(\bx) \\
&= \int \frac{d^2\bx}{2\pi}e^{-i\bs l\cdot \bx }   A(\bx)\left[ \int \frac{d^2\bs k}{2\pi}    e^{i \bs k \cdot \bx} B_{\bk} \right]^*\\
&= \int \frac{d^2\bx}{2\pi}e^{-i\bs l\cdot \bx }   A(\bx) B^*(\bx) 
\end{align*}
QED.}
\end{proof}

\begin{lemma}
\label{LEM2}
For any functions $A_{-\bl}^* = A_\bl$ and $B_{-\bl}^* = B_\bl$ then
\begin{equation}
\int \frac{d^2\bs k}{2\pi}  A_{\bs k +\bs l} B^*_\bk \cos(2\varphi_{\bk + \bl}- 2\varphi_{\bk})= \int \frac{d^2\bx}{2\pi} e^{-i\bl\cdot\bx} Re[\hat A(\bx) \hat B^*(\bx) ]
\end{equation}
where $\hat A_\bk = A_\bk e^{i2\varphi_\bk}$,  $\hat B_\bk = B_\bk e^{i2\varphi_\bk}$ and $\varphi_\bk$ is the phase (or angle) fo $\bk$
\end{lemma}
\begin{proof} 
{
\footnotesize
First notice that 
\begin{align*}
\int \frac{d^2\bs k}{2\pi}  A_{\bs k +\bs l} B^*_\bk \cos(2\varphi_{\bk + \bl}- 2\varphi_{\bk}) &= 
\int \frac{d^2\bs k}{2\pi}  A_{\bs k +\bs l} B^*_\bk Re\left[ e^{i2\varphi_{\bk + \bl}- i2\varphi_{\bk}} \right]\\
& = 
\frac{1}{2}\int \frac{d^2\bs k}{2\pi}  A_{\bs k +\bs l} B^*_\bk e^{i2\varphi_{\bk + \bl}- i2\varphi_{\bk}} 
 + \frac{1}{2}\int \frac{d^2\bs k}{2\pi}  A_{\bs k +\bs l} B^*_\bk e^{-i2\varphi_{\bk + \bl}+ i2\varphi_{\bk}}  \\
 & = 
\frac{1}{2}\int \frac{d^2\bs k}{2\pi}  \hat A_{\bs k +\bs l} \hat B^*_\bk 
 + \frac{1}{2}\int \frac{d^2\bs k}{2\pi}  \tilde A_{\bs k +\bs l} \tilde B^*_\bk \\
 & = 
 \frac{1}{2}\int \frac{d^2\bx}{2\pi} e^{-i\bl\cdot\bx} \hat A(\bx) \hat B^*(\bx) 
+  \frac{1}{2}\int \frac{d^2\bx}{2\pi} e^{-i\bl\cdot\bx} \tilde A(\bx) \tilde B^*(\bx) 
\end{align*}
where $\tilde A_\bk = A_\bk e^{-i2\varphi_\bk} $ and $\tilde B_{\bk} = B_\bk e^{-i2\varphi_\bk}$. Notice also that $e^{i2\varphi_{-\bk}}= e^{i2(\varphi_{\bk}+\pi)}= e^{i2\varphi_{\bk}}$. Therefore
\begin{align*}
 \tilde A(\bx) &= \int \frac{d^2\bs k}{2\pi}  e^{i\bk \cdot \bx}  A_\bk e^{-i2\varphi_\bk} = \int \frac{d^2\bs k}{2\pi}  e^{i\bk \cdot \bx}  A^*_{-\bk} e^{-i2\varphi_{-\bk}} \\
   &=\left[ \int \frac{d^2\bs k}{2\pi}  e^{-i\bk \cdot \bx}  A_{-\bk} e^{i2\varphi_{-\bk}}\right]^*=\left[ \int \frac{d^2\bs k}{2\pi}  e^{i\bk \cdot \bx}  A_{\bk} e^{i2\varphi_{\bk}}\right]^*=\hat A^*(\bs x)
\end{align*}
and $ \tilde B(\bx) =  \hat B^*(\bx)$ similarly. Therefore 
\begin{align*}
\int \frac{d^2\bs k}{2\pi}  A_{\bs k +\bs l} B^*_\bk \cos(2\varphi_{\bk + \bl}- 2\varphi_{\bk}) &=  \int \frac{d^2\bx}{2\pi} e^{-i\bl\cdot\bx} \left[\frac{1}{2}\hat A(\bx) \hat B^*(\bx) 
+  \frac{1}{2}\hat A^*(\bx) \hat B(\bx) \right] \\
&=  \int \frac{d^2\bx}{2\pi} e^{-i\bl\cdot\bx} Re \left[\hat A(\bx) \hat B^*(\bx) 
 \right] 
\end{align*}
QED.}
\end{proof}

\begin{lemma}
\label{LEM3}
For any functions $A_{-\bl}^* = A_\bl$ and $B_{-\bl}^* = B_\bl$ then
\begin{equation}
\int \frac{d^2\bs k}{2\pi}  A_{\bs k +\bs l} B^*_\bk \sin(2\varphi_{\bk + \bl}- 2\varphi_{\bk})= \int \frac{d^2\bx}{2\pi} e^{-i\bl\cdot\bx} Im[\hat A(\bx) \hat B^*(\bx) ]
\end{equation}
where $\hat A_\bk = A_\bk e^{i2\varphi_\bk}$,  $\hat B_\bk =  B_\bk e^{i2\varphi_\bk}$ and $\varphi_\bk$ is the phase (or angle) of the vector $\bk$.
\end{lemma}
\begin{proof}
{
\footnotesize
First notice that 
\begin{align*}
\int \frac{d^2\bx}{2\pi} e^{-i\bl\cdot\bx} \hat A(\bx) \hat B^*(\bx) 
&=\int \frac{d^2\bs k}{2\pi}  \hat A_{\bs k +\bs l} \hat B^*_\bk, \qquad\text{by Lemma \ref{LEM1}}  \\
&=\int \frac{d^2\bs k}{2\pi}  A_{\bs k +\bs l} B^*_\bk e^{i(2\varphi_{\bk + \bl}- 2\varphi_{\bk})} \\
&= \int \frac{d^2\bs k}{2\pi}  A_{\bs k +\bs l} B^*_\bk \cos(2\varphi_{\bk + \bl}- 2\varphi_{\bk})+i\int \frac{d^2\bs k}{2\pi}  A_{\bs k +\bs l} B^*_\bk \sin(2\varphi_{\bk + \bl}- 2\varphi_{\bk})\\
&=
 \int \frac{d^2\bx}{2\pi} e^{-i\bl\cdot\bx} Re[\hat A(\bx) \hat B^*(\bx) ] +i\int \frac{d^2\bs k}{2\pi}  A_{\bs k +\bs l} B^*_\bk \sin(2\varphi_{\bk + \bl}- 2\varphi_{\bk}), \qquad\text{by Lemma \ref{LEM2}}.
\end{align*}
QED.}
\end{proof}

\begin{lemma}
\label{LEM11}
For any functions $A_{-\bl}^* = A_\bl$ and $B_{-\bl}^* = B_\bl$ then
\begin{align}
\int \frac{d^2\bs k}{2\pi}  A_{\bs k +\bs l} B^*_\bk \sin^2(2\varphi_{\bk + \bl}- 2\varphi_{\bk})&= \frac{1}{2}\int \frac{d^2\bx}{2\pi} e^{-i\bl\cdot\bx}\left( A(\bx)B(\bx)-Re[ {\check A}(\bx) {\check B}^*(\bx) ]\right)\\
\int \frac{d^2\bs k}{2\pi}  A_{\bs k +\bs l} B^*_\bk \cos^2(2\varphi_{\bk + \bl}- 2\varphi_{\bk})&= \frac{1}{2}\int \frac{d^2\bx}{2\pi} e^{-i\bl\cdot\bx}\left( A(\bx)B(\bx)+Re[{\check A}(\bx) {\check B}^*(\bx) ]\right)
\end{align}
where $\check A_\bk = A_\bk e^{i4\varphi_\bk} $, $\check B_{\bk} = B_\bk e^{i4\varphi_\bk}$ and $\varphi_\bk$ is the phase (or angle) of the vector $\bk$.
\end{lemma}
\begin{proof}
First notice that 
\begin{align*}
\int \frac{d^2\bs k}{2\pi}  A_{\bs k +\bs l} B^*_\bk \cos^2(2\varphi_{\bk + \bl}- 2\varphi_{\bk}) &= 
\int \frac{d^2\bs k}{2\pi}  A_{\bs k +\bs l} B^*_\bk \left(Re\left[ e^{i2\varphi_{\bk + \bl}- i2\varphi_{\bk}} \right]\right)^2\\
& = 
\int \frac{d^2\bs k}{2\pi}  A_{\bs k +\bs l} B^*_\bk \left( \frac{1}{2}  e^{i2\varphi_{\bk + \bl}- i2\varphi_{\bk}} 
 + \frac{1}{2} e^{-i2\varphi_{\bk + \bl}+ i2\varphi_{\bk}} \right)^2 \\
&= \int \frac{d^2\bs k}{2\pi}  A_{\bs k +\bs l} B^*_\bk \left( \frac{1}{4}  e^{i4\varphi_{\bk + \bl}- i4\varphi_{\bk}} 
 +\frac{1}{4} e^{-i4\varphi_{\bk + \bl}+ i4\varphi_{\bk}} +\frac{1}{2}\right) \\
 & = 
\frac{1}{4}\int \frac{d^2\bs k}{2\pi}  \bar A_{\bs k +\bs l} \bar B^*_\bk 
+  \frac{1}{4}\int \frac{d^2\bs k}{2\pi}  \tilde A_{\bs k +\bs l} \tilde B^*_\bk + \frac{1}{2}\int \frac{d^2\bs k}{2\pi}   A_{\bs k +\bs l}  B^*_\bk  \\
 & = 
 \frac{1}{2}\int \frac{d^2\bx}{2\pi} e^{-i\bl\cdot\bx} Re\left[\bar A(\bx) \bar B^*(\bx) \right]
 + \frac{1}{2}\int \frac{d^2\bx}{2\pi} e^{-i\bl\cdot\bx} A(\bx) B(\bx) 
\end{align*}
where $\bar A_\bk = A_\bk e^{i4\varphi_\bk} $ and $\bar B_{\bk} = B_\bk e^{i4\varphi_\bk}$ and  $\tilde A_\bk = A_\bk e^{-i4\varphi_\bk} $ and $\tilde B_{\bk} = B_\bk e^{-i4\varphi_\bk}$. This follows since $ \bar A(\bx) =\tilde A^*(\bx)$ and $ \bar B(\bx) =\tilde B^*(\bx)$.
QED.
\end{proof}


 %
{\em Remark:} The continuous noise spectrum for the discrete additive noise (with standard deviation $\sigma$ $\mu K$) is given by $C^{NN}_{\bs \ell}= \sigma^2 \Delta \bs x$ where $\Delta\bs x$ denotes the pixel area in position space.

{\flushleft\textcolor{blue}{$\uparrow$------------end long version---------}}\newline
} \fi

\begin{acknowledgments}
We gratefully acknowledge helpful discussions with D. Hanson, L. Knox and A. van Engelen.
\end{acknowledgments}

\bibliography{ref}

\end{document}